\documentclass[12pt]{article}

\usepackage[dvipsnames]{xcolor}
\usepackage{amsmath,amsfonts,amsthm,amssymb,graphicx,arydshln,umoline,appendix,newtxtext,newtxmath,subcaption,physics,color,setspace,enumerate,paralist,comment,mathtools,tikz,bm}

\usepackage{natbib}
\setcitestyle{authoryear,open={(},close={)}}

\usepackage[left=1.15in,right=1.15in,top=1.15in,bottom=1.15in]{geometry}
\usepackage{hyperref}
\hypersetup{colorlinks=true, linkcolor=red, citecolor=blue}

\newtheorem{lemma}{Lemma}
\newtheorem*{lemma*}{Lemma}

\newtheorem{proposition}{Proposition}
\newtheorem{theorem}[proposition]{Theorem}
\newtheorem{corollary}[proposition]{Corollary}

\newtheorem*{claim*}{Claim}

\theoremstyle{definition}

\newtheorem*{basicassumption}{Basic Assumption}
\newtheorem{example}{Example}
\newtheorem{definition}{Definition}

\theoremstyle{remark}
\newtheorem{remark}{Remark}

\newcommand{\N}{\mathbb{N}}
\newcommand{\R}{\mathbb{R}}

\newcommand{\1}{\mathbf{1}}

\newcommand{\calM}{\mathcal{M}}

\newcommand{\calS}{\mathcal{S}}

\newcommand{\calU}{\mathcal{U}}

\newcommand{\argmax}{\mathop{\rm arg~max}\limits}

\DeclareMathOperator{\E}{E}
\DeclareMathOperator{\Var}{Var}

\DeclareMathOperator{\Cov}{Cov}

\DeclareMathOperator{\Aut}{Aut}

\DeclareMathOperator{\diag}{diag}
\DeclareMathOperator{\Diag}{Diag}

\begin{document}

\title{\bf LQG Information Design\thanks{We are grateful to seminar and conference participants at CPW-CTWE at Waseda University, HK Junior Micro Theory Workshop, International Conference on Game Theory at Stony Brook, ITAM, JEA Autumn Meeting, Kobe University, Kyoto University, National University of Singapore, SING, the University of Tokyo, and VIEE at Concordia University for their valuable comments and feedback. This work is supported by Grant-in-Aid for Scientific Research Grant Numbers 15K03348 and 18H05217.}}

\author{Masaki Miyashita\footnote{The University of Hong Kong, \tt{masaki11@hku.hk}} \and Takashi Ui\footnote{Kanagawa University and Hitotsubashi Institute for Advanced Study, \tt{oui@econ.hit-u.ac.jp}}}

\date{August 2025}
\maketitle

\begin{abstract}
This paper addresses information design in a workhorse model of network games, where agents have linear best responses, the information designer maximizes a quadratic objective, and the payoff-relevant state follows a multivariate Gaussian distribution.
We formulate the problem as a semidefinite program and establish strong duality to characterize the optimal information structure.
A necessary and sufficient condition for optimality is given by a simple linear relationship between the induced equilibrium strategy profile and the state.
Leveraging this characterization, we show that the state is fully revealed in an aggregative form for welfare maximization, while individual agents may remain only partially informed.
When agent roles are interchangeable, the optimal information structure inherits the same degree of symmetry, which facilitates computation.
In such cases, we show that the optimal amount of information revealed to each agent is closely linked to the network's chromatic number.


\end{abstract}

\thispagestyle{empty}



\newpage
\setcounter{page}{1}
\section{Introduction}

Information design is the problem of choosing an information structure in a game to induce equilibrium behavior that maximizes the designer's objective.
\cite{bergemannmorris2013} propose a two-step approach to this problem.
The first step identifies the set of joint distributions over actions and states that are implementable by some information structure.
The second step selects the distribution within this set that maximizes the designer's objective.
Their key insight lies in the first step: the set of implementable distributions coincides with the set of \emph{Bayes correlated equilibria} (BCE), which prescribe incentive-compatible action recommendations that each agent is willing to follow.

We advance this approach by developing the second, optimization step.
To this end, we focus on environments in which both the agents' payoffs and the designer's objective are quadratic in actions and states, and the state vector is normally distributed.
We refer to this class of problems as \emph{LQG information design}, where LQG stands for linear best responses, a quadratic objective, and Gaussian uncertainty.\footnote{The acronym ``LQG'' originates from control theory, where LQG control concerns linear systems driven by Gaussian noise, with the objective of determining an optimal output feedback law that minimizes a quadratic cost.}

Our analysis pursues two objectives.
First, we provide a unified framework of information design that accommodates a workhorse model of network games \citep{ballesteretal2006, bramoulleetal2014}.
We apply our results to analyze optimal information provision in networks, which serves as the leading example in this paper.

Second, we offer structural insights into information design with multiple agents.
Our analysis highlights two salient properties of information structures.
The first, \emph{noise-freeness}, refers to the case where agents' equilibrium actions are deterministic conditional on the state realization.
This property appears in \citet{bhm2015} as characterizing the information structure that maximizes aggregate volatility in symmetric LQG games.
The second property, \emph{state-identifiability}, is the structural dual of noise-freeness; it holds when the state realization is uniquely identified from the profile of equilibrium actions.
Economically, this property describes a situation where the state is fully revealed in an aggregative form, while each individual agent may remain only partially informed.

We provide a complete characterization of the optimal BCE.
We first show that any information structure is outcome-equivalent to some Gaussian information structure, in which the state and action vectors are jointly normally distributed.
This implies that, without loss of generality, the designer can restrict attention to Gaussian information structures, even though non-Gaussian alternatives are a priori feasible.
Our main theorem then shows that a Gaussian information structure is optimal if and only if, in the corresponding BCE, the strategy profile $\bm{\sigma}=(\sigma_1,\ldots,\sigma_n)$ and the state vector $\bm{\theta}=(\theta_1,\ldots,\theta_m)$ satisfy the following linear constraint:
\begin{equation}
A_\Lambda \qty(\bm{\sigma} - \E\qty[\bm{\sigma}]) = B_\Lambda \qty(\bm{\theta} - \E\qty[\bm{\theta}]),	
\label{key main eq}
\end{equation}
where $A_\Lambda$ is an $n\times n$ matrix and $B_\Lambda$ is an $n\times m$ matrix. 
These matrices are determined explicitly by the primitives of the game and by the Lagrange multipliers $\Lambda$ associated with the agents' incentive constraints.

The optimality condition \eqref{key main eq} provides key insights into the structural properties discussed earlier.
Specifically, if $A_\Lambda$ is invertible, the equilibrium strategy profile $\bm{\sigma}$ can be expressed as a function of the state $\bm{\theta}$, implying that the optimal information structure is noise-free.
Conversely, if $B_\Lambda$ has full column rank so that it admits a left inverse, $\bm{\theta}$ can be recovered as a function of $\bm{\sigma}$, and the information structure is state-identifiable.

The matrices $A_\Lambda$ and $B_\Lambda$ are determined by the Lagrange multipliers $\Lambda$, which themselves endogenously arise in the dual formulation of the designer’s problem.
This motivates us to identify primitive conditions under which these matrices have full rank.
To this end, we consider a benevolent designer who maximizes social welfare.
In this case, $B_\Lambda$ turns out to have full rank, so the optimal information structure must be state-identifiable.
The intuition is that when the designer’s objective aligns with the agents’ payoffs, providing more information enables their actions to better match the realized state, thereby increasing welfare.
However, even for welfare maximization, fully informing each individual agent may be suboptimal since providing the same signal can cause coordination frictions.

By contrast, whether the optimal information structure is noise-free depends on the dimensionality of the state.
When the state is high-dimensional, one can construct linearly independent statistics that are only partially informative in isolation but jointly identify the state, in which case the optimal information structure is both state-identifiable and noise-free.
When agents instead respond to a one-dimensional state, we find that the optimal state-identifiable information structure is noisy, with each agent receiving a private signal distorted by idiosyncratic noise.

Mathematically, we show that LQG information design is formulated as a convex optimization problem known as \emph{semidefinite programming} (SDP).\footnote{SDP is a class of convex optimization problems that generalizes linear programming by replacing nonnegative vector variables with semidefinite matrix variables. See, e.g., \citet{vandenbergheboyd1996} for further details.}
This formulation builds on two key properties of BCE established by \citet{bergemannmorris2013}.
First, the mean equilibrium actions are invariant across all BCE.
This property allows the expected value of any quadratic objective function to be expressed as a linear function of the action-state covariance matrix.
Second, a given positive semidefinite matrix represents the action-state covariance of some BCE if and only if it satisfies a system of linear constraints derived from agents’ incentive compatibility.
These properties imply that the designer’s problem reduces to maximizing a linear function of a positive semidefinite matrix subject to linear constraints characterizing the set of BCE, mirroring the defining feature of SDP.
We show that strong duality holds in our problem, which in turn yields the characterization of optimality via \eqref{key main eq}.

We illustrate our findings through an application to network games.
In Section~\ref{sec_ex}, we introduce a variant of the network intervention problem studied in \citet{galeottietal2020} and show how it can be formulated as SDP.
Leveraging the computational advantage of SDP, we conduct numerical simulations that reveal several traits of optimal information structures, including state-identifiability and noise-freeness.
We also find that when agents play interchangeable roles in the underlying network, the optimal information structure inherits this symmetry.
These observations motivate our general analysis, whose implications are formally established in Section~\ref{sec_network}, where we return to this application.

The application also illustrates how strategic interaction shapes the optimal provision of information for multiple agents.
When the game exhibits strategic complements, the benevolent designer fully reveals the state to each individual agent.
By contrast, with fairly strong strategic substitutes, the optimal information structure keeps agents partially informed while fully revealing the state in an aggregative form.
Intuitively, this arises because while more information helps align actions with the realized state, it also increases correlation across agents’ actions, which lowers welfare under strategic substitutes.
In that case, we further show that the optimal accuracy of private information is closely linked to the network’s \emph{chromatic number}, defined as the minimum number of colors needed to color agents so that connected agents receive different colors.

\medskip

{\it Related literature.}
LQG models, originated in \cite{radner1962}, have been widely used to study the role of endogenous information, including information sharing \citep{vives1984,clarke1983,galor1985}, information acquisition \citep{lietal1987,vives1988,hellwigveldkamp2009,myattwallace2012,colomboetal2014}, public information disclosure \citep{morrisshin2002,angeletospavan2007}, and organizational and political communication \citep{dessein2006,calvoarmengol2015,dewanmyatt2008}.
Our work contributes to this tradition by developing an LQG model for information design and Bayesian persuasion.\footnote{See \citet{bergemannmorris2019} and \citet{kamenica2019} for surveys on information design and Bayesian persuasion. Our study builds on an outcome-based approach to information design, as in \citet{bergemannmorris2013, bergemannmorris2016a, bergemannmorris2016b, bergemannmorris2019}, \citet{bhm2015, bhm2017}, and \citet{taneva2019}.
\citet{mathevetetal2020} develop a different belief-based approach that extends the concavification method of \citet{kamenicagentzkow2011} to a multiple-agent setting.}

We apply the duality principle to information design.
A closely related paper is \citet{smolinyamashita2023}, who study information design in a general setting with concave utility functions.
They formulate the problem as an infinite-dimensional linear program over probability measures on actions and states and establish weak duality.
In contrast, we focus on a more specialized setting with quadratic payoffs, which allows us to reduce the problem to SDP over the action-state covariance matrix.
This reduction enables us to establish strong duality, which yields a sharper characterization of optimality.

This paper also relates to the literature on network interventions, an active area of research across multiple disciplines.\footnote{See, for example, \cite{valente2012} for a general introduction to the subject.}
In economics, the seminal work by \cite{ballesteretal2006} formulate the intervention problem as a network game and quantify the structural importance of each agent by comparing equilibrium with and without the agent.
\cite{galeottietal2020} consider interventions that alter targeted agents’ base action levels.
Within our model, such interventions are captured as exogenous changes in payoff structures.
In contrast, we focus on altering information structures to influence equilibrium behavior, thereby providing complementary insights for the designer seeking to shape desired economic outcomes.\footnote{This form of informational intervention has gathered attention in recent empirical literature. For example, \citet{banerjeeetal2013} examine the effect of information transmission on microfinance participation using demographic and social network data.}

\medskip

{\it Organization.}
Section~\ref{sec_ex} offers numerical simulations in the leading example.
Section~\ref{sec_model} introduces the model.
Section~\ref{sec_main} presents the main results.
Section~\ref{sec_nfsi} analyzes state-identifiability and noise-freeness in detail.
Section~\ref{sec_personal} provides a special case in which either property holds.
Section~\ref{sec_network} revisits the leading example and applies the main results.
Section~\ref{sec_conc} concludes.


\section{Leading example}
\label{sec_ex}

We consider a variant of the network games studied in \cite{galeottietal2020}.
Let $[n] \equiv \{1, \ldots, n\}$ be a finite set of agents, and let $G \equiv [g_{ij}]_{n \times n}$ be an adjacency matrix with $g_{ii} = 0$ for all $i \in [n]$.
The matrix $G$, referred to as a \emph{network}, captures the strategic connections among agents.
In this section, we assume these connections are unweighted and undirected, meaning $g_{ij} = g_{ji} \in \{0,1\}$ for all $i \neq j$, although our general analysis does not rely on these assumptions.
The \emph{degree} of agent $i$ is the number of her neighbors, defined as $d_i(G) \equiv \sum_{j \neq i} g_{ij}$.
When all agents have the same degree, we write it simply as $d(G)$.

\begin{figure}
\begin{tabular}{ccc}

\begin{minipage}{.33\textwidth}
\centering
\begin{tikzpicture}[node distance=1in]
\node [circle, draw, thick] (1) {1};
\node [circle, draw, thick, right of = 1] (2) {2};
\node [circle, draw, thick, below of = 2] (3) {3};
\node [circle, draw, thick, left of = 3] (4) {4};
\draw[<->, >=stealth, thick] (1) -- (2);
\draw[<->, >=stealth, thick] (1) -- (3);
\draw[<->, >=stealth, thick] (1) -- (4);
\draw[<->, >=stealth, thick] (2) -- (3);
\draw[<->, >=stealth, thick] (2) -- (4);
\draw[<->, >=stealth, thick] (3) -- (4);
\end{tikzpicture}
\end{minipage}

\begin{minipage}{.33\textwidth}
\centering
\begin{tikzpicture}[node distance=1in]
\node [circle, draw, thick] (1) {1};
\node [circle, draw, thick, right of = 1] (2) {2};
\node [circle, draw, thick, below of = 2] (3) {3};
\node [circle, draw, thick, left of = 3] (4) {4};
\draw[<->, >=stealth, thick] (1) -- (2);
\draw[<->, >=stealth, thick] (1) -- (4);
\draw[<->, >=stealth, thick] (2) -- (3);
\draw[<->, >=stealth, thick] (3) -- (4);
\end{tikzpicture}
\end{minipage}

\begin{minipage}{.33\textwidth}
\centering
\begin{tikzpicture}[node distance=0.5in]
\node [circle, draw, thick] (1) {1};
\node [circle, draw, thick, above of = 1, yshift = 0.15in] (2) {2};
\node [circle, draw, thick, below right of = 1, xshift = 0.2in] (3) {3};
\node [circle, draw, thick, below left of = 1, xshift = -0.2in] (4) {4};
\draw[<->, >=stealth, thick] (1) -- (2);
\draw[<->, >=stealth, thick] (1) -- (3);
\draw[<->, >=stealth, thick] (1) -- (4);
\end{tikzpicture}
\end{minipage}

\end{tabular}
\caption{In the complete network (left), all agents are connected with one another.
In the cyclical network (middle), agents are arranged on a circle and receive strategic influences from adjacent agents.
In the star-shaped network (right), agent 1 is the central agent connected with all peripheral agents.
}
\label{fig_sym_graph}
\end{figure}

Each agent $i$ simultaneously chooses an investment level $a_i \in \R$ at cost $a_i^2/2$.
The private marginal return on investment depends on an unknown personal state $\theta_i$ and the localized aggregate of other agents' investment levels $\bm{a}_{-i} \equiv (a_j)_{j\neq i}$ as follows:
\begin{equation*}
\phi_i(\bm{a}_{-i},\theta_i) \equiv \theta_i + \beta \sum_{j \neq i} g_{ij} a_j.
\end{equation*}
Here, $\beta \in \R$ represents the degree of strategic externality.
The case of $\beta < 0$ captures \emph{strategic substitutes}, where agents have incentives to be differentiated from their neighbors \citep{vives1999, goyal2001}.
In contrast, $\beta > 0$ reflects \emph{strategic complements} \citep{ballesteretal2006}, where each agent's marginal return is enhanced when her neighbors choose higher investment levels.
As common in the literature, assume that $|\beta|$ is less than the inverse of the spectral radius of $G$.

For an action profile $\bm{a} \equiv (a_i, \bm{a}_{-i})$, agent $i$'s payoff is given by
\begin{equation*}
u_i(\bm{a},\theta_i) \equiv \phi_i(\bm{a}_{-i},\theta_i) \cdot a_i - \frac{a_i^2}{2}.
\end{equation*}
This payoff structure bears similarity to the one in \cite{galeottietal2020}, but a crucial difference is that each agent's base-level marginal return $\theta_i$ is a random variable.
Assume that $\bm{\theta} = (\theta_1,\ldots,\theta_n)$ follows a symmetric multivariate normal distribution with $\E\qty[\bm{\theta}] = \bm{0}$, $\Var\qty[\theta_i] = 1$, and $\Cov\qty[\theta_i, \theta_j] = \rho$.
Here, $\rho \in [0,1]$ represents the correlation coefficient between the payoff-relevant states of different agents.
The case with $\rho = 1$ is referred to as the \emph{common-value case}; in this case, $\theta_i$ and $\theta_j$ are perfectly correlated with each other.
The other cases with $\rho < 1$ are referred to as \emph{private-value cases}.

Each agent chooses her investment level without knowing the realization of $\bm{\theta}$.
Instead, an \emph{(information) designer} provides private signals $\eta_i$ to each agent $i$ according to a pre-announced joint distribution $\pi$ over $(\bm{\eta}, \bm{\theta}) \equiv (\eta_1, \ldots, \eta_n, \theta_1, \ldots, \theta_n)$, referred to as an \emph{information structure}.
Given $\pi$, agent $i$'s strategy $\sigma_i$ maps her signal $\eta_i$ to an action $\sigma_i(\eta_i) \in \R$.
A strategy profile $\bm{\sigma}^\pi \equiv (\sigma_1^\pi,\ldots,\sigma_n^\pi)$ forms a \emph{Bayes–Nash equilibrium} (BNE) under $\pi$ if each $\sigma_i^\pi$ maximizes agent $i$'s expected payoff conditional on $\eta_i$, given opponents’ strategies $\bm{\sigma}^\pi_{-i} \equiv (\sigma_j^\pi)_{j \neq i}$.

In this section, we assume that the designer seeks to maximize \emph{utilitarian welfare}, i.e., the sum of all agents’ payoffs.
Notably, since both the agents’ payoffs and the designer’s objective are quadratic functions, the problem can be expressed in terms of the action-state covariance matrix induced by $\pi$.
Let $x_{ij} = \Cov[\sigma^\pi_i, \sigma^\pi_j]$ denote the covariance between agents $i$ and $j$’s actions.
Also, let $y_{ij} = \Cov[\sigma^\pi_i, \theta_j]$ denote the covariance between agent $i$’s action and agent $j$'s payoff state.
The designer can influence agents' incentives by sending private signals that may be correlated with each other and with the state.

Similarly to \cite{uiyoshizawa2015}, we can show that the designer’s expected objective in BNE is proportional to the total variance of agents' equilibrium strategies:
\begin{equation} \label{obj_ex}
\sum_{i \in [n]} x_{ii}.
\end{equation}
Moreover, as shown by \citet{bergemannmorris2013, bergemannmorris2016a}, we can, without loss of generality, identify each agent’s signal as a recommendation for the action she should take.
Such action recommendations are feasible only if each agent has an incentive to follow it.
This incentive constraint can be characterized by the following \emph{obedience condition}:\footnote{This extends the obedience condition of \citet{bergemannmorris2013} to asymmetric settings.}
\begin{equation} \label{const_ex1}
x_{ii} = y_{ii} + \beta \sum_{j \neq i} g_{ij} x_{ij}, \quad \forall i \in [n]. \end{equation}
Additionally, since $X =[x_{ij}]_{n \times n}$ and $Y = [y_{ij}]_{n \times n}$ are submatrices of the action-state covariance matrix, they must satisfy the following positive semidefiniteness constraint:
\begin{equation} \label{const_ex2}
\mqty[X & Y \\ Y^\top & \Var\qty[\bm{\theta}]] \text{ is positive semidefinite}.
\end{equation}
In our Lemma~\ref{lem_moment}, we formally show that conditions \eqref{const_ex1} and \eqref{const_ex2} are necessary and sufficient for any candidate action-state covariance matrix to be inducible in BNE under some (Gaussian) information structure.

Thus, the designer's problem reduces to maximizing the linear objective \eqref{obj_ex} subject to the linear constraints \eqref{const_ex1} and the positive semidefiniteness constraint \eqref{const_ex2}.
In fact, these are the defining mathematical features of SDP, which admits several optimization tools that are readily available.
Below, we present numerical results obtained from simulations using CVXPY, a Python-based modeling language for convex optimization.

First, we find that when a network exhibits symmetry between agents, the optimal information structure preserves the same degree of symmetry.
Specifically, in the complete network (Figure~\ref{fig_sym_graph}, left), all agents are structurally identical.
In the cyclical network (middle), although agents have different sets of neighbors, there is an intuitive sense in which they remain symmetric: for any pair of distinct agents $(i,j)$, there exists a permutation that maps $i$ to $j$ while leaving the network structure unchanged.\footnote{We formalize this notion of symmetry in Section~\ref{sec_sym} by introducing some group-theoretic definitions.}
In the star-shaped network (right), all peripheral agents share this sort of symmetry, whereas the central agent does not.
In each case, the numerically obtained optimal information structure yields an action-state covariance matrix that remains invariant under relabeling of symmetric agents' indices.

\begin{figure}[t]
\begin{center}
\includegraphics[width=\linewidth]{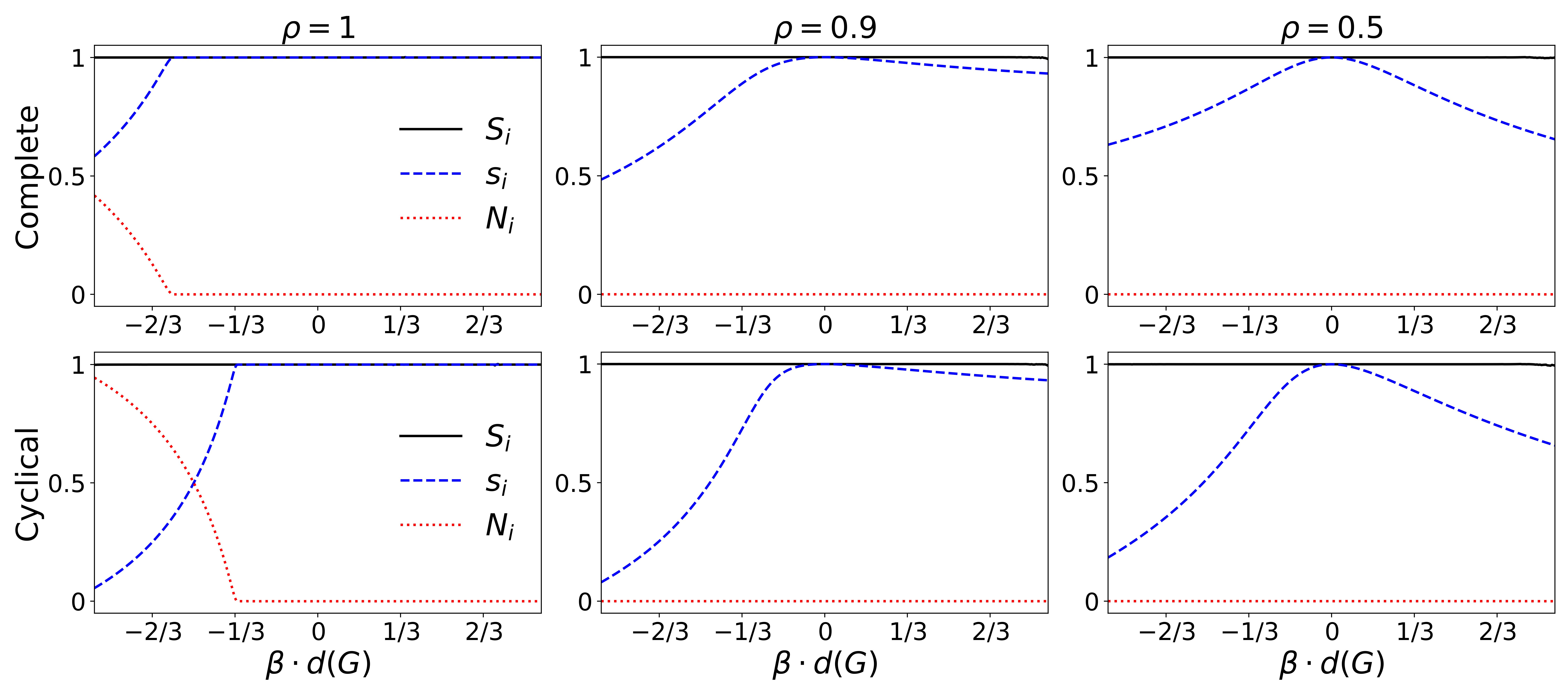}
\caption{With $n=4$, the values of $S_i$, $s_i$, and $N_i$ are plotted against $\beta d(G)$ for different $\rho$ under the complete and cyclical graphs.
These plots do not depend on the choice of $i$.}
\label{fig_comp_cycle}
\end{center}
\end{figure}

Second, any optimal information structure fully reveals the entire state to the agents in aggregate, while each individual agent may remain only partially informed.
We formalize this idea using variance reduction as a measure of signal informativeness.
Let
\begin{equation*}
s_i \equiv \frac{\Var \qty[\theta_i] - \Var[\theta_i \mid \sigma^\pi_i]}{\Var\qty[\theta_i]}
\end{equation*}
denote the proportion of variance in agent $i$'s personal state $\theta_i$ eliminated by conditioning on agent $i$'s action recommendation.
Moreover, let
\begin{equation*}
S_i \equiv \frac{\Var \qty[\theta_i] - \Var[\theta_i \mid \sigma^\pi_1,\ldots,\sigma^\pi_n]}{\Var \qty[\theta_i]}
\end{equation*}
denote the proportion of variance in $\theta_i$ eliminated by conditioning on the entire profile of all agents' action recommendations.

In Figure~\ref{fig_comp_cycle}, we display $s_i$ and $S_i$ as functions of $\beta$ for various values of $\rho$ under both complete and cyclical networks.
We observe that $S_i = 1$ across all parameter values, indicating that $\theta_i$ is always fully revealed in aggregate.
However, each individual agent may remain only partially informed since $s_i < 1$ when $\beta$ is negative and below some cutoff in the common-value case.
Moreover, $s_i < 1$ whenever $\beta \neq 0$ in all private-value cases.
This finding highlights the importance of privatized information, as such aggregative full disclosure requires providing personalized signals to different agents rather than revealing public information.

Third, whether the optimal information structure induces non-fundamental volatility depends on the distribution of personal states.
To quantify this result, we compute
\begin{equation*}
N_i \equiv \frac{\Var \qty[\sigma^\pi_i \mid \theta_1,\ldots,\theta_n]}{\Var \qty[\sigma^\pi_i]},
\end{equation*}
which measures the proportion of variance in agent $i$’s equilibrium strategy that can be explained by the state.
Note that $N_i$ ranges from $0$ to $1$, with lower values indicating less non-fundamental volatility in $\sigma^\pi_i$ unrelated to $\bm{\theta}$.

As shown in Figure~\ref{fig_comp_cycle}, in the common-value case, $N_i > 0$ holds whenever the common state is not fully disclosed to agent $i$.
In contrast, we find $N_i=0$ in all private-value cases, indicating that the equilibrium strategy entails no non-fundamental volatility.
This contrast suggests that the designer’s approach to achieving aggregative full disclosure depends on whether agents' payoff-relevant states are common or private.
When the state is common, each agent's signal incorporates idiosyncratic noise that cancels out in the aggregate but introduces non-fundamental variation at the individual level.
By contrast, in the private-value case, each agent’s action recommendation is deterministic in the state without extraneous randomness; instead, the optimal information structure coordinates agents’ actions by shaping signals that embed not only each agent’s own state but also those of others.

\begin{figure}[t]
\begin{center}
\includegraphics[width=\linewidth]{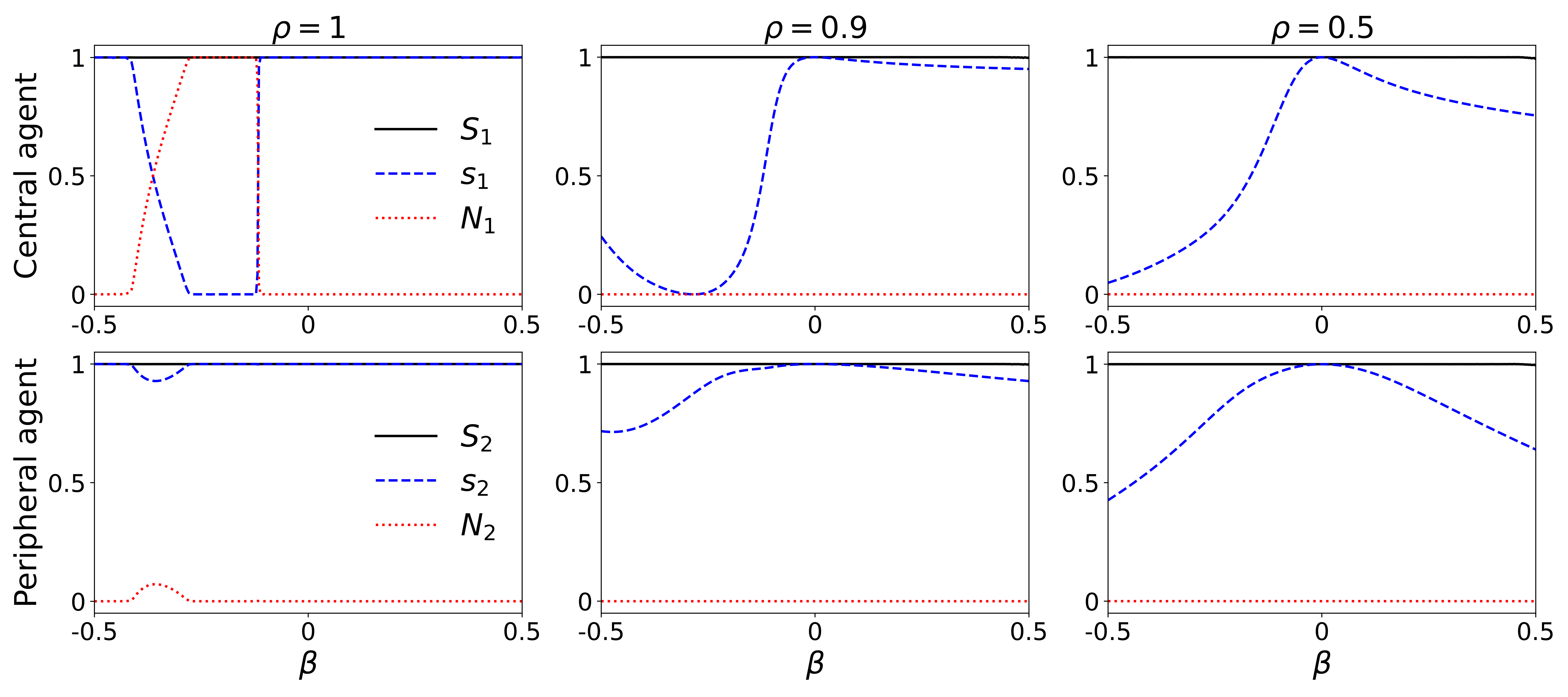}
\caption{For the undirected star-shaped network with $n=4$ (Figure~\ref{fig_sym_graph}, right), the values of $S_i$, $s_i$, and $N_i$ are plotted against $\beta$ for the central agent ($i=1$) in the first row, and for a representative peripheral agent ($i=2$) in the second row.}
\label{fig_star}
\end{center}
\end{figure}

These findings extend to networks that lack full symmetry.
Figure~\ref{fig_star} reports numerical results for the undirected star-shaped network, where the first row corresponds to the central agent and the second to a representative peripheral agent.
Even though the network does not treat all agents symmetrically, we find $S_i = 1$ for all agents and all parameter values, confirming the optimality of aggregative full disclosure.
Also, $N_i = 0$ in all private-value cases, confirming the absence of non-fundamental volatility.

Lastly, in the common-value case, we discuss the interplay between the structure of strategic interactions and the optimal degree of each agent's signal informativeness.
In both complete and cyclical networks, the left panels of Figure~\ref{fig_comp_cycle} show that the optimality of full disclosure depends monotonically on the total externality each agent receives (or equivalently, gives): there exists a threshold such that full disclosure is optimal whenever $\beta d(G)$ exceeds it.
Since this threshold is negative, full disclosure is always optimal under strategic complements.
However, under strategic substitutes, the strength of externality must remain small enough for full disclosure to be optimal.

This threshold varies across networks, and the condition for full disclosure to be optimal is weaker in complete networks than in cyclical networks.
Moreover, even when full disclosure is suboptimal, agents receive more informative signals in complete networks than in cyclical networks.
This raises a natural question: what property of the network determines how much information the optimal information structure reveals to agents?
A key determinant turns out to be the ``chromatic number'' from graph coloring problems.
Among nontrivial networks, complete and cyclical networks have the largest and smallest chromatic numbers, which leads to stark differences in the informativeness of optimal signals.

Interestingly, however, when agent roles are not fully interchangeable, the optimality of full disclosure is no longer characterized by a monotonic parameter region.
In particular, for star-shaped networks, the left panel of Figure~\ref{fig_star} shows that full disclosure is optimal when $\beta$ is either large or small enough, but not when $\beta$ is at an intermediate level.


\section{Model}
\label{sec_model}

The following matrix notations are used. 
For a positive integer $n \in \N$, let $[n] \equiv \{1,\ldots,n\}$.
Given $n,m \in \N$, we denote by $\calM^{n,m}$ the set of $n \times m$ matrices, $\calS^n$ the set of $n \times n$ symmetric matrices, $\calS^n_+$ the set of all positive semidefinite matrices, and $\calS^n_{\rm diag}$ the set of $n \times n$ diagonal matrices.
Denote by $A^\top$ the transpose of $A$.
We write $A \succeq B$ (resp.\ $A \succ B$) when $A-B$ is positive semidefinite (resp.\ positive definite).
For two matrices of the same size, $A = [a_{ij}]_{n \times m}$ and $B = [b_{ij}]_{n \times m}$, their Frobenius inner product is defined by $A \bullet B = \sum_{i=1}^n \sum_{j=1}^m a_{ij} b_{ij}$.
For a square matrix $A = [a_{ij}]_{n \times n}$, let $\diag (A) = (a_{11},\ldots,a_{nn})$ denote the vector of the diagonal entries of $A$, and let $\tr(A) = \sum_{i=1}^n a_{ii}$ denote the trace.
Note that $A \bullet B = \tr(A^\top B)$ holds.
For a vector $\bm{x} \in \R^n$, let $\Diag(\bm{x})$ be the diagonal matrix with diagonal entries $\bm{x} = (x_1,\ldots,x_n)$.
Unless otherwise stated, all vectors are treated as column vectors.


\subsection{Setup}

We generalize the model from the leading example in Section~\ref{sec_ex} so that each agent's utility depends on a possibly weighted aggregate of their neighbors' actions and a weighted aggregate of multi-dimensional states.

Each agent $i \in [n]$ chooses an action $a_i \in \mathbb{R}$.
Their payoffs depend quadratically on the action profile $\bm{a} = (a_1,\ldots,a_n)$ and an $m$-dimensional payoff-relevant state $\bm{\theta} = (\theta_1,\ldots,\theta_m)$, where $m \geq 2$, as specified by:
\begin{equation} \label{payoff}
u_i(\bm{a},\bm{\theta}) \equiv
\underbrace{\qty( \sum_{k \in [m]} r_{ik} \theta_{k} - \sum_{j \in [n] \setminus \{i\}} q_{ij} a_j)}_{\text{private marginal return}} \ \cdot \ a_i
\ - \  \underbrace{\qty(\frac{q_{ii}}{2}) \cdot a_i^2}_{\text{private cost}}
\ + \ \underbrace{\bar{u}_i(\bm{a}_{-i}, \bm{\theta})}_{\substack{\text{pure} \\ \text{externality}}}.
\end{equation}
Here, $q_{ij}$ and $r_{ik}$ are constants, and $\bar{u}_i$ is an arbitrary function that depends only on the opponents' actions $\bm{a}_{-i}$ and the state $\bm{\theta}$.
Let $Q \equiv [q_{ij}]_{n\times n}$ and $R \equiv [r_{ik}]_{n\times m}$ be the matrices collecting the coefficients in \eqref{payoff}.
The matrix $Q$ encodes the strategic interactions among agents, while $R$ describes how the fundamental state influences individual payoffs.
The function $\bar{u}_i$ captures additional externalities that are independent of $a_i$.
Throughout, we impose the following regularity assumption.

\begin{basicassumption} \label{asm_basic}
The matrix $Q+Q^\top$ is positive definite.
The state $\bm{\theta}$ follows a multivariate Gaussian distribution with mean $\bar{\bm{\theta}} \equiv \E \qty[\bm{\theta}] \in \R^m$ and variance $Z \equiv \Var \qty[\bm{\theta}] \in \calS^m_+$, where $\Var[\theta_k] > 0$ for all $k \in [m]$.
\end{basicassumption}

The positive definiteness of $Q+Q^\top$ is satisfied in many applications and guarantees the unique existence of an equilibrium across all information structures.
That $Z \neq O$ is a minimal assumption to make our problem non-trivial.
In the payoff function \eqref{payoff}, $\bar{u}_i$ has no influence on agent's decisions so that their strategic concerns can be summarized by the matrices $Q$ and $R$ alone.
Also, since $\bm{\theta}$ is normally distributed, its distribution can be summarized by $(\bar{\bm{\theta}}, Z)$, while $\bar{\bm{\theta}}$ will be of no importance in our analysis.
In this regard, the basic game structure can be summarized by $(Q,R,Z)$.

The designer is assumed to hold an objective function $v$ that is also quadratic in $(\bm{a},\bm{\theta})$.
This arises naturally, for instance, when the designer seeks to maximize social welfare, defined as $v = \sum_{i=1}^n u_i$.
Given that each $\bar{u}_i$ is quadratic, the social welfare function remains quadratic.
More generally, we consider the following parametric form for the designer's objective:
\begin{equation} \label{objective}
v(\bm{a},\bm{\theta}) \equiv \mqty[\bm{a} \\ \bm{\theta}]^\top \mathbf{V} \mqty[\bm{a} \\ \bm{\theta}] + \bar{v}(\bm{a},\bm{\theta}), \quad {\rm where} \quad \mathbf{V} \equiv \mqty[V & W/2 \\ W^\top/2 & O] \in \calS^{n+m}.
\end{equation}
Here, $\mathbf{V}$ is a block matrix representing the quadratic coefficients in the designer's objective, and $\bar{v}$ is any linear function of $(\bm{a},\bm{\theta})$.
It is without loss of generality to assume that $\mathbf{V}$ is symmetric, as $\mathbf{V}$ can be replaced by $(\mathbf{V}+\mathbf{V}^\top)/2$ without changing $v$.
In addition, the bottom-right block can be set to zero since we can normalize any constant terms if necessary.

The designer sends a private signal $\eta_i$ to each agent $i$ according to a pre-announced joint probability distribution over $(\bm{\eta},\bm{\theta}) = (\eta_1,\ldots,\eta_n,\theta_1,\ldots,\theta_m)$.
This distribution $\pi$ is called an \emph{information structure}.
In particular, when $(\bm{\eta},\bm{\theta})$ is jointly normally distributed, $\pi$ is called a \emph{Gaussian information structure}.
Let $\Pi$ denote the set of all information structures, and let $\Pi^{\rm g} \subseteq \Pi$ be its subset consisting of all Gaussian ones.
The designer can commit to an information structure following the standard timeline in the literature:
In the first stage, the designer selects an information structure $\pi \in \Pi$ and informs all agents of $\pi$.
In the second stage, the realization of $(\bm{\eta}, \bm{\theta})$ is drawn according to $\pi$, and each agent $i$ chooses her optimal action to maximize the conditional expected payoff given $\eta_i$ under $\pi$.


\subsection{Bayes Nash equilibrium}
\label{sec_BNE}

To analyze the optimal information structure, we first characterize the equilibrium configuration in the second-stage game.
Given an information structure $\pi$, agent $i$'s strategy is a measurable function $\sigma_i$ of $\eta_i$, satisfying $\E |\sigma_i(\eta_i)|^2 < \infty$.
A strategy profile $\bm{\sigma} = (\sigma_1,\ldots,\sigma_n)$ forms a \emph{Bayes Nash equilibrium} (BNE) under $\pi$ if each agent's strategy maximizes her conditional expected payoff, given the strategies of others.
Formally, $\bm{\sigma}$ is a BNE if, almost surely (a.s.),
\begin{equation*}
\sigma_i(\eta_i)\in \argmax_{a_i \in \R}
\E^\pi \qty[u_i \qty((a_i, \bm{\sigma}_{-i}), \bm{\theta}) \mid \eta_i], \quad \forall i \in [n],
\end{equation*}
where $\bm{\sigma}_{-i} = (\sigma_{j}(\eta_j))_{j\neq i}$, and $\E^\pi[\cdot \mid \eta_i]$ denotes the conditional expectation given $\eta_i$ under $\pi$.
We simply write $\E\qty[\cdot]$ instead of $\E^\pi\qty[\cdot]$ as long as there is no risk of confusion.

It is known that a unique BNE exists for any information structure under Basic Assumption; see Proposition~4 of \cite{ui2016}.
Thus, we denote by $\bm{\sigma}^\pi$ the unique BNE corresponding to each information structure $\pi$.

\begin{lemma}
\label{lem_BNE}
If $Q+Q^\top$ is positive definite, then a unique BNE $\bm{\sigma}^\pi$ exists for each $\pi \in \Pi$.
In particular, if $\pi \in \Pi^{\rm g}$, then $(\bm{\sigma}^\pi, \bm{\theta})$ is jointly normally distributed.
\end{lemma}

Since $u_i$ is quadratic and concave, $\bm{\sigma}^\pi$ forms a BNE if and only if the following best-response condition holds for all $i \in [n]$:
\begin{equation}
\sum_{j \in [n]} q_{ij}\E[\sigma_j^\pi \mid \eta_i]= \sum_{k \in [m]} r_{ik}\E[\theta_k \mid \eta_i], \label{FOC}
\end{equation}
where $\sigma_j = \sigma_j(\eta_j)$.
This condition imposes linear constraints on the conditional expectations of the state and equilibrium actions.

Moreover, applying the law of iterated expectations, we obtain the corresponding restrictions on unconditional expectations and covariances.
Specifically, our Lemma~\ref{lem_moment} below shows that the unique equilibrium $\bm{\sigma}^\pi$ under any (possibly, non-Gaussian) information structure $\pi$ must satisfy the following moment restrictions:
\begin{gather}
\sum_{j \in [n]} q_{ij}\E[\sigma_j^\pi]= \sum_{k \in [m]} r_{ik}\E[\theta_k], \quad \forall i \in [n], \label{moment1} \\
\sum_{j \in [n]}  q_{ij}\Cov[\sigma_i^\pi, \sigma_j^\pi]=\sum_{k \in [m]} r_{ik} \Cov[\sigma_i^\pi, \theta_{k}], \quad \forall i \in [n].
\label{moment2}
\end{gather}

In matrix form, the first equation \eqref{moment1} can be rewritten as $Q \E \qty[\bm{\sigma}^\pi] = R \bar{\bm{\theta}}$.
Thus, since $Q$ is invertible under Basic Assumption, the expected action profile is uniquely determined as
\[
\bar{\bm{a}} \equiv \E \qty[{\bm{\sigma}}^\pi] = Q^{-1}R \bar{\bm{\theta}}.
\]
This shows that $\bar{\bm{a}}$ remains constant across all information structures, implying that the designer cannot manipulate expected action levels through information design.

Turning to the second equation \eqref{moment2}, it becomes $\diag(Q\Var \qty[\bm{\sigma}^\pi]) = \diag \qty(R \Cov \qty[\bm{\theta}, \bm{\sigma}^\pi])$ in matrix form, which places restrictions on the action-state covariance matrix, though in a less restrictive manner than the first equation.
Specifically, let $\mathbf{M}^\pi$ denote the action-state covariance matrix in the unique equilibrium under an information structure $\pi$, defined as
\[
\mathbf{M}^\pi \equiv \mqty[\Var \qty[\bm{\sigma}^\pi] & \Cov\qty[\bm{\sigma}^\pi, \bm{\theta}] \\ \Cov\qty[\bm{\theta}, \bm{\sigma}^\pi] & Z].
\]
Since $\mathbf{M}^\pi$ is symmetric and its lower-right block is given, it contains $(n(n+2m+1))/2$ free variables, while \eqref{moment2} imposes only $n$ linear constraints on these variables.

The following lemma establishes that these restrictions are not only necessary but also sufficient for an $(n+m)$-by-$(n+m)$ positive semidefinite matrix to be induced as the action-state covariance matrix under some Gaussian information structure.

\begin{lemma} \label{lem_moment}
Given any information structure $\pi$, let $X=\Var[\bm{\sigma}^\pi]$ and $Y=\Cov[\bm{\sigma}^\pi,\bm{\theta}]$.
Then,
\begin{equation} \label{feasibility}
\diag(XQ^\top )=\diag(YR^\top) \quad \text{and} \quad
\mathbf{M}_{X,Y} \equiv \mqty[X & Y \\ Y^\top & Z] \succeq O.
\end{equation}
Conversely, for any $X \in \calS^n$ and $Y \in \calM^{n \times m}$, if \eqref{feasibility} is satisfied, then there exists a Gaussian information structure $\pi$ such that $\mathbf{M}^{\pi} = \mathbf{M}_{X,Y}$. 
\end{lemma}

\begin{definition}
A pair of matrices $(X,Y) \in \calS^n \times \calM^{n,m}$ is \emph{(primal) feasible} if it satisfies \eqref{feasibility}.
\end{definition}

Lemma \ref{lem_moment} shows that while the condition \eqref{feasibility} is necessarily satisfied under \emph{any} information structure, it is sufficient for any candidate matrix to be induced under some \emph{Gaussian} information structure.
As a result, we have
\[
\qty{\mathbf{M}^\pi: \pi\in \Pi}
= \qty{\mathbf{M}_{X,Y}: \text{$(X,Y)$ is feasible}}
= \qty{\mathbf{M}^\pi: \pi\in \Pi^{\rm g}}.
\label{covariance key eq}
\]
In other words, the set of covariance matrices over all BNEs induced by arbitrary information structures coincides with that over all Gaussian ones.
Together with Lemma~\ref{lem_obj} below, this observation ensures the optimality of Gaussian information structures when the designer's objective is quadratic in the state and actions.


\subsection{LQG information design}
\label{sec_lqgd}

Now, returning to the first stage, the designer aims to maximize her expected objective by anticipating that agents play the unique BNE under the chosen information structure.
\emph{LQG information design} is the problem of finding an information structure that attains this maximum:
\begin{equation} \label{problem 1}
v^{\rm p} \equiv \max_{\pi \in \Pi} \E\qty[v(\bm{\sigma}^\pi, \bm{\theta})].
\end{equation}
We refer to $\E\qty[v(\bm{\sigma}^\pi, \bm{\theta})]$ as the value of an information structure $\pi$.
An information structure $\pi$ is \emph{optimal} if its value achieves $v^{\rm p}$.


\begin{lemma} \label{lem_obj}
For any information structure $\pi \in \Pi$, we have
\[
\E\qty[v(\bm{\sigma}^\pi, \bm{\theta})] = \mathbf{V} \bullet \mathbf{M}^\pi +
\underbrace{\qty(\bar{v}(\bar{\bm{a}}, \bar{\bm{\theta}}) + \mqty[\bar{\bm{a}} \\ \bar{\bm{\theta}}]^\top \mathbf{V} \mqty[\bar{\bm{a}} \\ \bar{\bm{\theta}}])}_{(*)}.
\]
\end{lemma}

Lemma~\ref{lem_obj} shows that the value of any information structure $\pi$ can be expressed as a linear function of the action-state covariance matrix it induces.
Since the constant term $(*)$ does not affect LQG information design, we normalize $\bar{v}$ so that $(*) = 0$ henceforth.
Consequently, the value of $\pi$ is given by $\mathbf{V} \bullet \mathbf{M}^\pi$.
At the same time, Lemma~\ref{lem_moment} fully characterizes the set of all feasible $\mathbf{M}^\pi$.
Combining these results, we can reformulate the problem \eqref{problem 1} as the following optimization problem:
\begin{equation} \label{primal_sdp}
v^{\rm p} = \max_{X \in \calS^n,\, Y \in \calM^{n,m}}  \mathbf{V} \bullet \mathbf{M}_{X,Y} \quad \text{subject to} \quad
\diag \qty(XQ^\top)=\diag \qty(YR^\top), \quad
\mathbf{M}_{X,Y} \succeq O. \tag{P}
\end{equation}

This problem is referred to as the \emph{primal problem}, as we will construct its dual in the next section.
The choice variables in this formulation are two matrices, $X$ and $Y$, which represent the variance of actions and the covariance between actions and states, respectively.
Crucially, the designer’s objective depends on the choice of information structures only through the induced action-state covariance matrix $\mathbf{M}_{X,Y}$ determined by $(X,Y)$.
Thus, we identify two information structures $\pi$ and $\pi'$ as equivalent if they induce the same covariance matrix, i.e., $\mathbf{M}^\pi = \mathbf{M}^{\pi'}$.
Moreover, we say that an optimal information structure $\pi$ is \emph{unique} if this equality holds for all optimal information structures.

The solution to \eqref{primal_sdp} depends on the agents' payoff parameters $(Q,R)$, the designer's payoff parameters $(V,W)$, and the state variance matrix $Z$.
Collecting these, we refer to the profile $(Q, R, V, W, Z)$ as the \emph{environment} of LQG information design.

\begin{example} \label{ex_no}
\emph{No disclosure}, denoted by $\underline{\pi}$, is an information structure in which agents receive no information, meaning each $\eta_i$ is constant.
Under $\underline{\pi}$, agents' equilibrium strategies remain constant at $\bar{\bm{a}}$, so the value of $\underline{\pi}$ is $0$.
\end{example}

\begin{example} \label{ex_full}
\emph{Full disclosure}, denoted by $\overline{\pi}$, provides every agent with the exact realization of the state, $\eta_i = \bm{\theta}$ for all $i \in [n]$.
Since $\bm{\theta}$ becomes common knowledge under $\overline{\pi}$, agents follow the action profile $Q^{-1}R\bm{\theta}$ for each state realization.
The value of $\overline{\pi}$ is then given as
\[
\overline{C} \bullet Z, \quad {\rm where} \quad
\overline{C} \equiv (Q^{-1}R)^\top VQ^{-1}R + \frac{(Q^{-1}R)^\top W+W^\top Q^{-1}R}{2}.
\]
\end{example}

The values of $\underline{\pi}$ and $\overline{\pi}$ can both be expressed in the linear form $C \bullet Z$ with respect to the state variance $Z$, where $C = O$ for $\underline{\pi}$ and $C = \overline{C}$ for $\overline{\pi}$.
In fact, as will be implied by strong duality, this is a general property of optimal information structures.


\section{Main results}
\label{sec_main}

Lemma~\ref{lem_moment} provides the key insight that the set of action-state covariance matrices attainable under all information structures coincides with that under all Gaussian ones.
Combined with Lemma~\ref{lem_obj}, this implies the designer can focus exclusively on Gaussian information structures without loss of optimality.

\begin{proposition} \label{prop_gauss}
There exists an optimal information structure that is Gaussian.
\end{proposition}

This result is novel in the literature, although similar conclusions have been anticipated in more specific contexts by \cite{bergemannmorris2013}, \cite{bhm2015}, \cite{tamura2018}, and \cite{smolinyamashita2023}.
In light of Proposition~\ref{prop_gauss}, we can henceforth focus on the class of Gaussian information structures.  
Note that a non-Gaussian information structure can also be optimal if it induces the same action-state covariance matrix as an optimal Gaussian information structure, although these are regarded as equivalent in our analysis.

Furthermore, Lemmas~\ref{lem_moment} and~\ref{lem_obj} allow us to reformulate LQG information design as \eqref{primal_sdp}.
This formulation aligns with the structure of SDP, which is characterized by the maximization of a linear function of a positive semidefinite matrix subject to linear constraints.
Beyond its numerical tractability, as demonstrated in Section~\ref{sec_ex}, SDP offers several analytical advantages as well.
In particular, its associated dual problem can be systematically constructed and enjoys well-established duality results.

The dual problem of \eqref{primal_sdp} is formulated as follows.
We introduce Lagrange multipliers $\lambda_1,\ldots,\lambda_n$ corresponding to each agent’s incentive constraint and denote $\Lambda = \Diag\qty(\lambda_1,\ldots,\lambda_n)$.
Using these, we define the following augmented incentive matrices:
\[
A_\Lambda \equiv \frac{\Lambda Q + Q^\top \Lambda}{2} - V
\quad {\rm and} \quad B_{\Lambda} \equiv \frac{\Lambda R + W}{2},
\]
The dual problem is then given as the following minimization problem, where the optimization variables are the diagonal matrix $\Lambda$ and a symmetric $m$-by-$m$ matrix $\Gamma$:
\begin{equation} \label{dual_sdp}
v^{\rm d} \equiv \min_{\Lambda \in \calS^n_{\rm diag},\, \Gamma\in \calS^{m}} Z \bullet \Gamma \quad {\rm subject \ to} \quad \mqty[A_\Lambda & -B_\Lambda \\ - B_\Lambda^\top & \Gamma] \succeq O. \tag{D}
\end{equation}
We say that a pair of matrices $(\Lambda,\Gamma)$ is \emph{dual feasible} when the constraint in \eqref{dual_sdp} is satisfied.
The next result shows that the dual optimal value $v^{\rm d}$ is equal to the primal optimal value $v^{\rm p}$, i.e., strong duality holds.

\begin{proposition} \label{prop_dual}
Both the primal problem \eqref{primal_sdp} and the dual problem \eqref{dual_sdp} admit solutions.
For any primal feasible $(X,Y)$ and any dual feasible $(\Lambda, \Gamma)$, it holds that $Z \bullet \Gamma \ge \mathbf{V} \bullet \mathbf{M}_{X,Y}$.
Moreover, strong duality holds, i.e., $v^{\rm p} = v^{\rm d}$.
\end{proposition}

Since $Z \succeq O$, the objective of \eqref{dual_sdp} is nondecreasing in $\Gamma$ with respect to the matrix order $\succeq$.
The dual feasibility constraint, in turn, imposes a lower bound on feasible $\Gamma$.
Specifically, by the standard property of positive semidefinite block matrices, dual feasibility implies
\[
\Gamma \succeq B_\Lambda^\top A_\Lambda^+ B_\Lambda \equiv C_\Lambda,
\]
where $A_\Lambda^+$ denotes the Moore--Penrose pseudoinverse of $A_\Lambda$.\footnote{See 
Proposition 10.2.5 of \cite{bernstein2018}, which we also record as Lemma~\ref{lem_block}. The Moore--Penrose inverse of $M\in \calM^{m,n}$ is the unique matrix $M^+\in \calM^{n,m}$ such that $MM^+M=M$, $M^+MM^+=M^+$, $(MM^+)^\top=MM^+$, and $(M^+M)^\top=M^+M$.}
This observation allows us to simplify \eqref{dual_sdp} as follows:
\begin{equation} \label{dual_sdp2} 
v^{\rm d} = \min_{\Lambda \in \calS^n_{\rm diag}} \quad Z \bullet C_\Lambda \quad {\rm subject \ to} \quad
\mathbf{M}_{\Lambda} \equiv \mqty[A_\Lambda & -B_\Lambda \\ - B_\Lambda^\top & C_\Lambda] \succeq O. \tag{D'}
\end{equation}

\begin{definition}
A matrix $\Lambda \in \calS^n_{\rm diag}$ is \emph{dual feasible} if it satisfies the constraint in \eqref{dual_sdp2}.
In particular, it is \emph{dual optimal} if $v^{\rm d} = Z \bullet C_\Lambda$.
\end{definition}

We are now ready to present our main theorem.
It characterizes the optimality of a given information structure through a system of matrix equations involving the induced action-state covariance and Lagrangian multipliers.
This characterization holds for any information structure, including those that may be non-Gaussian, while Proposition~\ref{prop_gauss} guarantees that there always exists an optimal information structure that is Gaussian.
By focusing on such a Gaussian optimal information structure, the optimality condition further simplifies, and the theorem establishes a precise linear relationship between the induced equilibrium strategy profile and the state.

\begin{theorem} \label{thm_main}
For any information structure $\pi$, it is optimal if and only if there exists a dual feasible diagonal matrix $\Lambda$ such that
\begin{equation} \label{cs} \tag{CS}
A_\Lambda \Var \qty[\bm{\sigma}^\pi] = B_\Lambda \Cov \qty[\bm{\theta}, \bm{\sigma}^\pi]
\quad \text{and} \quad
A_\Lambda \Cov \qty[\bm{\sigma}^\pi, \bm{\theta}] = B_\Lambda Z.
\end{equation}
Moreover, when $\pi$ is Gaussian, it is optimal if and only if there exists a dual feasible diagonal matrix $\Lambda$ such that
\begin{equation} \label{opt}
A_\Lambda \qty(\bm{\sigma}^\pi-\bar{\bm{a}}) = B_\Lambda \qty(\bm{\theta} -\bar{\bm{\theta}}) \quad {\rm a.s.} \tag{OPT}
\end{equation}
\end{theorem}

\begin{definition}
A dual feasible matrix $\Lambda$ \emph{certifies} an optimal information structure $\pi$ if \eqref{cs} holds for $\pi$ and $\Lambda$.
Also, $\Lambda$ certifies a primal feasible pair $(X,Y)$ if it satisfies the equations $A_\Lambda X = B_\Lambda Y^\top$ and $A_\Lambda Y = B_\Lambda Z$.
\end{definition}

Theorem~\ref{thm_main} shows that a given information structure $\pi$ is optimal if and only if there exists an appropriate choice of Lagrangian multipliers that certify $\pi$ through condition \eqref{cs}.
This condition arises from the \emph{complementary slackness condition} in SDP, which asserts that the ranges of the primal matrix $\mathbf{M}_{X,Y}$ and the dual matrix $\mathbf{M}_{\Lambda}$ are orthogonal.
As explained in \cite{vandenbergheboyd1996}, this condition can be seen as generalizing the familiar complementary slackness condition in LP.

Specifically, just as in LP, the condition \eqref{cs} is obtained as a restatement of strong duality.
Consequently, any dual optimal $\Lambda$ that achieves $v^{\rm d}$ satisfies complementary slackness with any primal optimal $(X,Y)$.
Such a $\Lambda$ therefore serves as a ``universal'' certificate, meaning it can certify the optimality of any information structure.
This universality is a useful property to analyze the uniqueness of an optimal information structure.

\begin{corollary} \label{cor_certify}
The following are equivalent.
\begin{enumerate}[\rm i).]
\item $\Lambda$ is dual optimal.
\item $\Lambda$ certifies some optimal information structure.
\item $\Lambda$ certifies every optimal information structure.
\end{enumerate}
\end{corollary}

While condition~\eqref{cs} applies to any optimal information structure, including those that are not Gaussian, the characterization becomes sharper when we focus on Gaussian ones.
Specifically, Theorem~\ref{thm_main} shows that at the optimum, the equilibrium strategy profile is linearly related to the state via condition~\eqref{opt}, where the matrices $A_\Lambda$ and $B_\Lambda$ from the dual problem serve as the relevant coefficients in the linear relationship.


\section{Noise-freeness and state-identifiability}
\label{sec_nfsi}

The condition \eqref{opt} provides key insights into the structure of the optimal BCE by imposing bidirectional restrictions on the realizations of $\bm{\theta}$ and $\bm{\sigma}^{\pi}$: if the state $\bm{\theta}$ is realized, the equilibrium strategy profile $\bm{\sigma}^\pi$ must satisfy \eqref{opt}; conversely, by observing $\bm{\sigma}^{\pi}$, the realization of $\bm{\theta}$ can be inferred from \eqref{opt}. 
In particular, when $A_\Lambda$ is invertible, $\bm{\sigma}^\pi$ can be expressed as a function of $\bm{\theta}$, meaning the realization of $\bm{\sigma}^{\pi}$ is uniquely identified from that of $\bm{\theta}$. 
Conversely, when $B_\Lambda$ is (left) invertible, the reverse holds. To explore this, we recall the following properties of information structures, which are observed in the numerical solutions discussed in Section~\ref{sec_ex}.

\begin{definition} An information structure $\pi$ is \emph{noise-free} if $\Var \qty[\bm{\sigma}^\pi \mid \bm{\theta}] = O$ a.s. By contrast, $\pi$ is \emph{state-identifiable} if $\Var \qty[\bm{\theta} \mid \bm{\sigma}^\pi] = O$ a.s. \end{definition}

\begin{corollary} \label{cor_nfsi1}
The following statements hold:
\begin{enumerate}[\rm i).]
\item \label{nf1}
If there exists a solution to \eqref{dual_sdp} such that $A_\Lambda \succ O$, then any Gaussian optimal information structure is noise-free.
Moreover, the optimal information structure is unique.
\item \label{si1}
Suppose that $m \le n$.
If there exists a solution to \eqref{dual_sdp} such that $\rank(B_\Lambda) = m$, then any Gaussian optimal information structure is state-identifiable.
\end{enumerate}
\end{corollary}

Corollary~\ref{cor_nfsi1} provides sufficient conditions under which any optimal information structure satisfies noise-freeness or state-identifiability.
Specifically, when the dual problem admits a solution such that $A_\Lambda \succ O$, we can take the inverse of $A_\Lambda$ in \eqref{opt} to express $\bm{\sigma}^\pi$ as a function of $\bm{\theta}$.
As a consequence, the equilibrium strategy profile exhibits no extraneous volatility beyond the fundamental randomness of $\bm{\theta}$, thereby ensuring noise-freeness.
Moreover, in this case, the optimal information structure must be unique.
This result leverages the universality of certifiers established in Corollary~\ref{cor_certify}.

On the other hand, when $\rank(B_\Lambda) = m$ so that $B_\Lambda$ is left invertible, we can express $\bm{\theta}$ as a function of $\bm{\sigma}^\pi$ via \eqref{opt}. 
This implies that the realized state is uniquely determined by the equilibrium strategy profile, establishing state-identifiability.
In such situations, the state is effectively disclosed to the entire population of agents, while individual agents may remain only partially informed.
This reflects the designer's maximal use of information about the state to guide agents toward the designer-optimal outcome. 

We next turn to conditions under which a given noise-free or state-identifiable information structure is optimal.
The following result shows that, in such cases, one of the matrix equations in \eqref{cs} becomes redundant and can be omitted.

\begin{corollary} \label{cor_nfsi2}
Given any $\pi \in \Pi^{\rm g}$, let $X = \Var \qty[\bm{\sigma}^\pi]$ and $Y = \Cov \qty[\bm{\sigma}^\pi, \bm{\theta}]$.
\begin{enumerate}[\rm i).]
\item When $\pi$ is noise-free, it is optimal if and only if there exists a dual feasible matrix $\Lambda$ such that $A_\Lambda Y = B_\Lambda Z$.
\item When $\pi$ is state-identifiable, it is optimal if and only if there exists a dual feasible matrix $\Lambda$ such that  $A_\Lambda X = B_\Lambda Y^\top$.
\end{enumerate}
\end{corollary}

No disclosure $\underline{\pi}$ and full disclosure $\overline{\pi}$ from Examples~\ref{ex_no}~and~\ref{ex_full} correspond to special cases of noise-free information structures, with $\overline{\pi}$ additionally satisfying state-identifiability.
The next corollary characterizes the conditions under which each of these is optimal.

\begin{corollary} \label{cor_no_full}
The following statements hold:
\begin{enumerate}[\rm i).]
\item No disclosure $\underline{\pi}$ is optimal if and only if there exists a dual feasible matrix $\Lambda$ such that
\begin{equation} \label{opt_no} \tag{\theequation} \refstepcounter{equation}
B_\Lambda Z B_\Lambda^\top = O.
\end{equation}
\item Full disclosure $\overline{\pi}$ is optimal if and only if there exists a dual feasible matrix $\Lambda$ such that
\begin{equation} \label{opt_full}
\qty(A_\Lambda Q^{-1}R - B_\Lambda) Z \qty(A_\Lambda Q^{-1}R - B_\Lambda)^\top = O.
\end{equation}
\end{enumerate}
In particular, if $Z \succ O$, then \eqref{opt_no} and \eqref{opt_full} simplify to $B_\Lambda=O$ and $A_\Lambda Q^{-1}R = B_\Lambda$, respectively.
\end{corollary}

The optimality of these information structures hinges on the dimension of the payoff state.
Specifically, when $\bm{\theta}$ is multi-dimensional so that $Z \succ O$ and $m \geq 2$, both $\underline{\pi}$ and $\overline{\pi}$ are generically suboptimal.
This is because each of the conditions $B_\Lambda = O$ and $A_\Lambda Q^{-1}R = B_\Lambda$ consists of $nm$ linear equations, while $\Lambda$ has only $n$ free variables.
Thus, a solution does not exist generically.

In contrast, when the state is one-dimensional, the following result shows that $\underline{\pi}$ can be optimal for a generic range of environments.
More interestingly, it reveals a ``bang-bang'' property of optimal information structures: whenever no disclosure is suboptimal, any optimal information structure must be state-identifiable, ensuring that the state is fully disclosed to agents in aggregate.

\begin{corollary} \label{cor_opt_m1}
With $m=1$, denote by $r_i$ and $w_i$ the $i$-th entries of $n$-dimensional vectors $R$ and $W$, respectively.
Assuming $r_i \neq 0$ for all $i \in [n]$, let $\Delta = \Diag(w_{1}/r_{1},\ldots, w_{n}/r_{n})$ and $\tilde{V} = V+(\Delta Q+Q^\top \Delta)/2$.
\begin{enumerate}[\rm i).]
\item No disclosure $\underline{\pi}$ is optimal if and only if $O \succeq \tilde{V}$, and uniquely optimal if $O \succ \tilde{V}$.
\item Otherwise, if $O \not\succeq \tilde{V}$, any optimal information structure is state-identifiable.
\end{enumerate}
\end{corollary}

To understand the difference between the one-dimensional and multi-dimensional cases, observe that the condition $B_\Lambda = O$ holds if and only if $w_{ik} = - \lambda_i r_{ik}$ for all $i \in [n]$ and $k \in [m]$.
That is, each row of $W$ must be proportional to the corresponding row of $R$, which indicates that the way the designer evaluates action-state correlations must align with how each agent values each dimension of the state.
When $m = 1$, this requirement is satisfied by setting $\lambda_i = - w_i/r_i$.
Constructing $\Lambda$ this way, $\underline{\pi}$ is optimal if and only if $A_\Lambda \succeq O$, which holds for a nontrivial set of payoff parameters.
In contrast, when $m \geq 2$, this proportionality condition is generally violated, rendering $\underline{\pi}$ suboptimal.

Formally, we say that an environment satisfies the \emph{proportional decision criterion} (PDC) if there exists a diagonal matrix $\Delta = \Diag (\delta_1,\ldots,\delta_n)$ such that $W = \Delta R$.
In other words, PDC holds when each row $(w_{i1}, \ldots, w_{im})$ of $W$ is obtained by scaling the corresponding row $(r_{i1}, \ldots, r_{im})$ of $R$ by a row-specific constant $\delta_i$.
When $m\ge 2$, this proportionality condition appears knife-edge in an algebraic sense.

We note, however, that PDC is satisfied in several economically relevant contexts.
For instance, when the designer’s objective is state-independent (i.e., $W = O$), PDC trivially holds with $\delta_i = 0$ for all $i \in [n]$.
Generalizing this case, the next section focuses on an environment in which each agent's payoff-relevant uncertainty is summarized by her ``personal'' state.
This setup is built on the assumption that the designer's objective depends on the state only through these personal states, thereby ensuring PDC.


\section{Personal-state environment}
\label{sec_personal}

We now focus on an environment in which the state vector is redefined in terms of what we call agents' personal states.
Throughout, we normalize $\bar{u}_i \equiv 0$ so that each agent's payoff function is given by:
\begin{equation} \label{payoff_personal}
u_i(\bm{a}, \bm{\theta}) = \qty(\hat{\theta}_i - \sum_{j \neq i} q_{ij} a_j) \cdot a_i - \frac{q_{ii} a_i^2}{2}, \quad {\rm where} \quad \hat{\theta}_i \equiv \sum_{k \in [m]} r_{ik} \theta_k.
\end{equation}
Here, $\hat{\theta}_i$ is referred to as agent $i$'s \emph{personal state}, which summarizes the exogenous payoff-relevant random variables for agent $i$.
This corresponds to a generalized version of the payoff structure considered in the network game example from Section~\ref{sec_ex}.

In the remainder of the main text, we treat the vector of personal states $\hat{\bm{\theta}} = (\hat{\theta}_1, \ldots, \hat{\theta}_n)$, rather than the original payoff states $\bm{\theta} = (\theta_1, \ldots, \theta_m)$, as the primitive of the model, and interpret the matrix $Z$ as representing $\Var[\hat{\bm{\theta}}]$.
Note that $Z$ may not be invertible, even if the original variance matrix $\Var[\bm{\theta}]$ is.
For instance, we have $\rank(Z) = 1$ if the vector $(r_{i1}, \ldots, r_{im})$ remains the same across all $i$ so that they respond identically to $\bm{\theta}$.
It is worth emphasizing that using personal states is \emph{not} without loss of generality in studying LQG information design---though it naturally arises in several contexts---since this reduction implicitly assumes that the designer’s objective depends on $\bm{\theta}$ only through $\hat{\bm{\theta}}$.

In this section, we further assume that the designer’s objective does not depend on cross terms between agent $i$'s action $a_i$ and $j$'s personal state $\hat{\theta}_j$.
This assumption is reasonable in several contexts, as each personal state $\hat{\theta}_i$  only affects agent $i$'s payoff.
For example, it holds when the designer is oriented to welfare maximization, as will be discussed later.
To sum up, the class of environments considered in this section corresponds to the following parametric assumption in our general model.

\begin{definition}
An environment is called a \emph{personal-state environment} if
\[
m=n, \quad R=I, \quad {\rm and} \quad W=\Diag (w_1,\ldots,w_n).
\]
\end{definition}

Evidently, PDC holds in any personal-state environment.
The following lemma shows that, under PDC, the designer’s objective can be transformed into an equivalent state-independent form such that $W=O$ without changing the value of any information structure.

\begin{lemma} \label{lem_personal_W0}
Consider a general environment that satisfies PDC with respect to $\Delta \in \calS^n_{\rm diag}$.
Then, for any primal feasible pair $(X,Y)$, we have
\[
V \bullet X + W \bullet Y = \tilde{V} \bullet X, \quad {\rm where} \quad \tilde{V} = V+\frac{\Delta Q+Q^\top \Delta}{2}.
\]
\end{lemma}

In view of this lemma, the designer’s original objective $(V, W)$ can be equivalently transformed into a state-independent form $(\tilde{V}, O)$, where $\tilde{V} = V + (WQ+Q^\top W)/2$.
Under this transformed objective, the matrices $A_\Lambda$ and $B_\Lambda$ are given as
\begin{equation} \label{AB_personal_state}
A_\Lambda = \frac{\Lambda Q + Q^\top \Lambda}{2} - \tilde{V} \quad {\rm and} \quad B_\Lambda = \frac{\Lambda}{2}.
\end{equation}
Note that $B_\Lambda$ takes a particularly simple form, with its rank determined solely by the number of nonzero Lagrange multipliers.
This observation enables us to derive simple sufficient conditions for state-identifiability and noise-freeness in terms of model primitives.

\begin{corollary} \label{cor_personal1}
In a personal-state environment, let $\tilde{V} = V+(WQ+Q^\top W)/2$.
\begin{enumerate}[\rm i).]
\item \label{cor_personal_part1}
No disclosure $\underline{\pi}$ is optimal if and only if $O \succeq \tilde{V}$, and uniquely optimal if $O \succ \tilde{V}$.
\item \label{cor_personal_part2}
Otherwise, if $O \not\succeq \tilde{V}$, any optimal information structure identifies at least one agent's personal state, i.e., $\Var [\hat{\theta}_i \mid \sigma^\pi_1,\ldots,\sigma^\pi_n] = 0$ holds for some $i \in [n]$.
In particular, when $\tilde{v}_{ii} > 0$ for all $i \in [n]$, any optimal information structure is state-identifiable.
\end{enumerate}
\end{corollary}

Corollary~\ref{cor_personal1} extends the key insights from Corollary~\ref{cor_opt_m1} in the one-dimensional case to more general personal-state environments.
Part~\eqref{cor_personal_part1} shows that no disclosure is optimal when the designer's augmented matrix $\tilde{V}$ is negative semidefinite.
By contrast, part~\eqref{cor_personal_part2} establishes that if $\tilde{V}$ is not negative semidefinite, then at least one agent's personal state must be fully revealed under any optimal information structure.
Note that when the fundamental state is one-dimensional, any pair of personal states $\hat{\theta}_i$ and $\hat{\theta}_j$ are perfectly correlated.
Hence, fully revealing one agent’s personal state is equivalent to fully revealing the entire state vector, and in this sense, the two corollaries yield identical implications when $m=1$.

Beyond the one-dimensional case, Corollary~\ref{cor_personal1} further shows that $\tilde{v}_{ii} > 0$ for all agents is sufficient for any optimal information structure to be state-identifiable.
Under this condition, the next result adds that if $Z$ is nonsingular so that each agent's personal state contains an idiosyncratic component orthogonal to others' personal states, the optimal information structure is not only state-identifiable but also unique and noise-free.
In stating it, let $D_M = \Diag (m_{11}, \ldots, m_{nn})$ denote the diagonal matrix formed from the diagonal entries of a given square matrix $M = [m_{ij}]_{n \times n}$.

\begin{corollary} \label{cor_personal2}
In a personal-state environment, let $\tilde{V} = V+(WQ+Q^\top W)/2$.
Suppose that $Z \succ O$ and $\tilde{v}_{ii} > 0$ for all $i \in [n]$.
Then, there exists a unique optimal information structure that is state-identifiable and noise-free.
Moreover, full disclosure $\overline{\pi}$ is optimal if and only if $\tilde{V} \succeq O$ and $D_{\tilde{V}}^{-1} \tilde{V} = D_Q^{-1} Q$.
\end{corollary}

In Corollary~\ref{cor_personal2}, we also obtains the optimality condition for full disclosure $\overline{\pi}$ when $Z \succ O$.
Specifically, while the optimality of $\underline{\pi}$ is characterized by $O \succeq \tilde{V}$, the optimality of $\overline{\pi}$ requires not only the contrasting condition $\tilde{V} \succeq O$ but also an additional linear condition $D_{\tilde{V}}^{-1} \tilde{V} = D_Q^{-1} Q$, which is equivalently written as:
\begin{equation} \label{welfare_full_opt_private}
\frac{\tilde{v}_{ij}}{\tilde{v}_{ii}} = \frac{q_{ij}}{q_{ii}}, \quad \forall i \neq j.
\end{equation}
This condition requires the designer's matrix $\tilde{V}$ to be well aligned with the agents’ payoff structure $Q$: after normalizing $\tilde{V}$ and $Q$ by dividing each row by its corresponding diagonal entry, the two matrices must become identical.
As illustrated in the next example, this alignment eliminates strategic externalities among agents in the context of welfare maximization, highlighting that it is a rather rare case for $\overline{\pi}$ to be optimal when $Z$ is nonsingular.

\begin{example} \label{ex_welfare}
The personal-state environment naturally arises when the information designer is oriented toward welfare maximization.
Specifically, assume that the designer's objective is given as the following weighted aggregate of agents' payoffs:
\begin{equation} \label{eq_obj_welfare}
v(\bm{a},\hat{\bm{\theta}}) = 
\sum_{i \in [n]} v_i \cdot \qty(\qty(\hat{\theta}_i - \textstyle \sum_{j \neq i} q_{ij} a_j) \cdot a_i - \frac{q_{ii} a_i^2}{2}),
\end{equation}
where $v_i > 0$ for all $i \in [n]$.
In particular, the designer aims to maximize \emph{utilitarian welfare} if $v_i = 1$ for all $i \in [n]$.

Under this specification, the designer's value of any information structure $\pi$ can be expressed as an affine function of the equilibrium action volatility:\footnote{See Lemma~\ref{lem_obj_welfare} in Appendix~\ref{app_omit} for a formal calculation.}
\[
\E \qty[v(\bm{\sigma}^\pi, \hat{\bm{\theta}})] = \frac{1}{2} \sum_{i \in [n]} v_i q_{ii} \cdot \qty(\Var \qty[\sigma^\pi_i] + \bar{a}_i^2).
\]
This implies that welfare maximization maps to the personal-state environment with $V = \Diag (v_1 q_{11}, \ldots, v_n q_{nn})$ and $W = O$, up to some positive affine transformation of $v$.

Now, since all diagonal entries of $V$ are positive by assumption, Corollary~\ref{cor_personal1} ensures that any optimal information structure is state-identifiable.
Moreover, if $Z \succ O$, Corollary~\ref{cor_personal2} implies that the optimal information structure is noise-free and unique.
It also shows that the optimality of full disclosure $\overline{\pi}$ requires that condition~\eqref{welfare_full_opt_private} holds.
In the context of welfare maximization, since $V$ is diagonal and $W=O$, this condition holds only when $q_{ij} = 0$ for all $i \neq j$.
In other words, the optimality of $\overline{\pi}$ is a knife-edge case that arises only when there are no strategic externalities.
This highlights the importance of providing personalized information in the presence of strategic externalities and idiosyncratic payoff states.
\end{example}


\section{Applications to network games}
\label{sec_network}

In this section, we apply our general results to network intervention problems by revisiting the setting from Section~\ref{sec_ex}.

\begin{definition}
An environment is called a \emph{network environment} if
\[
m=n, \quad Q=I-\beta G, \quad {\rm and} \quad R=I,
\]
where $G = [g_{ij}]_{n \times n}$ is an adjacency matrix with $g_{ii}=0$ for all $i \in [n]$, and $\beta \in \R$ is a parameter such that $Q+Q^\top \succ O$.
If, in addition, $W=\Diag (w_1,\ldots,w_n)$, it is called a \emph{personal-state network environment}.
\end{definition}

\begin{remark}
In a network environment, the condition $Q + Q^\top \succ O$ holds if and only if $\beta$ satisfies
\begin{equation} \label{q_network}
\mu_{\rm min}\qty(\frac{G+G^\top}{2})^{-1} < \beta < \mu_{\rm max}\qty(\frac{G+G^\top}{2})^{-1},
\end{equation}
where $\mu_{\rm max}(\cdot)$ and $\mu_{\rm min}(\cdot)$ denote the largest and smallest eigenvalues of a (symmetric) matrix, respectively.
Note that since all entries of $G$ are nonnegative, the Perron--Frobenius theorem implies $\mu_{\rm max} (\textstyle \frac{G + G^\top}{2} ) \ge 0$, with strict inequality unless $G=O$.
Moreover, this implies that $\mu_{\rm min} (\textstyle \frac{G + G^\top}{2} ) \le 0$ since $\tr(G) = 0$.
\end{remark}

\subsection{Symmetry in networks}
\label{sec_sym}

A network $G$ describes the physical links between agents. 
Analogously, we can interpret the covariance matrix $\mathbf{M}^\pi$ as representing their informational links under an information structure $\pi$.
In this section, we show that when the network exhibits certain symmetries, the optimal information structure inherits the same degree of symmetry, provided that the designer's objective aligns with the symmetry in $G$.
Although this may seem intuitive, the result is not trivial since even in symmetric games, the designer could, in principle, benefit from providing discriminatory signals to different agents.

To formalize the notion of symmetry, we define the \emph{automorphism group} of $G$, denoted $\Aut(G)$, as the set of permutations $\tau: [n] \to [n]$ such that $g_{ij} = g_{\tau(i)\tau(j)}$ for all $i,j \in [n]$.
Then, the \emph{$G$-orbit} of agent $i$  is defined as
\begin{equation*}
O_i(G) \equiv \qty{j \in [n]: \tau(i) = j \text{ for some } \tau \in \Aut(G)}.
\end{equation*}
In words, $\Aut(G)$ consists of relabelings of agents that leave $G$ unchanged, and $O_i(G)$ collects agents who can be reached from $i$ via such relabelings.
If $j \in O_i(G)$, we can interpret this as agents $i$ and $j$ playing interchangeable roles in the network $G$.
It is easy to see that orbits partition $[n]$ into equivalent classes.

We say that $\Aut(G)$ is \emph{transitive} if $O_i(G) = [n]$ for all (or, equivalently, some) $i$, in which case all agents’ roles are interchangeable.
By slight abuse of terminology, we also call $G$ \emph{transitive} when $\Aut(G)$ is transitive.
In general, the coarseness of orbits captures the degree of symmetry that $G$ entails.
For example, the automorphism group of a complete network consists of all permutations, implying full symmetry among agents.
Aside from the complete network, the cyclical network in Section~\ref{sec_ex} is also transitive, whereas the star-shaped network is not since $[n]$ is partitioned into two orbits, $\{1\}$ and $\{2,\ldots,n\}$.

For any permutation $\tau:[n] \to [n]$, we define the associated permutation matrix as
\[
P_\tau \equiv \mqty[\mathbf{e}_{\tau(1)} & \cdots &  \mathbf{e}_{\tau(n)}],
\]
where the $i$-th column $\mathbf{e}_{\tau(i)}$ is given as the $\tau(i)$-th unit vector.
For a square matrix $M$, let $\tau(M) \equiv P_{\tau} M P_{\tau}^{\top}$ be the matrix obtained by permuting the rows and columns of $M$ according to $\tau$. 
We say that $M$ is \emph{$\tau$-invariant} if $M = \tau(M)$.
Also, $M$ is \emph{$G$-invariant} if it is $\tau$-invariant for all $\tau \in \Aut(G)$.
For example, $M$ is invariant for any permutation when its diagonal entries are all identical, and its off-diagonal entries are all identical as well.

A key property of transitive networks is that each agent has the same number of incoming and outgoing links.
Formally, for each $i$, define the \emph{in-degree} and \emph{out-degree} as $d_i^{\rm in}(G) \equiv \sum_{j \neq i} g_{ji}$ and $d_i^{\rm out}(G) \equiv \sum_{j \neq i} g_{ij}$.
We say that $G$ is \emph{regular} if there exists an integer $d(G)$ such that $d_i^{\rm in}(G) = d_i^{\rm out}(G) = d(G)$ for all $i \in [n]$, in which case $d(G)$ is simply called the \emph{degree} of $G$.
More generally, the next lemma shows that if matrix $M$ is invariant with respect to some transitive network $G$, then all of its rows sum to the same constant and all of its columns sum to that same constant.

\begin{lemma} \label{lem_regular}
If $G$ is a transitive network and $M = [m_{ij}]_{n \times n}$ is a $G$-invariant matrix, then all row‐sums $\sum_{j \in [n]} m_{ij}$ and all column-sums $\sum_{j \in [n]} m_{ji}$ are equal to a common constant.
In particular, any transitive network is regular.
\end{lemma}

Since any network $G$ itself is $G$-invariant, the agents' payoff matrices $Q = I-\beta G$ and $R = I$ are $G$-invariant.
In the next definition, we further postulate that the designer's objective and the distribution of personal states are aligned with the symmetry present in $G$.
In addition, we define the notion of $G$-invariance for an information structure based on the $G$-invariance of the action-state covariance matrix induced by it.

\begin{definition}
A network environment is \emph{$G$-invariant} if $V$, $W$, and $Z$ are $G$-invariant matrices.
An information structure $\pi$ is \emph{$G$-invariant} if $\Var[\bm{\sigma}^\pi]$ and $\Cov[\bm{\sigma}^\pi, \bm{\theta}]$ are $G$-invariant matrices.
\end{definition}

\begin{proposition} \label{prop_sym}
In a $G$-invariant network environment, there exists an optimal information structure $\pi$ that is $G$-invariant.
Moreover, $\pi$ is certified by $\Lambda = \Diag (\lambda_1,\ldots,\lambda_n)$ such that $\lambda_i = \lambda_j$ whenever $j \in O_i(G)$.
In particular, $\lambda_1 = \cdots =  \lambda_n$ if $G$ is transitive.
\end{proposition}

Proposition~\ref{prop_sym} shows that when the network environment is invariant with respect to $G$, there exists a $G$-invariant optimal information structure.
Moreover, the certifying Lagrangian multipliers can be chosen to form a $G$-invariant diagonal matrix.
In particular, when the network is transitive, the multipliers can be uniform across all agents.
This result is useful for simplifying the computation of the optimal information structure.

Now, consider a personal-state network environment that is invariant with respect to some transitive network $G$.
In this case, we must have $W = wI$ for some scalar $w \in \R$.
Thus, the environment exhibits PDC, as each row of $W$ is obtained by scaling the corresponding row of $R = I$ by $w$.
We then obtain the following bang-bang result, which is analogous to Corollary~\ref{cor_opt_m1}.

\begin{corollary} \label{cor_sym_bangbang}
In a $G$-invariant personal-state network environment, suppose that $G$ is transitive so that $W=wI$ for some $w \in \R$.
Let $\tilde{V} = V+w/2 \cdot (Q+Q^\top)$.
\begin{enumerate}[\rm i).]
\item \label{cor_sym_bang1}
No disclosure $\underline{\pi}$ is optimal if and only if $O \succeq \tilde{V}$, and uniquely optimal if $O \succ \tilde{V}$.
\item \label{cor_sym_bang2}
Otherwise, if $O \not\succeq \tilde{V}$, any optimal information structure is state-identifiable.
If, in addition, $Z \succ O$, then the optimal information structure is unique and noise-free.
\end{enumerate}
\end{corollary}


\subsection{Welfare maximization in symmetric networks}
\label{sec_welfare_sym}

In this section, we explore utilitarian welfare maximization in network environments.
To this end, we recall the symmetric distribution of agents' personal states from Section~\ref{sec_ex}, where $Z = [z_{ij}]_{n \times n}$ is given as $z_{ii} = 1$ and $z_{ij} = \rho$ for $i \neq j$.
The case $\rho = 1$ corresponds to the common-value case, while $\rho < 1$ corresponds to the private-value case.

As discussed in Section~\ref{ex_welfare}, utilitarian welfare maximization maps to an environment with $V = I$ and $W = O$.
Corollary~\ref{cor_sym_bangbang} then implies that the optimal information structure must always be state-identifiable.
In the private-value case, it is furthermore unique and noise-free, while full disclosure $\overline{\pi}$ is suboptimal except in the knife-edge case where $\beta = 0$.
Intuitively, when agents are concerned with different payoff-relevant random variables, full disclosure allows them to infer others’ personal states and adjust their actions excessively, leading to misalignment with their own personal states.

By contrast, as the following proposition shows, this generic suboptimality of $\overline{\pi}$ does not arise in the common-value case, where all agents are concerned about the common state $\theta \equiv \theta_1 = \cdots = \theta_n$.
The proposition also characterizes the optimal degree of each agent's signal informativeness when $\overline{\pi}$ is suboptimal.


\begin{proposition} \label{prop_welfare_regular}
Consider any $G$-invariant personal-state network environment where $G$ is transitive.
Full disclosure $\overline{\pi}$ is optimal if and only if
\begin{equation} \label{regular_full_condition}
\beta \ge - \frac{1}{d(G) - 2\mu_{\rm min}(\frac{G+G^\top}{2})}.
\end{equation}
Otherwise, if \eqref{regular_full_condition} is violated, there exists an optimal information structure $\pi$ that is Gaussian, $G$-invariant, and state-identifiable yet not noise-free.
Moreover, under such $\pi$, each agent’s action recommendation reduces the common state variance as follows:
\begin{equation} \label{regular_partial}
s_i \equiv \frac{\Var \qty[\theta] - \Var \qty[\theta \mid \sigma^\pi_i]}{\Var \qty[\theta]} = \frac{1/|\beta| + \mu_{\rm min}(\frac{G+G^\top}{2})}{d(G) - \mu_{\rm min}(\frac{G+G^\top}{2})}.
\end{equation}
\end{proposition}

Assuming that the underlying network is transitive, Proposition~\ref{prop_welfare_regular} shows that $\overline{\pi}$ is optimal as long as $\beta$ remains above a threshold specified in \eqref{regular_full_condition}.
In particular, since the cutoff value is always negative, $\overline{\pi}$ is optimal whenever the game features strategic complements ($\beta \ge 0$).
By contrast, under strategic substitutes ($\beta < 0$), $\overline{\pi}$ is optimal only if the strength of the negative externality is sufficiently small.
Specifically, when $\beta$ falls below the threshold, partial disclosure becomes optimal, where the informativeness of each agent’s signal is given explicitly in \eqref{regular_partial} in terms of the spectrum of $G$.

Intuitively, when $\beta < 0$, the designer faces a ``congestion'' trade-off: agents are concerned with the common payoff state, yet their welfare decreases when their actions are aligned too much due to strategic substitutes.
To balance the tension between the maximal use of information and mitigating congestion, the designer employs a state-identifiable information structure such that the true state is revealed in aggregate while each individual agent is kept only partially informed.

The cutoff value in \eqref{regular_full_condition} increases with both the degree $d(G)$ and the magnitude of the smallest eigenvalue $\mu_{\rm min} (\frac{G + G^\top}{2})$.
To interpret this dependence, we rewrite \eqref{regular_full_condition} in terms of $\beta d(G)$ as follows:
\begin{equation} \label{regular_full_condition2}
\beta d(G) \ge -\frac{h(G)-1}{h(G)+1}, \quad {\rm where} \quad h(G) \equiv 1-\frac{d(G)}{\mu_{\rm min} (\frac{G+G^\top}{2})}.
\end{equation}
The quantity $\beta d(G)$ is interpreted as a measure of total externality that agent $i$ receives from others, as seen in the coefficient on the linear term of $a_i$ in $i$'s payoff:
\[
\phi_i(\bm{a}_{-i}, \theta_i)
= \theta_i + 
\beta d(G)
\times
\underbrace{\frac{1}{d(G)}\sum_{j \neq i}g_{ij}a_j}_{\substack{\text{average action} \\ \text{of $i$'s neighbors}}}.
\]
In \eqref{regular_full_condition2}, the cutoff for $\beta d(G)$ is given as a decreasing function of $h(G)$, known as \emph{Hoffman bound}---a key concept in spectral graph theory that is closely related to the chromatic number of the graph.\footnote{The \emph{chromatic number} of an undirected graph is the minimum number of colors needed to color all nodes of the graph so that no two adjacent nodes share the same color. Moreover, the \emph{Hoffman bound} is known to serve as a lower bound on the chromatic number; see Section 16.9 of \cite{spielman2012}. For instance, the complete graph of size $n$ has chromatic number $n$, while any bipartite graph has chromatic number $2$, which, in both cases, coincides with the Hoffman bound. While this spectral interpretation of $h(G)$ applies specifically to undirected networks, Proposition~\ref{prop_welfare_regular} remains valid for both directed and undirected networks.}
Specifically, for nontrivial undirected graphs of size $n$, the Hoffman bound $h(G)$ ranges between $2$ and $n$, with these extremes respectively attained by bipartite and complete graphs, exemplified as below.

\begin{example}
For the complete network $K$ of size $n$, we have $d(K) = n-1$ and $\mu_{\rm min}(K) = -1$, so $h(K) = n$.
Thus, $\overline{\pi}$ is optimal if and only if $(n-1)\beta$ remains above $-(n-1)/(n+1)$, which monotonically approaches $-1$ as $n$ grows large.
On the other hand, when $\overline{\pi}$ is suboptimal, each individual's signal reduces the state variance by $(1/|\beta|-1)/n$. 
\end{example}

\begin{example}
An (undirected) network $G$ is \emph{bipartite} if $[n]$ can be partitioned into two subsets such that $g_{ij} = 0$ whenever $i$ and $j$ belong to the same subset.
It is a well-known fact in spectral graph theory that $d(G) = -\mu_{\rm min}(G)$ holds for any bipartite graph; see Section 16.6 of \cite{spielman2012}.
This implies $h(G) = 2$, and thus $\overline{\pi}$ is optimal if and only if $\beta d(G) > -1/3$.
\end{example}

\begin{example}
When $n$ is even, an undirected cyclical network (Figure~\ref{fig_sym_graph}, middle) is bipartite, since $[n]$ can be partitioned into odd- and even-indexed agents, with no edges within each group.
In general, the eigenvalues of the cyclical network of size $n$ are given by $2 \cos(2\pi j/n)$ for $j \in \{0, \ldots, n-1\}$.
When $n$ is even, the minimum eigenvalue is $-2$, attained when $j = n/2$.
When $n$ is odd, the minimum eigenvalue is $2 \cos(\pi(n \pm 1)/n)$, which is strictly greater than $-2$ but tends to $-2$ as $n$ becomes large.
This explains why, in Figure~\ref{fig_cutoff}, the cutoff value for cyclical networks with odd $n$ deviates from $-1/3$ but converges to it as $n$ increases.
Further, when $\overline{\pi}$ is suboptimal, the informativeness of each agent's signal is (nearly) equal to $(1/|\beta| - 2)/4$.
\end{example}

\begin{figure}[t]
\begin{center}
\includegraphics[width=12cm]{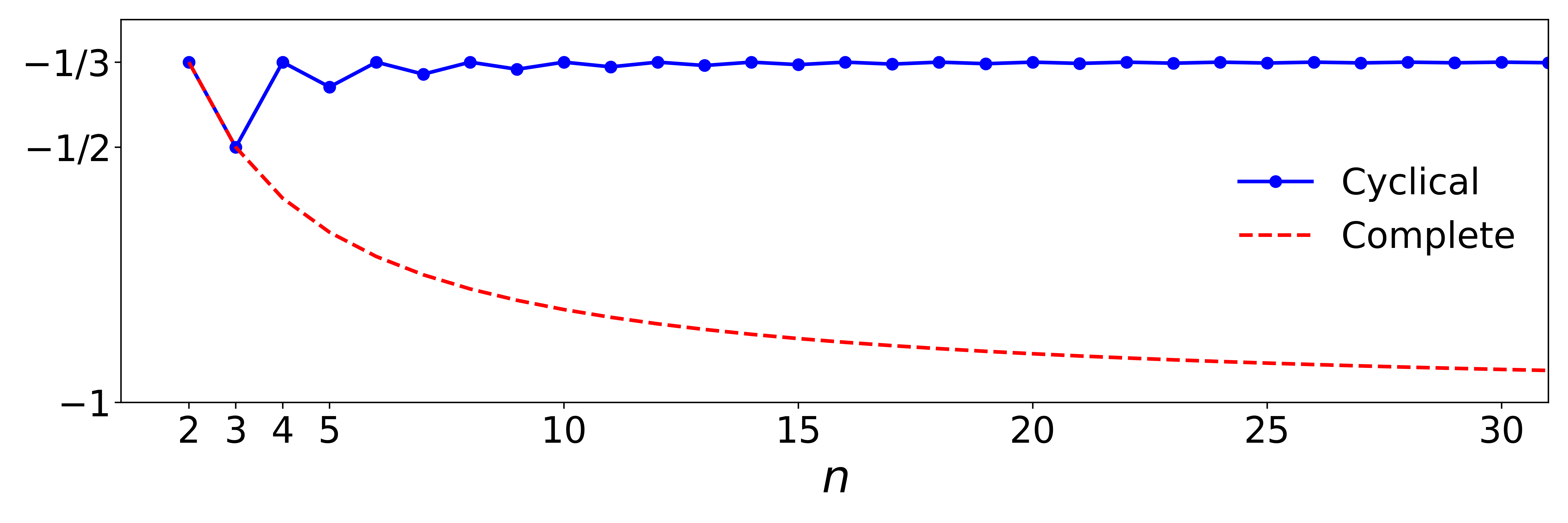}
\caption{Cutoff values that $\beta d(G)$ must exceed for $\overline{\pi}$ to be optimal, plotted against $n$.}
\label{fig_cutoff}
\end{center}
\end{figure}

Proposition~\ref{prop_welfare_regular} aligns with our numerical results both qualitatively and quantitatively.
For transitive networks, the optimality of $\overline{\pi}$ is monotonically supported by the strength of strategic interaction: there exists a threshold such that $\overline{\pi}$ is optimal if and only if $\beta d(G)$ exceeds it.
This threshold is given in \eqref{regular_full_condition2} and varies across networks and the number of agents.
In Figure~\ref{fig_cutoff}, we plot the cutoff values for complete and cyclical networks as a function of $n$.
For complete networks, the threshold decreases monotonically towards $-1$ as $n$ increases, indicating that full disclosure becomes optimal over a wider range of parameters as the population grows.
In contrast, for cyclical networks, it remains almost constant at $-1/3$.
Thus, the condition for $\overline{\pi}$ to be optimal is weaker in complete networks than in cyclical networks.\footnote{In terms of the cutoff value of $\beta$, the complete network has cutoff $-1/(n+1)$, while the cyclical graph has cutoff approximately $-1/6$, with equality when $n$ is even. Thus, without normalizing $d(G)$, the cutoff value of $\beta$ is lower for cyclical networks when $n \ge 6$.}

Moreover, even when $\overline{\pi}$ is suboptimal, the optimal information structure reveals more to each agent in complete networks than in cyclical networks.
To see this, we rewrite the signal informativeness $s_i$ from \eqref{regular_partial} in terms of the Hoffman bound $h(G)$ as follows:
\[
s_i = \frac{1}{|\beta| d(G)} - \frac{1+|\beta| d(G)}{|\beta| d(G) \cdot h(G)}.
\]
This expression increases with $h(G)$, holding $\beta d(G)$ fixed, whence the optimal signal is more informative in complete networks than in cyclical networks.
These two cases represent polar extremes, as $h(G)$ reaches its maximum value $n$ in complete networks and its minimum value $2$ in cycles.
Accordingly, they result in the most and least informative signal structures, respectively, among transitive networks that feature interchangeable agent roles.

When agent roles are not interchangeable, however, the optimality of $\overline{\pi}$ is no longer characterized by a monotonic parameter region.
Indeed, as shown in Figure~\ref{fig_star}, for undirected star-shaped networks, $\overline{\pi}$ is optimal when $\beta$ is  sufficiently small or sufficiently large, but not at intermediate values.

We anticipate that this non-monotonic pattern arises because even when all agents receive the same information, their equilibrium strategies can be negatively correlated if payoffs are asymmetric.
Specifically, in the undirected star-shaped network with general $n \ge 2$,  a straightforward calculation shows that the BNE under $\overline{\pi}$ is given by:
\[
\sigma_1 = \frac{1+(n-1)\beta}{1-(n-1)\beta^2} \cdot \theta \quad {\rm and} \quad \sigma_i =  \frac{1+\beta}{1-(n-1)\beta^2} \cdot \theta, \quad \forall i \ge 2.
\]
Observe that while all agents' strategies are linear in $\theta$, the signs of the coefficients can differ: when $-1 < \beta < -1/(n-1)$, the central agent's strategy $\sigma_1$ is negatively correlated with $\theta$, while every peripheral agent's strategy $\sigma_i$ remains positively correlated.

Intuitively, although all agents prefer to align their actions with the state, the central agent---connected to all others---faces stronger incentives to differentiate from others due to substitutability.
From the designer’s perspective, this contrarian behavior by the central agent can mitigate the welfare loss from congestion, supporting the optimality of $\overline{\pi}$.
However, at intermediate values of $\beta$, the central and peripheral agents' actions become positively correlated, eliminating this effect and rendering $\overline{\pi}$ suboptimal.


\section{Discussion}
\label{sec_conc}

This paper lays the foundation for understanding optimal informational interventions in LQG environments by developing a technical methodology grounded in relevant mathematical tools.
Our main contribution is the characterization of optimal information structures through the condition \eqref{opt}, which can be used to uncover general properties---such as state-identifiability and noise-freeness---and to derive explicit solutions in symmetric environments.
To conclude, we discuss some extensions of our analysis, some of which are addressed in the online appendix of this paper and our complementary work \cite{miyashita_ui_large}.

\medskip

{\it Beyond welfare maximization.}
The optimality condition \eqref{opt} takes the form of a high-dimensional system of equations, which is generally difficult to solve analytically.
Assuming that the designer maximizes utilitarian welfare simplifies the problem considerably, as it implies a particularly tractable objective, i.e.,  $V = I$ and $W = O$.

Alternatively, assuming agent homogeneity in the underlying network game allows us to go beyond welfare maximization.
Online Appendix~\ref{app_comp} provides closed-form solutions for general objectives in the case of complete networks.
Also, in \cite{miyashita_ui_large}, we characterize general optimal information structures in a continuum-population model with homogeneous agents, as in \cite{angeletospavan2007} and \cite{bergemannmorris2013}.\footnote{In that setting, we prove that \emph{targeted disclosure} is always optimal: the designer fully reveals the state realization to agents in a targeted group, while withholding all information from the rest.
Crucially, implementing this information revelation does not rely on the normality of the state.
It is worth noting that identifying the optimal size of the targeted group requires a continuum population, as the designer's payoff depends continuously on the fraction of informed agents. This type of fine-tuned targeting would not be feasible in a finite-agent setting.}

\medskip

{\it Beyond Gaussian state.}
This paper builds on three fundamental assumptions, i.e., agents have linear best responses, the designer has a quadratic objective, and the state is Gaussian.
Among the three, the last assumption is arguably the least critical for our analysis because, even without it, the value of the SDP problem \eqref{primal_sdp} still provides an upper bound on the designer’s expected payoff.
This is because the first part of Lemma~\ref{lem_moment} remains valid in the non-Gaussian case; that is, any action-state covariance matrix inducible by information \emph{necessarily} satisfy the same obedience and positive semidefiniteness conditions, as long as agents' have linear best responses.
Moreover, as long as the designer's payoff is quadratic, Lemma~\ref{lem_obj} implies that her objective depends only on the covariance matrix---hence, the achievable expected value must be bounded by the value of \eqref{primal_sdp}.

It remains an open question whether this upper bound is tight in the present finite-population setting.
However, with a continuum of homogeneous agents, the bound is tightly achieved by using the optimal information structure of   \cite{miyashita_ui_large}.

\medskip

{\it Public information design.}
In many real-world contexts, it is reasonable to restrict attention to public information structures, where the designer must disclose the same information to all agents.
A central bank, for instance, may disseminate macroeconomic forecasts uniformly to firms and consumers, constrained not only by practical limitations but also by institutional commitments to transparency \citep{morrisshin2002, angeletospavan2007, uiyoshizawa2015}.
Online Appendix~\ref{app_pub} analyzes LQG information design in this case.
Focusing on public signals gains analytical tractability, as it removes agents' strategic uncertainty about others' actions.
Incidentally, our problem reduces to a Bayesian persuasion problem with a single representative agent, which is analogous to the problem studied by \cite{tamura2018}.
Leveraging his solution method, we derive a closed-form expression for the optimal public information structure.
It is worth noting that the optimal public signal is generally suboptimal relative to what can be achieved with optimal private (non-public) signals.
This indicates that the designer typically benefits from manipulating agents' beliefs about strategic uncertainty.\footnote{\cite{arielibabichenko2019} and \cite{hoshino2022} also highlight the importance of leveraging strategic uncertainty in games with finite actions and binary states.}


\appendix

\renewcommand{\thesection}{A}
\renewcommand{\theequation}{A\arabic{equation}}
\renewcommand{\thelemma}{A\arabic{lemma}}
\renewcommand{\thecorollary}{A\arabic{corollary}}
\renewcommand{\theproposition}{A\arabic{proposition}}

\renewcommand{\theHlemma}{A\arabic{lemma}}
\renewcommand{\theHproposition}{A\arabic{proposition}}
\renewcommand{\theHequation}{A\arabic{equation}}

\setcounter{section}{0}
\setcounter{equation}{0}
\setcounter{lemma}{0}
\setcounter{proposition}{0}

\section{Proofs of main results}
\label{app_main}

This Appendix~\ref{app_main} contains the proofs of our main results, as well as auxiliary lemmas, appearing in Sections~\ref{sec_model}~and~\ref{sec_main}.
Proofs for the remaining sections are provided in Online Appendix~\ref{app_omit}.

We begin by recalling some mathematical results concerning SDP.
For a more comprehensive discussion, see \cite{vandenbergheboyd1996}.
Given $C,A_1,\ldots,A_m \in \calS^n$ and $\bm{b} \in \R^m$, a canonical SDP problem takes the following form:
\begin{equation} \label{app_p}
v^{\rm p} = \max_{X \in \calS^n} C \bullet X \quad \text{subject to} \quad A_i \bullet X = b_i \quad \forall i \in [m], \quad X \succeq O.
\end{equation}
The dual of this problem is given by:
\begin{equation} \label{app_d}
v^{\rm d} = \min_{\bm{y} \in \R^{m}} \bm{b}^\top \bm{y} \quad \text{subject to} \quad \sum_{i\in[m]} y_i A_i \succeq C.
\end{equation}
Each problem is said to be \emph{feasible} if it has at least one choice satisfying all the constraints.
It is \emph{strictly feasible} if there exists a feasible choice that satisfies the matrix inequality constraint strictly.
The following lemma summarizes standard results on weak and strong duality between the primal and dual problems.

\begin{lemma}\label{lem_dual1}
If both primal and dual problems are feasible, then both admit solutions and $- \infty < v^{\rm p} \le v^{\rm d} < \infty$ holds.
In particular, if a primal-dual feasible pair $(X,\bm{y})$ satisfies
\begin{equation} \label{app_kkt}
X \bullet \qty(\sum_{i\in[m]} y_i A_i - C) = 0,
\end{equation}
then $X$ attains $v^{\rm p}$ and $y$ attains $v^{\rm d}$.
If, in addition, either problem is strictly feasible, then $v^{\rm p} = v^{\rm d}$ holds, and any primal-dual feasible pair $(X,\bm{y})$ satisfies \eqref{app_kkt} if and only if $X$ attains $v^{\rm p}$ and $\bm{y}$ attains $v^{\rm d}$.
\end{lemma}

For ease of reference, we record some known facts about pseudoinverse matrices.
The next lemma characterizes the positive semidefiniteness of block matrices based on generalized Schur complements; see, e.g., Proposition 10.2.5 of \cite{bernstein2018}.

\begin{lemma} \label{lem_block}
For any block matrix $\mathbf{A} = \qty[\begin{smallmatrix} A&B\\ B^\top&C \end{smallmatrix}] \in \calS^{n+m}$, we have $\mathbf{A} \succeq O$ if and only if $A \succeq O$, $B = AA^+B$, and $C \succeq B^\top A^+ B$.
\end{lemma}

The next result provides a well-known conditional distribution formula for multivariate normal variables.
The version stated below, drawn from \cite{rao1973}, does not require the non-singularity of the variance matrix.

\begin{lemma} \label{lem_cond_normal}
Let $\bm{\xi} = (\bm{\xi}_1, \bm{\xi}_2)$ be a multivariate normal random vector with
\[
\E \mqty[\bm{\xi}_1 \\ \bm{\xi}_2] = \mqty[\bm{\mu}_1 \\ \bm{\mu}_2], \quad \Var \mqty[\bm{\xi}_1 \\ \bm{\xi}_2] = \mqty[\Sigma_{11} & \Sigma_{12} \\ \Sigma_{21} & \Sigma_{22}].
\]
Then,
\[
\E\qty[\bm{\xi}_1 \mid \bm{\xi}_2] = \mu_1 + \Sigma_{12}\Sigma_{22}^+\qty(\bm{\xi}_2-\bm{\mu}_2), \quad
\Var\qty[\bm{\xi}_1 \mid \bm{\xi}_2] = \Sigma_{11} - \Sigma_{12}\Sigma_{22}^+ \Sigma_{21}.
\]
\end{lemma}


\subsection*{Proof of Lemma \ref{lem_moment}}

Fix any information structure $\pi$.
Taking unconditional expectations of both sides of \eqref{FOC} and applying the law of iterated expectations yields \eqref{moment1}.
Subtracting \eqref{moment1} from \eqref{FOC}, we obtain
\[
\sum_{j \in [n]} q_{ij} \bigl(\E[\sigma_j^\pi \mid \eta_i] - \bar{a}_j \bigr)
= \sum_{k \in [m]} r_{ik} \bigl(\E[\theta_k \mid \eta_i] - \bar{\theta}_k \bigr).
\]
Multiplying both sides of this equation by $(\sigma_i-\E[\sigma_i])$ and then taking unconditional expectations, we obtain \eqref{moment2}, again by the law of iterated expectations.\footnote{To see this more explicitly, note that $\sigma^\pi_i$ is $\eta_i$-measurable. Hence, by the law of iterated expectations,
\[
\E\bigl[\bigl(\sigma^\pi_i - \bar{a}_i \bigr) \bigl(\E[\sigma_j^\pi \mid \eta_i] - \bar{a}_j \bigr) \bigr] = \E \bigl[ \E\bigl[ \bigl(\sigma^\pi_i - \bar{a}_i \bigr) \bigl(\sigma_j^\pi - \bar{a}_j \bigr) \mid \eta_i \bigr] \bigr] = \E\bigl[\bigl(\sigma^\pi_i - \bar{a}_i \bigr) \bigl(\sigma_j^\pi - \bar{a}_j \bigr) \bigr]
= \Cov \bigl[\sigma^\pi_i, \sigma^\pi_j \bigr].
\]
Similarly, we can see that $\E[(\sigma^\pi_i - \bar{a}_i \bigr) (\E[\theta_k \mid \eta_i] - \bar{\theta}_k)] = \Cov [\sigma^\pi_i, \theta_k]$.}

Conversely, fix any feasible pair $(X,Y)$.
Consider a Gaussian information structure $\pi$ such that each agent $i$ receives a one-dimensional Gaussian signal $\eta_i$ satisfying $\E\qty[\bm{\eta}]=\bar{\bm{a}}$, $\Var\qty[\bm{\eta}] = X$, and $\Cov\qty[\bm{\eta}, \bm{\theta}] = Y$.
Such an information structure exists since $\mathbf{M}_{X,Y} \succeq O$.
By $\diag(XQ^\top) = \diag(YR^\top)$, the following holds for all $i \in [n]$:
\[
\sum_{j\in[n]}q_{ij}\Cov[\eta_i, \eta_j]=\sum_{k\in[m]} r_{ik}\Cov[\eta_i, \theta_{k}].
\]
Assume for now that $\Var\qty[\eta_i] > 0$.
Multiplying both sides by $(\eta_i-\bar{a}_i)/\Var[\eta_i]$, we obtain
\[
\sum_{j\in[n]}  q_{ij} \qty(\frac{\Cov[\eta_i, \eta_j]}{\Var[\eta_i]} \cdot \qty(\eta_i-\bar{a}_i)) = \sum_{k\in[m]} r_{ik} \qty(\frac{\Cov[\eta_i, \theta_{k}]}{\Var[\eta_i]} \cdot \qty(\eta_i-\bar{a}_i)).
\]
By Lemma \ref{lem_cond_normal}, this is equivalent to
\[
\sum_{j\in[n]}  q_{ij} \qty(\E[\eta_j \mid \eta_i]-\bar{a}_j)=\sum_{k\in[m]} r_{ik}\qty(\E[\theta_{k} \mid \eta_i]-\bar{\theta}_k).
\]
By the construction of $\bar{\bm{a}}$, this simplifies to
\begin{equation}
\sum_{j \in [n]}  q_{ij}\E[\eta_j \mid \eta_i]=\sum_{k \in [m]} r_{ik}\E[\theta_{k} \mid \eta_i]. \label{FOC in proof}
\end{equation}
Note that \eqref{FOC in proof} also holds when $\Var\qty[\eta_i] = 0$, as in that case, we have $\E[\eta_j \mid \eta_i] = \bar{a}_j$ and $\E[\theta_{k} \mid \eta_i] = \bar{\theta}_k$.
Now, let $\bm{\sigma}$ be the strategy profile with $\sigma_i(\eta_i)=\eta_i$ for all $i \in [n]$.
Then \eqref{FOC in proof} implies that $\sigma$ forms a BNE under $\pi$.
Moreover, we have $\mathbf{M}^{\pi} = \mathbf{M}_{X,Y}$ by construction.
\hfill {\it Q.E.D.}


\subsection*{Proof of Lemma \ref{lem_obj}}

Fix any information structure $\pi$ and any realization of $\bm{\sigma}^\pi$ and $\bm{\theta}$.
Since the trace operator is invariant under cyclic permutations, we have
\[
\mqty[\bm{\sigma}^\pi- \bar{\bm{a}} \\ \bm{\theta} -\bar{\bm{\theta}}]^\top \mathbf{V} \mqty[\bm{\sigma}^\pi- \bar{\bm{a}} \\ \bm{\theta} -\bar{\bm{\theta}}]
= \tr(\mqty[\bm{\sigma}^\pi- \bar{\bm{a}} \\ \bm{\theta} -\bar{\bm{\theta}}]^\top \mathbf{V} \mqty[\bm{\sigma}^\pi- \bar{\bm{a}} \\ \bm{\theta} -\bar{\bm{\theta}}])
= \tr(\mathbf{V} \mqty[\bm{\sigma}^\pi- \bar{\bm{a}} \\ \bm{\theta} -\bar{\bm{\theta}}] \mqty[\bm{\sigma}^\pi- \bar{\bm{a}} \\ \bm{\theta} -\bar{\bm{\theta}}]^\top).
\]
By linearity of the trace operator, taking expectations yields
\[
\E \qty[\mqty[\bm{\sigma}^\pi- \bar{\bm{a}} \\ \bm{\theta} -\bar{\bm{\theta}}]^\top \mathbf{V} \mqty[\bm{\sigma}^\pi- \bar{\bm{a}} \\ \bm{\theta} -\bar{\bm{\theta}}]]
= \tr \qty(\mathbf{V} \mathbf{M}^\pi) = \mathbf{V} \bullet \mathbf{M}^\pi.
\]
Hence, $\E \qty[v(\bm{\sigma}^\pi, \bm{\theta})]$ is given as in the lemma by the specification of $v$ as in \eqref{objective}.
\hfill {\it Q.E.D.}


\subsection*{Proof of Proposition \ref{prop_dual}}

To apply Lemma~\ref{lem_dual1}, we first reformulate the problem \eqref{primal_sdp} in the canonical form \eqref{app_p}.
For $i,j \in [n+m]$, let $E_{ij}$ be the square matrix of size $n+m$ whose $(i,j)$-entry is $1$ and all others are $0$.
For each $i \in [n]$, define $\Phi_i \in \calS^{n+m}$ by
\[
\Phi_i = \qty(\sum_{j \in [n]} q_{ij}\frac{E_{ij} + E_{ji}}{2}) - \qty(\sum_{k\in[m]}r_{ik}\frac{E_{i,n+k} + E_{n+k,i}}{2}).
\]
Also, for each $k,k' \in [m]$, define $\Psi_{kk'} \in \calS^{n+m}$ by
\[
\Psi_{kk'} = \frac{E_{n+k,n+k'} + E_{n+k',n+k}}{2}.
\]

Now, consider any symmetric block matrix $\mathbf{M}$ of the form
\begin{equation*}
\mathbf{M} = \mqty[X & Y \\ Y^\top & Z], \quad \text{where} \quad X \in \calS^n, \quad Y \in \calM^{n,m}, \quad Z \in \calS^m.
\end{equation*}
Since $X$ is symmetric, we have $\Phi_i \bullet \mathbf{M} = 0$ if and only if $\sum_{j \in [n]} q_{ij} x_{ij} = \sum_{k\in[m]} r_{ik} y_{ik}$.
Also, since $Z$ is symmetric, $\Psi_{kk'} \bullet \mathbf{M} = \Cov[\theta_k, \theta_{k'}]$ if and only if $z_{kk'} = \Cov[\theta_k, \theta_k']$.
Hence, the problem \eqref{primal_sdp} can be rewritten as follows:
\begin{alignat*}{2}
v^{\rm p} = &\ \max_{\mathbf{M} \in \calS^{n+m}} \quad & &\mathbf{V} \bullet \mathbf{M} \\
&\text{subject to} \quad & &\Phi_i \bullet \mathbf{M} = 0, \quad \forall i \in [n], \\
& & &\Psi_{kk'} \bullet \mathbf{M} = \Cov\qty[\theta_k, \theta_{k'}], \quad \forall k,k' \in [m], \\
& & & \mathbf{M} \succeq O.
\end{alignat*}
This formulation is in the canonical SDP form.
Note that the problem is feasible since all constraints are satisfied by letting $X = O$, $Y = O$, and $Z = \Var[\bm\theta]$.

Let $\lambda_i$ denote the Lagrangian multiplier for each equality constraint $\Phi_i \bullet \mathbf{M} = 0$, and let $\gamma_{kk'}$ denote the one for $\Psi_{kk'} \bullet \mathbf{M} = \Cov[\theta_k,\theta_{k'}]$.
Define $\Lambda = \Diag (\lambda_1,\ldots,\lambda_n)$ and $\Gamma = [\gamma_{kk'}]_{m \times m}$.
Following \eqref{app_d}, the dual problem can be formulated as follows:
\begin{alignat*}{2}
v^{\rm d} = &\min_{\Lambda \in \calS^n_{\rm diag},\, \Gamma \in \calM^{m \times m}}  \quad & &\Var \qty[\bm \theta] \bullet \Gamma \\
&\quad \text{subject to} & &\sum_{i \in [n]} \lambda_i \Phi_i + \sum_{k,k' \in [m]} \gamma_{kk'} \Psi_{kk'} \succeq \mathbf{V}.
\end{alignat*}
From the definitions of $\Phi_i$ and $\Psi_{kk'}$, we obtain
\[
\sum_{i \in [n]} \lambda_i \Phi_i + \sum_{k,k' \in [m]} \gamma_{kk'} \Psi_{kk'}
= \mqty[ \qty(\Lambda Q + Q^\top \Lambda)/2 & - \Lambda R/2 \\ - R^\top \Lambda/2 & \qty(\Gamma+\Gamma^\top)/2].
\]
Since $v^{\rm d}$ depends on $\Gamma$ only through its symmetric part, it is without loss of optimality to restrict $\Gamma$ to be symmetric.
With this restriction, the above dual problem coincides with \eqref{dual_sdp}.

We show that the dual problem is strictly feasible.
Take $\Lambda = aI$ and $\Gamma = ab I$ for some positive numbers $a,b > 0$.
Then,
\[
\mqty[ \qty(\Lambda Q + Q^\top \Lambda)/2 & - \Lambda R/2 \\ - R^\top \Lambda/2 & \Gamma] = a \underbrace{\mqty[(Q+Q^\top)/2 & -R/2 \\ -R^\top/2 & bI]}_{= \ \mathbf{B}}.
\]
The matrix $\mathbf{B}$ is positive definite if and only if
\[
\frac{Q+Q^\top}{2} - \frac{RR^\top}{4b} \succ O.
\]
Since $Q+Q^\top \succ O$, this inequality holds for sufficiently large $b$.
Moreover, fixing such $b$, we have $a \mathbf{B} \succ \mathbf{V}$ for sufficiently large $a$.
Thus, by Lemma \ref{lem_dual1} it follows that $v^{\rm p} = v^{\rm d}$, and both problems are attainable.
\hfill {\it Q.E.D.}


\subsection*{Proof of Theorem~\ref{thm_main}}

First, given any information structure $\pi$, let $X = \Var [\bm{\sigma}^\pi]$ and $Y = \Cov [\bm{\sigma}^\pi, \bm{\theta}]$.
By Proposition~\ref{prop_dual} and Lemma~\ref{lem_dual1}, the primal feasible pair $(X,Y)$ attains $v^{\rm p}$---and thus $\pi$ is optimal---if and only if $\mathbf{M}_{\Lambda} \bullet \mathbf{M}_{X,Y} = 0$ holds for some dual feasible $\Lambda$.
Moreover, since both $\mathbf{M}_{X,Y}$ and $\mathbf{M}_{\Lambda}$ are positive semidefinite, the condition $\mathbf{M}_{\Lambda} \bullet \mathbf{M}_{X,Y} = 0$ is equivalent to $\mathbf{M}_{\Lambda} \mathbf{M}_{X,Y} = O$, which expands into the following system of matrix equations:
\begin{equation} \label{dual_cond1}
A_\Lambda X = B_\Lambda Y^\top, \quad
A_\Lambda Y = B_\Lambda Z, \quad
B_\Lambda^\top X = C_\Lambda Y^\top, \quad
B_\Lambda^\top Y = C_\Lambda Z.
\end{equation}
In this system, the first two equations are precisely \eqref{cs}.
Thus, it remains to show that the last two equations are redundant given the first two.

Consider any primal feasible $(X,Y)$ and dual feasible $\Lambda$ satisfying the first two equations in \eqref{dual_cond1}.
Since $C_\Lambda = B_\Lambda^\top A_\Lambda^+ B_\Lambda$ by definition, we obtain
\[
B_\Lambda^\top X - C_\Lambda Y^\top
= B_\Lambda^\top X - B_\Lambda^\top A_\Lambda^+ B_\Lambda Y^\top
= \qty(B_\Lambda^\top - B_\Lambda^\top A_\Lambda^+ A_\Lambda) X = O,
\]
where the second equality follows from $A_\Lambda X = B_\Lambda Y^\top$, and the third from $\mathbf{M}_\Lambda \succeq O$ together with Lemma~\ref{lem_block}.
Similarly,
\[
B_\Lambda^\top Y - C_\Lambda Z
= B_\Lambda^\top Y - B_\Lambda^\top A_\Lambda^+ B_\Lambda Z
= \qty(B_\Lambda^\top - B_\Lambda^\top A_\Lambda^+ A_\Lambda) Y = O,
\]
where the second equality follows from $A_\Lambda Y = B_\Lambda Z$, and the third again from $\mathbf{M}_\Lambda \succeq O$ and Lemma~\ref{lem_block}.
Thus, the last two equations in \eqref{dual_cond1} are automatically satisfied.

Now, consider any Gaussian information structure $\pi$.
Suppose that $\pi$ satisfies \eqref{opt} with some dual feasible diagonal matrix $\Lambda$.
Multiplying both sides of \eqref{opt} by $\qty(\bm{\sigma}^\pi -\bar{\bm{a}})^\top$ from the right and taking unconditional expectations, we obtain
\[
A_\Lambda \qty(\bm{\sigma}^\pi -\bar{\bm{a}}) \qty(\bm{\sigma}^\pi -\bar{\bm{a}})^\top = B_\Lambda \qty(\bm{\theta} -\bar{\bm{\theta}}) \qty(\bm{\sigma}^\pi -\bar{\bm{a}})^\top
\quad \Longrightarrow \quad
A_\Lambda \Var \qty[\bm{\sigma}^\pi] = B_\Lambda \Cov \qty[\bm{\theta}, \bm{\sigma}^\pi].
\]
Similarly, multiplying both sides of \eqref{opt} by $\qty(\bm{\theta} -\bar{\bm{\theta}})^\top$ from the right and taking unconditional expectations, we have
\[
A_\Lambda \qty(\bm{\sigma}^\pi -\bar{\bm{a}}) \qty(\bm{\theta} -\bar{\bm{\theta}})^\top = B_\Lambda \qty(\bm{\theta} -\bar{\bm{\theta}}) \qty(\bm{\theta} -\bar{\bm{\theta}})^\top
\quad \Longrightarrow \quad
A_\Lambda \Cov \qty[\bm{\sigma}^\pi, \bm{\theta}] = B_\Lambda Z.
\]
Hence, the first part of Theorem~\ref{thm_main} implies that $\pi$ is optimal.

To prove the converse, we show that \eqref{cs} can be equivalently stated as follows.

\begin{lemma} \label{lem_cs}
For any primal feasible pair $(X,Y)$ and dual feasible $\Lambda$, the following three conditions are equivalent:
\begin{enumerate}[\rm (CS1).]
\item \label{cs1}
$A_\Lambda X = B_\Lambda Y^\top$ and $A_\Lambda Y = B_\Lambda Z$.
\item \label{cs2}
$A_\Lambda X = B_\Lambda Y^\top$ and $B_\Lambda \qty(Z-Y^\top X^+Y) = O$.
\item \label{cs3}
$A_\Lambda \qty(X-YZ^+Y^\top) = O$ and $A_\Lambda Y = B_\Lambda Z$.
\end{enumerate}
\end{lemma}

\begin{proof}
Suppose $A_\Lambda X = B_\Lambda Y^\top$.
Since $Y=XX^+Y$ by Lemma~\ref{lem_block}, we have
\[
A_\Lambda Y = A_\Lambda XX^+Y = B_\Lambda Y^\top X^+ Y,
\]
which implies $A_\Lambda Y = B_\Lambda Z$ if and only if $B_\Lambda Z = B_\Lambda Y^\top X^+ Y$.
This establishes the equivalence between (CS1) and (CS2).

Similarly, suppose $A_\Lambda Y = B_\Lambda Z$.
Since $Y^\top = ZZ^+ Y^\top$, we have
\[
B_\Lambda Y^\top = B_\Lambda ZZ^+Y^\top = A_\Lambda YZ^+ Y^\top,
\]
which implies $B_\Lambda Y^\top = A_\Lambda X$ if and only if $A_\Lambda X = A_\Lambda YZ^+Y^\top$.
This establishes the equivalence between (CS1) and (CS3).
\end{proof}

Now, suppose that $\pi$ is optimal.
By Lemma~\ref{lem_cs}, this implies that (CS3) is satisfied.
We study the conditional distribution of the linear statistic $\bm{\xi} = A_\Lambda (\bm{\sigma}^\pi - \bar{\bm{a}})$ given $\bm{\theta}$.
Since $\bm{\xi}$ and $\bm{\theta}$ are jointly normally distributed, Lemma~\ref{lem_cond_normal} implies
\[
\E \qty[\bm{\xi} \mid \bm{\theta}]
= A_\Lambda \E \qty[\bm{\sigma}^\pi - \bar{\bm{a}} \mid \bm{\theta}]
= A_\Lambda YZ^+ \qty(\bm{\theta} - \bar{\bm{\theta}})
= B_\Lambda ZZ^+ \qty(\bm{\theta} - \bar{\bm{\theta}}),
\]
where the last equality follows from the second equation in (CS3).
Moreover, the conditional variance is
\[
\Var \qty[\bm{\xi} \mid \bm{\theta}]
= A_\Lambda \Var \qty[\bm{\sigma}^\pi - \bar{\bm{a}} \mid \bm{\theta}] A_\Lambda
= A_\Lambda \qty(X-YZ^+Y^\top) A_\Lambda
= O,
\]
where the last equality follows from the first equation in (CS3).
Hence, $\bm{\xi} = \E \qty[\bm{\xi} \mid \bm{\theta}]$ a.s., or equivalently,
\[
A_\Lambda (\bm{\sigma}^\pi - \bar{\bm{a}}) = B_\Lambda ZZ^+ \qty(\bm{\theta} - \bar{\bm{\theta}}) \quad {\rm a.s.}
\]
Next, define $\bm{\zeta} = (I-ZZ^+) \bm{\theta}$ and compute its variance:
\begin{align*}
\Var \qty[\bm{\zeta}]
&= (I-ZZ^+) Z (I-ZZ^+) \\
&= Z - ZZ^+Z - ZZZ^+ + ZZ^+ZZZ^+ = Z - Z - ZZZ^+ + ZZZ^+ = O.
\end{align*}
This implies that $\bm{\theta} = ZZ^+ \bm{\theta}$ a.s., and therefore $A_\Lambda (\bm{\sigma}^\pi - \bar{\bm{a}}) = B_\Lambda \qty(\bm{\theta} - \bar{\bm{\theta}})$ a.s.
\hfill {\it Q.E.D.}


\subsection*{Proof of Corollary~\ref{cor_certify}}
As shown in the proof of Theorem~\ref{thm_main}, for any primal feasible $(X,Y)$ and dual feasible $\Lambda$, we have $\mathbf{V} \bullet \mathbf{M}_{X,Y} = Z \bullet C_\Lambda$ if and only if $A_\Lambda X = B_\Lambda Y^\top$ and $A_\Lambda Y = B_\Lambda Z$.
The corollary follows immediately from this equivalence.
\hfill {\it Q.E.D.}


\bibliography{reference}

\begin{thebibliography}{42}
\providecommand{\natexlab}[1]{#1}
\providecommand{\url}[1]{\texttt{#1}}
\expandafter\ifx\csname urlstyle\endcsname\relax
  \providecommand{\doi}[1]{doi: #1}\else
  \providecommand{\doi}{doi: \begingroup \urlstyle{rm}\Url}\fi

\bibitem[Angeletos and Pavan(2007)]{angeletospavan2007}
G.-M. Angeletos and A.~Pavan.
\newblock {Efficient Use of Information and Social Value of Information}.
\newblock \emph{Econometrica}, 75\penalty0 (4):\penalty0 1103--1142, 2007.

\bibitem[Arieli and Babichenko(2019)]{arielibabichenko2019}
I.~Arieli and Y.~Babichenko.
\newblock {Private Bayesian Persuasion}.
\newblock \emph{Journal of Economic Theory}, 182:\penalty0 185--217, 2019.

\bibitem[Ballester et~al.(2006)Ballester, Calv{\'o}-Armengol, and Zenou]{ballesteretal2006}
C.~Ballester, A.~Calv{\'o}-Armengol, and Y.~Zenou.
\newblock {Who's Who in Networks. Wanted: The Key Player}.
\newblock \emph{Econometrica}, 74\penalty0 (5):\penalty0 1403--1417, 2006.

\bibitem[Banerjee et~al.(2013)Banerjee, Chandrasekhar, Duflo, and Jackson]{banerjeeetal2013}
A.~Banerjee, A.~G. Chandrasekhar, E.~Duflo, and M.~O. Jackson.
\newblock {The Diffusion of Microfinance}.
\newblock \emph{Science}, 341\penalty0 (6144):\penalty0 1236498, 2013.

\bibitem[Bergemann and Morris(2013)]{bergemannmorris2013}
D.~Bergemann and S.~Morris.
\newblock {Robust Predictions in Games with Incomplete Information}.
\newblock \emph{Econometrica}, 81\penalty0 (4):\penalty0 1251--1308, 2013.

\bibitem[Bergemann and Morris(2016{\natexlab{a}})]{bergemannmorris2016a}
D.~Bergemann and S.~Morris.
\newblock {Bayes Correlated Equilibrium and the Comparison of Information Structures in Games}.
\newblock \emph{Theoretical Economics}, 11\penalty0 (2):\penalty0 487--522, 2016{\natexlab{a}}.

\bibitem[Bergemann and Morris(2016{\natexlab{b}})]{bergemannmorris2016b}
D.~Bergemann and S.~Morris.
\newblock {Information Design, Bayesian Persuasion, and Bayes Correlated Equilibrium}.
\newblock \emph{American Economic Review}, 106\penalty0 (5):\penalty0 586--591, 2016{\natexlab{b}}.

\bibitem[Bergemann and Morris(2019)]{bergemannmorris2019}
D.~Bergemann and S.~Morris.
\newblock {Information Design:\ A Unified Perspective}.
\newblock \emph{Journal of Economic Literature}, 57\penalty0 (1):\penalty0 44--95, 2019.

\bibitem[Bergemann et~al.(2015)Bergemann, Heumann, and Morris]{bhm2015}
D.~Bergemann, T.~Heumann, and S.~Morris.
\newblock {Information and Volatility}.
\newblock \emph{Journal of Economic Theory}, 158:\penalty0 427--465, 2015.

\bibitem[Bergemann et~al.(2017)Bergemann, Heumann, and Morris]{bhm2017}
D.~Bergemann, T.~Heumann, and S.~Morris.
\newblock {Information and Interaction}.
\newblock \emph{Working Paper}, 2017.

\bibitem[Bernstein(2018)]{bernstein2018}
D.~S. Bernstein.
\newblock \emph{Scalar, Vector, and Matrix Mathematics: Theory, Facts, and Formulas}.
\newblock Princeton University Press, 2018.

\bibitem[Bramoull{\'e} et~al.(2014)Bramoull{\'e}, Kranton, and D'amours]{bramoulleetal2014}
Y.~Bramoull{\'e}, R.~Kranton, and M.~D'amours.
\newblock {Strategic Interaction and Networks}.
\newblock \emph{American Economic Review}, 104\penalty0 (3):\penalty0 898--930, 2014.

\bibitem[Calv{\'o}-Armengol et~al.(2015)Calv{\'o}-Armengol, De~Mart{\'\i}, and Prat]{calvoarmengol2015}
A.~Calv{\'o}-Armengol, J.~De~Mart{\'\i}, and A.~Prat.
\newblock {Communication and Influence}.
\newblock \emph{Theoretical Economics}, 10\penalty0 (2):\penalty0 649--690, 2015.

\bibitem[Clarke(1983)]{clarke1983}
R.~N. Clarke.
\newblock {Collusion and the Incentives for Information Sharing}.
\newblock \emph{The Bell Journal of Economics}, 14\penalty0 (2):\penalty0 383--394, 1983.

\bibitem[Colombo et~al.(2014)Colombo, Femminis, and Pavan]{colomboetal2014}
L.~Colombo, G.~Femminis, and A.~Pavan.
\newblock {Information Acquisition and Welfare}.
\newblock \emph{Review of Economic Studies}, 81\penalty0 (4):\penalty0 1438--1483, 2014.

\bibitem[Dessein and Santos(2006)]{dessein2006}
W.~Dessein and T.~Santos.
\newblock {Adaptive Organizations}.
\newblock \emph{Journal of Political Economy}, 114\penalty0 (5):\penalty0 956--995, 2006.

\bibitem[Dewan and Myatt(2008)]{dewanmyatt2008}
T.~Dewan and D.~P. Myatt.
\newblock {The Qualities of Leadership: Direction, Communication, and Obfuscation}.
\newblock \emph{American Political Science Review}, 102\penalty0 (3):\penalty0 351--368, 2008.

\bibitem[Gal-Or(1985)]{galor1985}
E.~Gal-Or.
\newblock {Information Sharing in Oligopoly}.
\newblock \emph{Econometrica}, 53\penalty0 (2):\penalty0 329--343, 1985.

\bibitem[Galeotti et~al.(2020)Galeotti, Golub, and Goyal]{galeottietal2020}
A.~Galeotti, B.~Golub, and S.~Goyal.
\newblock {Targeting Interventions in Networks}.
\newblock \emph{Econometrica}, 88\penalty0 (6):\penalty0 2445--2471, 2020.

\bibitem[Goyal and Moraga-Gonzalez(2001)]{goyal2001}
S.~Goyal and J.~L. Moraga-Gonzalez.
\newblock {R\&D Networks}.
\newblock \emph{RAND Journal of Economics}, 32\penalty0 (4):\penalty0 686--707, 2001.

\bibitem[Hellwig and Veldkamp(2009)]{hellwigveldkamp2009}
C.~Hellwig and L.~Veldkamp.
\newblock {Knowing What Others Know: Coordination Motives in Information Acquisition}.
\newblock \emph{Review of Economic Studies}, 76\penalty0 (1):\penalty0 223--251, 2009.

\bibitem[Hoshino(2022)]{hoshino2022}
T.~Hoshino.
\newblock {Multi-agent Persuasion: Leveraging Strategic Uncertainty}.
\newblock \emph{International Economic Review}, 63\penalty0 (2):\penalty0 755--776, 2022.

\bibitem[Kamenica(2019)]{kamenica2019}
E.~Kamenica.
\newblock {Bayesian Persuasion and Information Design}.
\newblock \emph{Annual Review of Economics}, 11:\penalty0 249--272, 2019.

\bibitem[Kamenica and Gentzkow(2011)]{kamenicagentzkow2011}
E.~Kamenica and M.~Gentzkow.
\newblock {Bayesian Persuasion}.
\newblock \emph{American Economic Review}, 101\penalty0 (6):\penalty0 2590--2615, 2011.

\bibitem[Li et~al.(1987)Li, McKelvey, and Page]{lietal1987}
L.~Li, R.~D. McKelvey, and T.~Page.
\newblock {Optimal Research for Cournot Oligopolists}.
\newblock \emph{Journal of Economic Theory}, 42\penalty0 (1):\penalty0 140--166, 1987.

\bibitem[Mathevet et~al.(2020)Mathevet, Perego, and Taneva]{mathevetetal2020}
L.~Mathevet, J.~Perego, and I.~Taneva.
\newblock {On Information Design in Games}.
\newblock \emph{Journal of Political Economy}, 128\penalty0 (4):\penalty0 1370--1404, 2020.

\bibitem[Miyashita and Ui(2024)]{miyashita_ui_large}
M.~Miyashita and T.~Ui.
\newblock {On the Pettis Integral Approach to Large Population Games}.
\newblock \emph{Working Paper}, 2024.

\bibitem[Morris and Shin(2002)]{morrisshin2002}
S.~Morris and H.~S. Shin.
\newblock {Social Value of Public Information}.
\newblock \emph{American Economic Review}, 92\penalty0 (5):\penalty0 1521--1534, 2002.

\bibitem[Myatt and Wallace(2012)]{myattwallace2012}
D.~P. Myatt and C.~Wallace.
\newblock {Endogenous Information Acquisition in Coordination Games}.
\newblock \emph{Review of Economic Studies}, 79\penalty0 (1):\penalty0 340--374, 2012.

\bibitem[Radner(1962)]{radner1962}
R.~Radner.
\newblock {Team Decision Problems}.
\newblock \emph{The Annals of Mathematical Statistics}, 33\penalty0 (3):\penalty0 857--881, 1962.

\bibitem[Rao(1973)]{rao1973}
C.~R. Rao.
\newblock \emph{{Linear Statistical Inference and its Applications}}.
\newblock Wiley New York, 2 edition, 1973.

\bibitem[Smolin and Yamashita(2023)]{smolinyamashita2023}
A.~Smolin and T.~Yamashita.
\newblock {Information Design in Smooth Games}.
\newblock \emph{Working Paper}, 2023.

\bibitem[Spielman(2012)]{spielman2012}
D.~Spielman.
\newblock {Spectral Graph Theory}.
\newblock \emph{Combinatorial Scientific Computing}, 2012.

\bibitem[Tamura(2018)]{tamura2018}
W.~Tamura.
\newblock {Bayesian Persuasion with Quadratic Preferences}.
\newblock \emph{Working Paper}, 2018.

\bibitem[Taneva(2019)]{taneva2019}
I.~Taneva.
\newblock {Information Design}.
\newblock \emph{American Economic Journal: Microeconomics}, 11\penalty0 (4):\penalty0 151--185, 2019.

\bibitem[Ui(2016)]{ui2016}
T.~Ui.
\newblock {Bayesian Nash Equilibrium and Variational Inequalities}.
\newblock \emph{Journal of Mathematical Economics}, 63:\penalty0 139--146, 2016.

\bibitem[Ui and Yoshizawa(2015)]{uiyoshizawa2015}
T.~Ui and Y.~Yoshizawa.
\newblock {Characterizing Social Value of Information}.
\newblock \emph{Journal of Economic Theory}, 158:\penalty0 507--535, 2015.

\bibitem[Valente(2012)]{valente2012}
T.~W. Valente.
\newblock {Network Interventions}.
\newblock \emph{Science}, 337\penalty0 (6090):\penalty0 49--53, 2012.

\bibitem[Vandenberghe and Boyd(1996)]{vandenbergheboyd1996}
L.~Vandenberghe and S.~Boyd.
\newblock {Semidefinite Programming}.
\newblock \emph{SIAM review}, 38\penalty0 (1):\penalty0 49--95, 1996.

\bibitem[Vives(1984)]{vives1984}
X.~Vives.
\newblock {Duopoly Information Equilibrium: Cournot and Bertrand}.
\newblock \emph{Journal of Economic Theory}, 34\penalty0 (1):\penalty0 71--94, 1984.

\bibitem[Vives(1988)]{vives1988}
X.~Vives.
\newblock {Aggregation of Information in Large Cournot Markets}.
\newblock \emph{Econometrica}, 56\penalty0 (4):\penalty0 851--876, 1988.

\bibitem[Vives(1999)]{vives1999}
X.~Vives.
\newblock \emph{{Oligopoly Pricing: Old Ideas and New Tools}}.
\newblock MIT press, 1999.

\end{thebibliography}


\newpage

\begin{center}
\LARGE{Online Appendix to ``LQG Information Design''}

\bigskip

\large{Masaki Miyashita \qquad Takashi Ui}
\end{center}

\begin{abstract}
This online appendix contains additional proofs and results that do not appear in the main text of ``LQG Information Design.''  
Appendix~\ref{app_omit} provides proofs for results in Sections~\ref{sec_nfsi}, \ref{sec_personal}, and \ref{sec_network} of the main paper.  
Appendix~\ref{app_comp} derives general optimal information structures in complete networks.  
Appendix~\ref{app_pub} characterizes the optimal information structure within the class of public information structures.
\end{abstract}


\setcounter{page}{1}

\renewcommand{\thesection}{B}
\renewcommand{\theequation}{B\arabic{equation}}
\renewcommand{\thelemma}{B\arabic{lemma}}
\renewcommand{\thecorollary}{B\arabic{corollary}}
\renewcommand{\theproposition}{B\arabic{proposition}}

\renewcommand{\theHlemma}{B\arabic{lemma}}
\renewcommand{\theHproposition}{B\arabic{proposition}}
\renewcommand{\theHequation}{B\arabic{equation}}

\setcounter{section}{1}
\setcounter{equation}{0}
\setcounter{lemma}{0}
\setcounter{proposition}{0}

\section{Omitted proofs}
\label{app_omit}

We begin by presenting a few additional lemmas that will be used in the remaining proofs.
The next result shows that when $Z$ is nonsingular, dual feasibility simplifies to a lower-dimensional condition at the optimum.

\begin{lemma} \label{lem_Zfull}
Suppose that $Z \succ O$.
If a diagonal matrix $\Lambda$ satisfies \eqref{cs} with some primal feasible pair $(X,Y)$, then $\mathbf{M}_\Lambda \succeq O$ if and only if $A_\Lambda \succeq O$.
\end{lemma}

\begin{proof}
If $Z$ is invertible, then \eqref{cs} implies that $B_\Lambda = A_\Lambda YZ^{-1}$.
Substituting this into $C_\Lambda$, we have $C_\Lambda = Z^{-1}Y^\top A_\Lambda YZ^{-1}$.
Thus, it follows that
\begin{equation*}
\mathbf{M}_{\Lambda}
= \mqty[A_\Lambda & -A_\Lambda YZ^{-1} \\ -Z^{-1}Y^\top A_\Lambda & Z^{-1}Y^\top A_\Lambda YZ^{-1}] 
= \mqty[I & O \\ O & YZ^{-1}]^{\top} \mqty[A_\Lambda & -A_\Lambda \\ -A_\Lambda & A_\Lambda] \mqty[I & O \\ O & YZ^{-1}],
\end{equation*}
from which we confirm that $\mathbf{M}_{\Lambda} \succeq O$ if and only if $A_\Lambda \succeq O$.
\end{proof}

An immediate consequence of strong duality is that the designer's optimal value $v^{\rm p}$ can be expressed as $C_\Lambda \bullet Z$, with $\Lambda$ being dual optimal.
Moreover, the following representations of $v^{\rm p}$ are also equivalent.

\begin{lemma} \label{lem_value}
For any optimal information structure $\pi$ certified by $\Lambda$, we have
\begin{align*}
v^{\rm p} &= C_\Lambda \bullet Z \\
&= A_\Lambda \bullet \Var \qty[\bm{\sigma}^\pi]
= B_{\Lambda} \bullet \Cov \qty[\bm{\sigma}^\pi, \bm{\theta}] \\
&= A_\Lambda\bullet \Var\qty[\E\qty[\bm{\sigma}^\pi \mid \bm{\theta}]] 
= C_\Lambda \bullet \Var\qty[\E \qty[\bm{\theta} \mid \bm{\sigma}^\pi]].
\end{align*}
\end{lemma}

\begin{proof}
Let $\pi$ be an optimal information structure, which we can assume Gaussian by Proposition~\ref{prop_gauss}.
By Proposition~\ref{prop_dual}, we have $v^{\rm p} = C_\Lambda \bullet Z$, where $\Lambda$ certifies $\pi$.
Moreover, by letting $X=\Var\qty[\bm{\sigma}^\pi]$ and $Y=\Cov\qty[\bm{\sigma}^\pi, \bm{\theta}]$, it is straightforward to verify that $C_\Lambda \bullet Z = B_\Lambda \bullet Y = A_\Lambda \bullet X$ by using $A_\Lambda X = B_\Lambda Y^\top$, $A_\Lambda Y = B_\Lambda Z$, and $C_\Lambda = B_\Lambda^\top A_\Lambda^+ B_\Lambda$.
Furthermore, since $\pi$ is Gaussian, Lemma~\ref{lem_cond_normal} implies that $\Var \qty[ \E \qty[\bm{\sigma}^\pi \mid \bm{\theta}]] = YZ^+Y^\top$.
Hence, we have
\[
A_\Lambda \bullet \Var \qty[\E \qty[\sigma^\pi \mid \theta]]
= \tr \qty(A_\Lambda YZ^+Y^\top)
= \tr \qty(B_\Lambda ZZ^+Y^\top)
= \tr \qty(B_\Lambda Y^\top) = B_\Lambda \bullet Y,
\]
where the third equality follows from $\mathbf{M}_{X,Y} \succeq O$ and Lemma~\ref{lem_block}.
Similarly, we have $\Var \qty[ \E \qty[\bm{\theta} \mid \bm{\sigma}^\pi]] = Y^\top X^+Y$, and thus,
\begin{align*}
&C_\Lambda \bullet \Var \qty[ \E \qty[\bm{\theta} \mid \bm{\sigma}^\pi]]
= \tr \qty(C_\Lambda Y^\top X^+ Y) \\
&\quad = \tr\qty(B_\Lambda^\top A_\Lambda^+ B_\Lambda Y^\top X^+ Y)
= \tr\qty(B_\Lambda^\top A_\Lambda^+ A_\Lambda XX^+ Y)
= \tr \qty(B_\Lambda^\top Y) = B_\Lambda \bullet Y,
\end{align*}
where the fourth equality follows from $\mathbf{M}_\Lambda \succeq O$, $\mathbf{M}_{X,Y} \succeq O$, and Lemma~\ref{lem_block}.
\end{proof}

The next lemma calculates the value of an information structure when the designer aims to maximize agents' welfare, thereby completing omitted formality in Example~\ref{ex_welfare}.

\begin{lemma} \label{lem_obj_welfare}
Let $v$ be specified as in \eqref{eq_obj_welfare}.
Then, for any information structure $\pi \in \Pi$, we have
\[
\E \qty[v(\bm{\sigma}^\pi, \hat{\bm{\theta}})] = \frac{1}{2} \sum_{i \in [n]} v_i q_{ii} \cdot \qty(\Var \qty[\sigma^\pi_i] + \bar{a}_i^2).
\]
\end{lemma}

\begin{proof}
Let $u_i$ be given as in \eqref{payoff_personal}.
By the law of iterated expectations, we have
\[
\E \qty[u_i(\bm{\sigma}^\pi_i, \hat{\theta}_i)]
= \E\Bigl[ \underbrace{\E \qty[\hat{\theta}_i - \textstyle \sum_{j \neq i} q_{ij} \sigma^\pi_j \mid \sigma^\pi_i]}_{(*)} \cdot \sigma^\pi_i - \frac{q_{ii}}{2} \cdot (\sigma^\pi_i)^2 \Bigr],
\]
where $(*)$ must equal to $q_{ii} \sigma^\pi_i$ by the best-response formula \eqref{FOC}.
Thus, $\E[u_i(\bm{\sigma}^\pi_i, \hat{\theta}_i)] = q_{ii}/2 \cdot \E|\sigma^\pi_i|^2$.
Aggregating this across all agents, since $\E|\sigma^\pi_i|^2 = \Var[\sigma^\pi_i] + \bar{a}_i^2$, we obtain the desired expression for the designer's objective.
\end{proof}


\subsection*{Proof of Corollary~\ref{cor_nfsi1}}

Only the uniqueness assertion in part~\eqref{nf1} requires proof, as the remaining claims follow directly from Theorem~\ref{thm_main}.
Assume there exists a dual optimal $\Lambda$ such that $A_\Lambda \succ O$.
Consider any primal solutions, $(X,Y)$ and $(\tilde{X},\tilde{Y})$.
By Corollary~\ref{cor_certify}, $\Lambda$ certifies both pairs, so the following equalities hold:
\[
A_\Lambda X = B_\Lambda Y^\top, \quad
A_\Lambda Y = B_\Lambda Z, \quad
A_\Lambda \tilde{X} = B_\Lambda \tilde{Y}^\top, \quad
A_\Lambda \tilde{Y} = B_\Lambda Z.
\]
From the second and forth equations, we have $A_\Lambda (Y-\tilde{Y}) = O$, from which $Y=\tilde{Y}$ since $A_\Lambda$ is non-singular.
Moreover, substituting $Y=\tilde{Y}$ into the first and third equations, we find $A_\Lambda (X-\tilde{X}) = O$, which similarly yields $X=\tilde{X}$ due to the non-singularity of $A_\Lambda$.
\hfill {\it Q.E.D.}


\subsection*{Proof of Corollary~\ref{cor_nfsi2}}

By Lemma~\ref{lem_cond_normal}, a Gaussian information structure is noise-free if and only if $X-YZ^+Y^\top = O$.
Also, it is state-identifiable if and only if $Z-Y^\top X Y = O$.
By these observations, Corollary~\ref{cor_nfsi2} follows immediately from Lemma~\ref{lem_cs}.


\subsection*{Proof of Corollary~\ref{cor_no_full}}

Under $\underline{\pi}$ from Example~\ref{ex_no}, agents' equilibrium strategies remain constant at $\bar{\bm{a}} = 0$.
Thus, \eqref{opt} is satisfied for $\underline{\pi}$ if and only if $B_\Lambda(\bm{\theta} - \bar{\bm{\theta}}) = \bm{0}$, or equivalently, $\Var[B_\Lambda(\bm{\theta} - \bar{\bm{\theta}})] = B_\Lambda Z B_\Lambda^\top = O$.
In contrast, under $\overline{\pi}$ from Example~\ref{ex_full}, the equilibrium strategy profile is given by $Q^{-1}R \bm{\theta}$.
Substituting this into \eqref{opt}, we have $(A_\Lambda Q^{-1}R - B_\Lambda) (\bm{\theta} - \bar{\bm{\theta}}) = \bm{0}$,
 or equivalently, $\Var[(A_\Lambda Q^{-1}R - B_\Lambda)(\bm{\theta} - \bar{\bm{\theta}})]  = O$.
 Calculating the variance, the optimality of $\overline{\pi}$ is characterized by the equation \eqref{opt_full}.
\hfill{\it Q.E.D.}


\subsection*{Proof of Corollary \ref{cor_opt_m1}}

Suppose that $m=1$, in which case $Z$ has full rank trivially.
As for the optimality of $\underline{\pi}$, notice that $B_\Lambda = O$ simplifies to $\lambda_i r_i + w_i = 0$ for each $i \in [n]$, thereby  dictating $\Lambda = -\Delta$, where $\Delta = \Diag\qty(w_1/r_1,\ldots,w_n/r_n)$.
Hence, Corollary~\ref{cor_no_full} implies that $\underline{\pi}$ is optimal if and only if $A_{-\Delta} \succeq O$, which is equivalent to $O \succeq \tilde{V}$.
In particular, if $O \succ \tilde{V}$, then $A_{-\Delta} \succ O$ so that Corollary~\ref{cor_nfsi1} implies unique optimality.

Now, suppose that $O \not \succeq \tilde{V}$.
By the previous conclusion, we know that $\underline{\pi}$ is suboptimal, meaning $v^{\rm p} > 0$.
Then, by Lemma~\ref{lem_value}, we must have $B_\Lambda \neq O$ for any $\Lambda$ certifying an optimal information structure.
Since $m=1$, this means that $B_\Lambda$ has full column rank, whence Corollary~\ref{cor_nfsi1} implies that any optimal information structure is state-identifiable.
\hfill{\it Q.E.D.}


\subsection*{Proof of Lemma \ref{lem_personal_W0}}

Suppose that PDC holds with respect to $\Delta = \Diag (\delta_1,\ldots,\delta_n)$, i.e., $w_{ik} = \delta_i r_{ik}$ for all $i \in [n]$ and $k \in [m]$.
Then, for any primal feasible pair $(X,Y)$, it follows that
\begin{align*}
W \bullet Y &= \sum_{i \in [n]} \sum_{k \in [m]} w_{ik} y_{ik}
= \sum_{i \in [n]} \delta_i \sum_{k \in [m]} r_{ik} y_{ik} \\
&= \sum_{i \in [n]} \delta_i \sum_{j \in [n]} q_{ij} x_{ij}
= \sum_{i,j \in [n]} \frac{\delta_i q_{ij} + \delta_j q_{ji}}{2} x_{ij}
= \frac{\Delta Q+Q^\top \Delta}{2} \bullet X,
\end{align*}
where the third equality follows from the obedience condition, and the fourth follows from the symmetry of $X$.
\hfill {\it Q.E.D.}


\subsection*{Proof of Corollary~\ref{cor_personal1}}

First, we prove part~\eqref{cor_personal_part1}.
Corollary~\ref{cor_no_full} implies that $\underline{\pi}$ is optimal if and only if $B_\Lambda Z B_\Lambda^\top = O$ for some dual feasible $\Lambda$, where $B_\Lambda = \Lambda/ 2$ by \eqref{AB_personal_state}.
Since the $(i,i)$-th entry of $\Lambda Z \Lambda$ is $\lambda_i^2 z_{ii}$, we must have $\Lambda = O$ to satisfy $B_\Lambda Z B_\Lambda^\top = O$.
Moreover, substituting $\Lambda = O$ into $A_\Lambda$ given as in \eqref{AB_personal_state}, we see that $A_\Lambda \succeq O$ if and only if $O \succeq \tilde{V}$, in which case $\Lambda = O$ is dual feasible.
Hence, by Corollary~\ref{cor_no_full}, $\underline{\pi}$ is optimal if and only if $O \succeq \tilde{V}$.
In particular, if $O \succ \tilde{V}$, then $A_\Lambda \succ O$ so that Corollary~\ref{cor_nfsi1} implies the unique optimality of $\underline{\pi}$.

Now, suppose that $O \not \succeq \tilde{V}$.
Let $\pi$ be any optimal information structure, certified by $\Lambda$.
Since Theorem~\ref{thm_main} implies that $A_\Lambda (\bm{\sigma}^\pi - \E[\bm{\sigma}^\pi]) = B_\Lambda (\hat{\bm{\theta}} - \E[\hat{\bm{\theta}}])$ a.s., and noting that $B_\Lambda = \Lambda/2$, we can represent $\hat{\theta}_i$ as a function of $\bm{\sigma}^\pi$ whenever $\lambda_i \neq 0$.
By assumption, $\Lambda = O$ is not dual feasible; otherwise, we would have $A_\Lambda = -\tilde{V}$ by \eqref{AB_personal_state}, so dual feasibility would require $-\tilde{V} \succeq O$, a contradiction.
So, there exists some $i$ for which $\lambda_i = 0$, whence $\Var[\hat{\theta}_i \mid \bm{\sigma}^\pi] = 0$ by the previous argument for such $i$.
In particular, since dual feasibility of $\Lambda$ implies that all diagonal entries of $A_\Lambda$ must be nonnegative, \eqref{AB_personal_state} yields $\lambda_i q_{ii} - \tilde{v}_{ii} \ge 0$ for all $i \in [n]$.
As $q_{ii} > 0$, we must have $\lambda_i > 0$ whenever $\tilde{v}_{ii} > 0$, that is, $\hat{\theta}_i$ is identified for such $i$.
In particular, if $\tilde{v}_{ii} > 0$ for all $i$, then $\pi$ must be state-identifiable.
\hfill {\it Q.E.D.}


\subsection*{Proof of Corollary \ref{cor_personal2}}

Suppose $Z \succ O$ and $\tilde{v}_{ii} > 0$ for all $i \in [n]$.
Let $\pi$ be any optimal information structure with $Y = \Cov [\bm{\sigma}^\pi, \hat{\bm{\theta}}]$, certified by $\Lambda$.
Since $A_\Lambda Y = B_\Lambda Z$ must hold by Theorem~\ref{thm_main}, we have
\[
\rank (A_\Lambda) \ge \rank (A_\Lambda Y) = \rank (B_\Lambda Z) = \rank (B_\Lambda),
\]
where the last equality holds by $Z \succ O$.
Since $B_\Lambda = \Lambda/2$, if $\rank (B_\Lambda) < n$, then there exists some $i \in [n]$ such that $\lambda_i = 0$.
By \eqref{AB_personal_state}, this means that the $(i,i)$-th entry of $A_\Lambda$ is equal to $-\tilde{v}_{ii}$, leading to a contradiction to dual feasibility, as we assume $\tilde{v}_{ii} > 0$.
Hence, we must have $\rank (A_\Lambda) = \rank (B_\Lambda) = n$.
By Corollary~\ref{cor_nfsi1}, this implies that the optimal information structure $\pi$ is unique, noise-free, and state-identifiable.

We now turn to verifying the optimality condition for $\overline{\pi}$.
Since $R = I$ and $Z \succ O$, condition~\eqref{opt_full} in Corollary~\ref{cor_no_full} holds if and only if $A_\Lambda Q^{-1} = B_\Lambda$.
Using \eqref{AB_personal_state}, a simple calculation shows that this condition is equivalent to
\begin{equation} \label{cor_personal2_pf1}
\tilde{V} = \frac{\Lambda Q}{2}.
\end{equation}
Examining the $(i,i)$-th entry of this equation yields $\lambda_i = 2\tilde{v}_{ii}/q_{ii}$.
Requiring this to hold for all $i \in [n]$, we see that condition~\eqref{opt_full} uniquely determines $\Lambda$ as $\Lambda = 2 D_{\tilde{V}} D_Q^{-1}$.
Then, substituting this into \eqref{cor_personal2_pf1} yields $\tilde{V} = D_{\tilde{V}} D_Q^{-1} Q$, or equivalently, $D_{\tilde{V}}^{-1} \tilde{V} =  D_Q^{-1} Q$.
Moreover, using $\Lambda = 2 D_{\tilde{V}} D_Q^{-1}$, we have
\[
\Lambda Q = 2 D_{\tilde{V}} D_Q^{-1} Q = 2 D_{\tilde{V}} D_{\tilde{V}}^{-1} \tilde{V} = 2 \tilde{V}.
\]
Also, since diagonal matrices commute and $\tilde{V}$ is symmetric by construction, we have
\[
Q^\top \Lambda = 2 \qty(D_Q^{-1} Q)^\top D_{\tilde{V}} = 2 \qty(D_{\tilde{V}}^{-1} \tilde{V})^\top D_{\tilde{V}}
= 2 \tilde{V} D_{\tilde{V}}^{-1} D_{\tilde{V}} = 2 \tilde{V}.
\]
Hence, it follows that $A_\Lambda = \tilde{V}$.
Since $Z \succ O$, Lemma~\ref{lem_Zfull} implies that $\Lambda$ is dual feasible if and only if $\tilde{V} \succeq O$.
To sum up, we have shown that condition~\eqref{opt_full} is satisfied for some dual feasible $\Lambda$ if and only if $\tilde{V} \succeq O$ and $D_{\tilde{V}}^{-1} \tilde{V} =  D_Q^{-1} Q$, where we must have $\Lambda = 2 D_{\tilde{V}} D_Q^{-1}$ by the second condition.
Corollary~\ref{cor_no_full} then concludes that $\overline{\pi}$ is optimal if and only if these two conditions hold.
\hfill {\it Q.E.D.}

\subsection*{Proof of Lemma \ref{lem_regular}}

Let $M = [m_{ij}]_{n \times n}$ be any $G$-invariant matrix, where $G$ is any transitive network.
Denote by $r_i = \sum_{j \in [n]} m_{ij}$ the $i$-th row-sum.
Then, for any permutation $\tau \in \Aut (G)$,
\[
r_i = \sum_{j \in [n]} m_{ij} = \sum_{j \in [n]} \qty[P_\tau M P_\tau^\top]_{ij} = \sum_{j \in [n]} m_{\tau(i), \tau(j)} = \sum_{j \in [n]} m_{\tau(i), j} = r_{\tau(i)}.
\]
Transitivity of $\Aut(G)$ guarantees that for any $i,j \in [n]$, there exists some $\tau \in \Aut(G)$ with $\tau(i) = j$, hence $r_i = r_j$.
That is, all row‐sums coincide.
Similarly, we can show all column‐sums $c_i = \sum_{j \in [n]} m_{ji}$ are also equal.
Lastly, since $\sum_{i,j \in [n]} m_{ij} = nr = nc$, we conclude that the common row‐sum $r$ equals the common column‐sum $c$.
\hfill {\it Q.E.D.}


\subsection*{Proof of Proposition \ref{prop_sym}}

It should be clear from construction that the solution set to any SDP problem is convex.

\begin{lemma} \label{lem_dual_conv}
Consider the canonical SDP problem \eqref{app_p} and its dual \eqref{app_d}.
If $(X,\bm{y})$ and $(\tilde{X},\tilde{\bm{y}})$ are primal-dual optimal pairs, then so is their convex combination $(\lambda X + (1-\lambda)\tilde{X}, \lambda \bm{y} + (1-\lambda)\tilde{\bm{y}})$ for any $\lambda \in [0,1]$.
\end{lemma}

Now, Let $(X,Y)$ be any solution to \eqref{primal_sdp}, which is not necessarily $G$-invariant.
Given any $\tau \in \Aut \qty(G)$, let us show that $(\tau(X),\tau(Y))$ is primal feasible and solves \eqref{primal_sdp}.
We write $P = P_\tau$ for simplicity.
Note that $P$ is an orthogonal matrix, i.e., $P^\top = P^{-1}$.
Then, since $Q = qG+I$ is $G$-invariant, it follows that
\begin{equation} \label{tau_feasible1}
\tau(X)Q^\top
= PXP^\top \qty(PQP^\top)^\top
= PXQ^\top P^\top.
\end{equation}
Since $(X,Y)$ is primal feasible, we have $\diag\qty(XQ^\top) = \diag\qty(Y)$.
This implies that 
\begin{equation}  \label{tau_feasible2}
\diag\qty(PXQ^\top P^\top) = \diag\qty(PYP^\top) = \diag\qty(\tau(Y)).
\end{equation}
By \eqref{tau_feasible1} and \eqref{tau_feasible2}, we get $\diag(\tau(X)Q^\top) = \diag\qty(\tau(Y))$.
Moreover, since $Z$ is $G$-invariant,
\[
\mathbf{M}_{\tau(X),\tau(Y)} = 
\mqty[\tau(X) & \tau(Y)\\\tau(Y)^\top & Z]
=\mqty[PXP^\top & PYP^\top \\ PY^\top P & PZP^\top]
=\mqty[P & O \\ O & P] \mathbf{M}_{X,Y} \mqty[P & O \\ O & P]^\top
\succeq O.
\]
Therefore, we have shown that $(\tau(X),\tau(Y))$ is primal feasible.
Furthermore, since $V$ and $W$ are $G$-invariant, by the orthogonality of $P$, it follows that
\begin{gather*}
V\bullet X = \tr\qty(VX) = \tr\qty(PVP^\top PXP^\top) = \tr\qty(VPXP^\top) = V \bullet \tau(X), \\
W\bullet Y = \tr\qty(WY^\top) = \tr\qty(PWP^\top PY^\top P^\top) = \tr\qty(WPY^\top P^\top) = W \bullet \tau(Y).
\end{gather*}
from which we confirm that $(\tau(X),\tau(Y))$ provides the same value as $(X,Y)$ to the designer.
Since $(X,Y)$ attains $v^{\rm p}$, so does $(\tau(X),\tau(Y))$.

Now, we take the average of $(\tau(X), \tau(Y))$ over all $ \tau \in \Aut(G)$ as follows:
\[
\bar{X} = \frac{\textstyle \sum_{\tau \in \Aut\qty(G)} \tau(X)}{\qty|\Aut\qty(G)|}
\quad {\rm and} \quad \bar{Y} = \frac{\textstyle \sum_{\tau \in \Aut\qty(G)} \tau(Y)}{\qty|\Aut\qty(G)|}.
\]
By construction, $\bar{X}$ and $\bar{Y}$ are $G$-invariant matrices.
Moreover, since $(\bar{X},\bar{Y})$ is a convex combination of $(\tau(X),\tau(Y))$, where each $(\tau(X),\tau(Y))$ solves \eqref{primal_sdp}, Lemma~\ref{lem_dual_conv} implies that $(\bar{X},\bar{Y})$ solves \eqref{primal_sdp} as well.

Let $(\Lambda,\Gamma)$ be any solution to \eqref{dual_sdp}.
By similar arguments, one can show that $(\tau\qty(\Lambda),\tau\qty(\Gamma))$ is a solution to \eqref{dual_sdp} for any $\tau \in \Aut\qty(G)$, and so is their average:
\[
\bar{\Lambda} = \frac{\textstyle \sum_{\tau \in \Aut\qty(G)} \tau(\Lambda)}{\qty|\Aut\qty(G)|}
\quad {\rm and} \quad \bar{\Gamma} = \frac{\textstyle \sum_{\tau \in \Aut\qty(G)} \tau(\Gamma)}{\qty|\Aut\qty(G)|}.
\]
Thus, we obtain the $G$-invariant diagonal matrix $\bar{\Lambda}$, which certifies $(\bar{X},\bar{Y})$.
\hfill {\it Q.E.D.}


\subsection*{Proof of Corollary \ref{cor_sym_bangbang}}

Proposition~\ref{prop_sym} implies that any optimal information structure can be certified by some scalar matrix $\Lambda = \lambda I$.
For such $\Lambda$, we calculate $A_\Lambda$ and $B_\Lambda$ from \eqref{AB_personal_state} as follows:
\begin{equation*}
A_\Lambda = \frac{\lambda(Q+Q^\top)}{2} - \tilde{V} \quad \text{and} \quad B_\Lambda = \frac{\lambda I}{2}.
\end{equation*}
Observe that $B_\Lambda Z B_\Lambda^\top = \lambda^2 Z/4$, so \eqref{opt_no} is satisfied if and only if $\lambda=0$.
Substituting $\lambda = 0$ into $A_\Lambda = O$, by Corollary~\ref{cor_no_full}, we confirm that $\underline{\pi}$ is optimal if and only if $O \succeq \tilde{V}$.
In particular, if $O \succ \tilde{V}$, then $A_\Lambda \succ O$ so that Corollary~\ref{cor_nfsi1} implies unique optimality.

Next, suppose that $O \not \succeq \tilde{V}$.
By \eqref{cor_sym_bang1}, we know that $\underline{\pi}$ is suboptimal, so $v^{\rm p} > 0$.
Then, Lemma~\ref{lem_value} implies that $B_\Lambda \neq O$, whence $\lambda \neq 0$.
This means that $B_\Lambda$ has full rank, and thus Corollary~\ref{cor_nfsi1} implies that $\pi$ is state-identifiable.
In addition, if $Z \succ O$, then we have $\rank(A_\Lambda) \ge \rank (B_\Lambda)$, as shown in the proof of Corollary~\ref{cor_personal2}.
So, $A_\Lambda$ has full rank.
Applying Corollary~\ref{cor_nfsi1} again, the optimal information structure is unique and noise-free.
\hfill {\it Q.E.D.}


\subsection*{Proof of Proposition~\ref{prop_welfare_regular}}

In this proof, instead of $R=I$, we let $R = \1 \1^\top/n$.
This does not change the problem setting because all agents are concerned about the common state, $\theta \equiv \theta_1 = \cdots = \theta_n$.
We normalize $\E \qty[\theta] = 0$ and $\Var \qty[\theta] = 1$ without loss of generality.
To avoid triviality, we focus on the case with $G \neq O$ so that $\mu_{\rm max} (\frac{G+G^\top}{2}) > 0 > \mu_{\rm min} (\frac{G+G^\top}{2})$.

\begin{lemma}
For any transitive $G \neq O$, it holds that $d(G) = \mu_{\rm max} (\frac{G+G^\top}{2}) \ge -\mu_{\rm min} (\frac{G+G^\top}{2})$, where the inequality is tight if $G$ is undirected and bipartite.
In particular, $d(G)$ equals to the spectral radius of $\frac{G+G^\top}{2}$, and $|\beta| d(G) < 1$ holds under Basic Assumption.
\end{lemma}

\begin{proof}
By the Perron--Frobenius theorem, $\mu_{\rm max} (\frac{G+G^\top}{2})$ equals to the spectral radius of $\frac{G+G^\top}{2}$, and particularly, $ \mu_{\rm max} (\frac{G+G^\top}{2}) \ge -\mu_{\rm min} (\frac{G+G^\top}{2})$ holds.
Our remaining task is to show that $d(G)$ equals to the maximal eigenvalue.

Write $\mu = \mu_{\rm max} (\frac{G+G^\top}{2})$ for simplicity.
By the Courant--Fisher theorem, we have
\[
\mu
= \max_{\bm x \neq \bm 0} \frac{\bm x^\top \qty(\frac{G+G^\top}{2}) \bm x}{\bm x^\top \bm x}
\ge \frac{\bm 1^\top \qty(\frac{G+G^\top}{2}) \bm 1}{\bm 1^\top \bm 1}
= \frac{\bm 1^\top G \bm 1}{n}
= \frac{1}{n} \sum_{i \in[n]} \sum_{j \neq i} g_{ji}
= d(G).
\]
For the converse direction, the Perron--Frobenius theorem implies there exists an eigenvector $\bm x$ with nonnegative coordinates that corresponds to $\mu$.
Choose any $i \in [n]$ such that $x_i \ge x_j$ for all $j \in [n]$.
Since $(\mu, \bm x)$ is an eigenpair of $\frac{G+G^\top}{2}$, we have $(\frac{G+G^\top}{2}) \bm x =  \mu \bm x$.
Moreover, noticing that $x_i > 0$ since $\bm{x} \neq \bm{0}$, the $i$-th coordinate of this equation implies that
\[
\mu = \sum_{j \neq i} \frac{g_{ij}+g_{ji}}{2} \cdot \frac{x_j}{x_i}
\le \sum_{j \neq i} \frac{g_{ij}+g_{ji}}{2} = d(G),
\]
where the last equality holds by Lemma~\ref{lem_regular}.
\end{proof}

With $\Lambda = \lambda I$, we calculate
\[
A_\Lambda = (\lambda-1) I - \lambda \beta \cdot \frac{G+G^\top}{2} \quad {\rm  and} \quad B_{\Lambda} = \frac{\lambda}{2n} \1 \1^\top.
\]
Since the diagonal entries of $A_\Lambda$ are all $\lambda - 1$, a necessary condition for dual feasibility is $\lambda \ge 1$.
The next two lemmas refine this condition.

\begin{lemma} \label{lem_regular_dual1}
The scalar matrix $\lambda I$ is dual feasible if
\begin{equation} \label{dual_regular1}
\lambda > \frac{1}{1- |\beta| d(G)}.
\end{equation}
\end{lemma}

\begin{proof}
By Lemma~\ref{lem_block}, a sufficient condition for dual feasibility is that $A_\Lambda \succ O$.
By the construction of $A_\Lambda$, this holds if $\lambda - 1 - \lambda \beta \mu > 0$ for all eigenvalues $\mu$ of $\frac{G+G^\top}{2}$.
In particular, for any such $\mu$,
\[
\lambda (1-\beta \mu) - 1
\ge \lambda \qty(1 - |\beta| d(G)) - 1.
\]
Since $|\beta| d(G) < 1$, the above right-hand side---and thus, every eigenvalue of $A_\Lambda$---is strictly positive if 
\eqref{dual_regular1} holds.
\end{proof}

\begin{lemma} \label{lem_regular_dual2}
Suppose that $\beta < 0$.
Then, the scalar matrix $\lambda I$ is dual feasible if and only if
\begin{equation} \label{dual_regular2}
\lambda \ge \frac{1}{1-\beta \mu_{\rm min} (\frac{G+G^\top}{2})}.
\end{equation}
\end{lemma}

\begin{proof}
By Lemma~\ref{lem_block}, a necessary and sufficient condition for dual feasibility is that $A_\Lambda \succeq O$ and $B_\Lambda = A_\Lambda A_\Lambda^+ B_\Lambda$.
Note that $\lambda < 1$ leads to the violation of $A_\Lambda \succeq O$.
So, henceforth, we focus on $\lambda \ge 1$,
Then, since $\beta < 0$, the minimum eigenvalue of $A_\Lambda$ is given as $\lambda - 1 - \lambda \beta \mu_{\rm min} (\frac{G+G^\top}{2})$.
Hence, since $1-\beta \mu_{\rm min} (\frac{G+G^\top}{2}) > 0$, we have $A_\Lambda \succeq O$ if and only if \eqref{dual_regular2}.

Next, we show that $B_\Lambda = A_\Lambda A_\Lambda^+ B_\Lambda$ is automatically satisfied under \eqref{dual_regular2}.
Since $A_\Lambda A_\Lambda^+$ is the projector onto the range of $A_\Lambda$, and since $B_\Lambda$ is a nonzero constant matrix, $B_\Lambda = A_\Lambda A_\Lambda^+ B_\Lambda$ is satisfied if $A_\Lambda \bm{x} = \1$ for some $\bm{x} \in \R^n$.
Since $G$ is regular, $G\1 = G^\top \1 = d(G) \1$.
This implies
\[
A_\Lambda \1 = (\lambda - 1) \1 - \lambda \beta \cdot \frac{G\1 + G^\top \1}{2} = \qty(\lambda (1 - \beta d(G)) - 1) \1.
\]
Also, noticing that $\beta < 0$ and $d(G) > 0 > - \mu_{\rm min}(\frac{G+G^\top}{2})$, \eqref{dual_regular2} implies
\[
\lambda (1-\beta d(G)) - 1
\ge \frac{1-\beta d(G)}{1-\beta \mu_{\rm min}(\frac{G+G^\top}{2})} - 1
> 0.
\]
Hence, we have $A_\Lambda \bm{x} = \1$ by taking $\bm{x} = \1/(\lambda (1-\beta d(G)) - 1)$.
\end{proof}

\begin{lemma} \label{lem_not_nf}
If an information structure $\pi$ is Gaussian, $G$-invariant, state-identifiable, and noise-free, then $\pi$ must be full disclosure $\overline{\pi}$.
\end{lemma}

\begin{proof}
Let $X = \Var \qty[\bm{\sigma}^\pi]$ and $\bm{y} = \Var \qty[\bm{\sigma^\pi}, \theta]$.
Since $\pi$ is $G$-invariant and $\Aut(G)$ is transitive, all diagonal entries of $X$ are constant, $x \equiv x_{11} = \cdots = x_{nn}$, as well as all entries of $\bm {y}$ are constant, $y \equiv y_1 = \cdots = y_n$.
Then, since $\pi$ is Gaussian and noise-free, Lemma~\ref{lem_cond_normal} implies
\[
\Var \qty[\sigma^\pi_i \mid \theta] = x - y^2= 0.
\]
Similarly, since $\pi$ is state-identifying,
\[
\Var \qty[\theta \mid \bm{\sigma}^\pi] = 1 - \bm{y}^\top X^+ \bm{y} = 1 - y^2 \1^\top X^+ \1 = 0.
\]
Combining these, we obtain
\begin{equation} \label{regular_sinf_full}
x\1^\top X^+ \1 = 1.
\end{equation}
This equation is satisfied when $X = x\1\1^\top$, corresponding the case where $\pi = \overline{\pi}$.

We now show that \eqref{regular_sinf_full} is satisfied by positive semidefinite $X$ only when $X = x\1\1^\top$.
Fix any $X \succeq O$.
Since $X^+ \succeq O$, $\langle \bm{a}, \bm{b} \rangle = \bm{a} X^+ \bm{b}$ defines an inner product operator on $\R^n$.
Taking $\bm{a} = \1$ and $\bm{b} = X\1$, we have
\[
\langle \bm{a}, \bm{a} \rangle  = \1^\top X^+ \1 \quad {\rm and} \quad
\langle \bm{b}, \bm{b} \rangle  = \1^\top XX^+X \1 = \1^\top X \1.
\]
Moreover, since $\1$ belongs to the range of $X$, as $X$ is $G$-invariant, we have
\[
\langle \bm{a}, \bm{b} \rangle = \1^\top X^+ X \1 = \1^\top \1 = n.
\]
By the Cauchy--Schwartz inequality,
\[
n^2 = |\langle \bm{a}, \bm{b} \rangle|^2 \le \langle \bm{a}, \bm{a} \rangle \cdot \langle \bm{b}, \bm{b} \rangle = \qty(\1^\top X^+ \1) \cdot \qty(\1^\top X \1).
\]
Then, it follows that
\[
x \1^\top X^+ \1 \ge \frac{n^2 x}{\1^\top X \1} = \frac{n^2 x}{\sum_{i,j \in [n]} x_{ij}} \ge \frac{n^2 x}{n^2 x} = 1,
\]
where the last inequality holds because $x \ge x_{ij}$ for all $i,j \in [n]$ by the positive semidefiniteness of $X$.
In particular, the inequality holds with equality if and only if $x = x_{ij}$ for all $i,j \in [n]$, i.e., $X = x\1\1^\top$.
\end{proof}

By Theorem~\ref{thm_main}, an information structure $\pi$ is optimal if and only if there exists a dual feasible choice of $\lambda$ such that
\[
A_\Lambda \mqty[\sigma^\pi_1 \\ \vdots \\ \sigma^\pi_n] = \mqty[\lambda \theta/2 \\ \vdots \\ \lambda \theta/2].
\]
By the construction of $A_\Lambda$, the $i$-th entry of this equation is written as follows:
\begin{equation} \label{regular_welfare_opt}
(\lambda-1) \sigma^\pi_i - \lambda \beta \sum_{j \neq i} \frac{g_{ij} + g_{ji}}{2} \sigma^\pi_j = \frac{\lambda \theta}{2}.
\end{equation}
Importantly, this equation indicates that $\theta$ is expressed as a linear function of the strategies of agent $i$ and her neighbors, i.e., those agents $j$ such that either $g_{ij}$ or $g_{ji}$ is nonzero.
In other words, at the optimum, any agent and her neighbors are collectively informed of the state, as the profile of their strategies identifies the realization of $\theta$.

Now, we examine the optimality of $\overline{\pi}$.
Under $\overline{\pi}$, all agents adopt the same equilibrium strategy given by $\sigma^\pi_i = \theta / (1-\beta d(G))$.
Plugging this into \eqref{regular_welfare_opt}, Lemma~\ref{lem_regular} yields
\[
\quad \qty(\frac{\lambda}{2} - \frac{1}{1-\beta d(G)}) \theta = 0.
\]
Since $\Var \qty[\theta] > 0$, this equation pins down the value of $\lambda$ as follows:
\[
\overline{\lambda} = \frac{2}{1-\beta d(G)}.
\]
Hence, $\overline{\pi}$ is optimal if and only if $\overline{\lambda} I$ is dual feasible.
When $\beta \ge 0$ so that $\beta = |\beta|$, Lemma~\ref{lem_regular_dual1} implies that $\overline{\lambda} I$ is dual feasible.
On the other hand, when $\beta < 0$, Lemma~\ref{lem_regular_dual1} implies $\overline{\lambda} I$ is dual feasible if and only if
\[
\frac{2}{1-\beta d(G)} \ge \frac{1}{1-\beta \mu_{\rm min} (\frac{G+G^\top}{2})}.
\]
This condition is equivalently expressed as \eqref{regular_full_condition}.

Next, we consider the case where the condition \eqref{regular_full_condition} is violated.
In this case, $\overline{\pi}$ is no longer optimal, but there is still an optimal information structure $\pi$ that is Gaussian, $G$-invariant and state-identifiable.
By Lemma~\ref{lem_not_nf}, however, such $\pi$ cannot be noise-free.
Hence, by Corollary~\ref{cor_nfsi1}, the associated matrix $A_\Lambda$ must be singular.
Moreover, since the violation of \eqref{regular_full_condition} implies $\beta < 0$, Lemma~\ref{lem_regular_dual2} pins down the corresponding value of $\lambda$ as follows:
\[
\hat{\lambda} = \frac{1}{1-\beta \mu_{\rm \min} (\frac{G+G^\top}{2})}.
\]

\begin{lemma} \label{lem_regular_xy}
Consider any $G$-invariant optimal information structure $\pi$, certified by $\lambda I$.
Then, $x = \Var[\sigma^\pi_i]$ and $y = \Cov[\sigma^\pi_i, \theta]$ are given as follows:
\[
x = \frac{\lambda y}{2} \quad {\rm and} \quad y = \frac{\lambda}{2\lambda (1 -\beta d(G)) - 2}.
\]
\end{lemma}

\begin{proof}
By taking the covariance between each side of \eqref{regular_welfare_opt} and $\theta$, we have
\[
(\lambda - 1)y - \lambda \beta \sum_{j \neq i} \frac{g_{ij} + g_{ji}}{2} y = \frac{\lambda}{2}.
\]
Noticing that $G$ is regular, arranging the above equation yields $y$ as in the lemma.
Also, by taking the covariance between each side of \eqref{regular_welfare_opt} and $\sigma^\pi_i$, we have
\[
(\lambda-1)x - \lambda \beta \sum_{j \neq i} \frac{g_{ij} + g_{ji}}{2} x_{ij} = \frac{\lambda y}{2}.
\]
Taking the average of this equation across all $ i \in [n]$, we get
\begin{equation} \label{eq_regular_x}
(\lambda-1) x - \frac{\lambda \beta}{2n} \sum_{i \in [n]} \sum_{j \neq i} g_{ij} x_{ij} - \frac{\lambda \beta}{2n} \sum_{i \in [n]} \sum_{j \neq i} g_{ji} x_{ij} = \frac{\lambda y}{2}.
\end{equation}
Notice that $\beta \sum_{ij} g_{ij} x_{ij} = x-y$ holds by primal feasibility.
Moreover, since $x_{ij} = x_{ji}$ for all $i \neq j$,
\[
\beta \sum_{i \in [n]} \sum_{j \neq i} g_{ji} x_{ij}
= \beta \sum_{i \in [n]} \sum_{j \neq i} g_{ji} x_{ji}
= \beta \sum_{j \in [n]} \sum_{i \neq j} g_{ji} x_{ji}
= n(x-y).
\]
Substituting this into \eqref{eq_regular_x}, it follows that $(\lambda - 1) x - \lambda (x-y) = \lambda y/2$, whence $x = \lambda y /2$.
\end{proof}

For the Gaussian optimal information structure $\pi$, Lemma~\ref{lem_regular_xy} implies that
\[
s_i = \frac{y^2}{x} = \frac{2y}{\lambda} = \frac{1}{\lambda (1-\beta d(G)) - 1}.
\]
Writing $d = d(G)$ and $\mu = \mu_{\rm min}(\frac{G+G^\top}{2})$ for simplicity, evaluating the above $s_i$ at $\lambda = \hat{\lambda}$ yields
\[
s_i = \frac{1}{\frac{1-\beta d}{1-\beta \mu} - 1}
= \frac{1-\beta \mu}{(1-\beta d) - (1-\beta \mu)}
= \frac{1-\beta \mu}{- \beta (d - \mu)} 
= \frac{1+|\beta| \mu}{|\beta| (d - \mu)}
=\frac{1/|\beta| + \mu}{d - \mu},
\]
as desired.
\hfill {\it Q.E.D.}


\renewcommand{\thesection}{C}
\renewcommand{\theequation}{C\arabic{equation}}
\renewcommand{\thelemma}{C\arabic{lemma}}
\renewcommand{\thecorollary}{C\arabic{corollary}}
\renewcommand{\theproposition}{C\arabic{proposition}}

\renewcommand{\theHlemma}{C\arabic{lemma}}
\renewcommand{\theHproposition}{C\arabic{proposition}}
\renewcommand{\theHequation}{C\arabic{equation}}

\setcounter{section}{2}
\setcounter{equation}{0}
\setcounter{lemma}{0}
\setcounter{proposition}{0}

\section{Optimal information structures in complete networks}
\label{app_comp}

This Appendix~\ref{app_comp} derives explicit solutions for optimal information structures under a general designer objective, assuming that the underlying network is complete: $K = \1\1^\top-I$.
We assume $n \ge 2$ throughout to avoid triviality.
For notational simplicity, we denote by $\bm{\theta} = (\theta_1,\ldots,\theta_n)$ the vector of agents' personal states.

The variance matrix of agents' personal states is given by $Z = [z_{ij}]_{n \times n}$, where $z_{ii} = 1$ and $z_{ij} = \rho$ for $i \neq j$, with $-1/(n-1) < \rho \le 1$.
This parametric restriction on $\rho$ ensures $Z \succeq O$.
Also, to align with Basic Assumption (see \eqref{q_network}), we focus on the following range of $\beta$:
\[
-1 < \beta < \frac{1}{n-1}.
\]
The designer’s $K$-invariant objective can be represented in a parametric form with parameters $v, w, c \in \R$:
\[
V = \begin{bmatrix}
v & c & \cdots & c \\
c & v & \cdots & c \\
\vdots & \vdots & \ddots & \vdots \\
c & c & \cdots & v
\end{bmatrix} \quad {\rm and} \quad
\quad W = wI.
\]
By Lemma~\ref{lem_personal_W0}, we can without loss of generality assume $w=0$---i.e., the designer's objective is state-independent---since otherwise, letting $\tilde{V} = V-w(I-\beta K)$ and $W = O$ yields an equivalent problem.
The setting described as above is referred to as the \emph{$K$-invariant network environment}.

Applying Corollary~\ref{cor_sym_bangbang}, we see that an optimal information structure either reveals nothing or fully reveals the state to the entire population, depending on the sign of the matrix $V$.
Specifically, $V$ is negative semidefinite if and only if
\begin{equation} \label{no_K}
v \le c \le -\frac{v}{n-1}.
\end{equation}
Thus, $\underline{\pi}$ is optimal if and only if condition \eqref{no_K} holds.
This requires $v \le 0$, reflecting the designer’s preference for suppressing the volatility of individual agents’ actions.
In addition, $c$ is required to have a small absolute value relative to $v$, meaning that the designer receives a small marginal value from the correlation between different agents' actions.
In contrast, if \eqref{no_K} is violated, then any optimal information structure must be state-identifiable, while the precise amount of information revealed to each agent remains ambiguous.

Proposition~\ref{prop_sym} implies that the covariance matrices $\Var[\bm{\sigma}^\pi]$ and $\Cov[\bm{\sigma}^\pi, \bm{\theta}]$ can be taken to be $K$-invariant at the optimum.
This leads us to consider a primal solution $(X,Y)$ that takes the following parametric form, for scalars $(x, y, \rho_x, \rho_y)$:
\begin{equation} \label{cov X Y}
X = \begin{bmatrix}
x & \rho_x x & \cdots & \rho_x x \\
\rho_x x & x & \cdots & \rho_x x \\
\vdots & \vdots & \ddots & \vdots \\
\rho_x x & \rho_x x & \cdots & x
\end{bmatrix} \quad {\rm and} \quad
Y = \begin{bmatrix}
y & \rho_y y & \cdots & \rho_y y \\
\rho_y y & y & \cdots & \rho_y y \\
\vdots & \vdots & \ddots & \vdots \\
\rho_y y & \rho_y y & \cdots & y
\end{bmatrix}.
\end{equation}
To ensure that $(X, Y)$ is primal feasible, these parameters must satisfy the linear constraint
\begin{equation} \label{feasible_K}
x = y + (n-1)\beta \cdot \rho_x x.
\end{equation}
Moreover, Proposition~\ref{prop_sym} implies that certifying Lagrange multipliers can be taken to be uniform across agents: $\lambda \equiv \lambda_1 = \cdots = \lambda_n$.

The optimal information structure is then characterized by solving for the five variables $(x, y, \rho_x, \rho_y, \lambda)$ using the complementary slackness condition, together with the primal feasibility constraint \eqref{feasible_K}.
In what follows, we perfom the detailed calculations, distinguishing between two cases depending on the distribution of agents’ personal payoff states.

\subsection{Common-value case}

We first consider the common-value case, i.e., $\rho = 1$, where $\theta \equiv \theta_1 = \cdots = \theta_n$ occurs with probability $1$.
In this case, we necessarily have $\rho_y = 1$.
By solving for the remaining variables, the next result offers a closed-form expression for optimal information structures.

\begin{proposition} \label{prop_common_K}
Suppose that $\rho = 1$.
\begin{enumerate}[\rm i).]
\item \label{prop_common_K1}
No disclosure $\underline{\pi}$ is optimal if and only if \eqref{no_K} holds.
\item \label{prop_common_K2}
Full disclosure $\overline{\pi}$ is optimal if and only if
\begin{equation} \label{full_common_K}
c \ge - \frac{v}{n-1} \quad {\rm and} \quad c \ge f(\beta) \cdot v, \quad {\rm where} \quad f(\beta) \equiv -\frac{1+(n+1)\beta}{2n-1+(n-1)\beta}.
\end{equation}
Here, $f(\beta)$ satisfies $-1/(n-1) < f(\beta) < 1$.
\item \label{prop_common_K3}
Otherwise, an optimal information structure $\pi$ is partial disclosure, characterized by
\begin{equation*}
x = \frac{(c-v)(\beta v+(2+\beta)c)}{4n(1+\beta)(c+\beta v)^2}, \quad
y = \frac{c-v}{2n(c+\beta v)}, \quad
\rho_x = - \frac{(1+2\beta)v+c}{(n-1)(\beta v+(2+\beta)c)}.
\end{equation*}
Moreover, under such $\pi$, each agent's action recommendation reduces the common state variance as follows:
\begin{equation*}
1-\Var\qty[\theta \mid \sigma^\pi_i] = \frac{(1+\beta)(c-v)}{n(\beta v + (2+\beta)c)} \equiv s(\beta, v,c).
\end{equation*}
\end{enumerate}
\end{proposition}

\begin{figure}[t]
\centering
\includegraphics[width=4in]{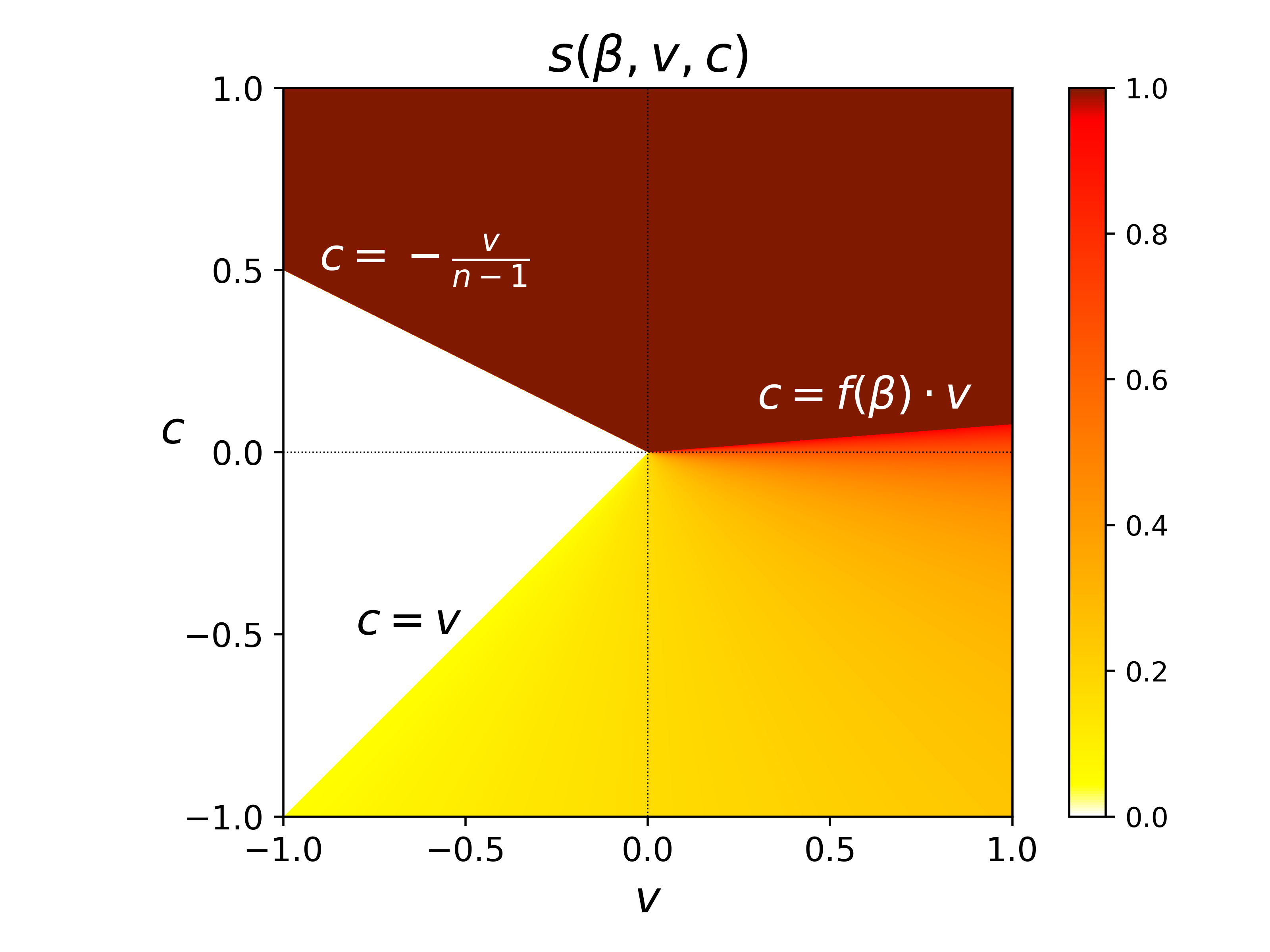} 	
\caption{In the common-value case ($\rho = 1$), an optimal information structure is either no disclosure (white), full disclosure (dark red), or partial disclosure (yellow-red gradient), depending on the values of $v$, $c$, and $\beta$. In the gradient region, the color range expresses the value of $s(\beta,v,c)$. The diagram is obtained by setting $n=3$ and $\beta=-1/3$.}
\label{fig_common_opt}
\end{figure}

The results of Proposition~\ref{prop_common_K} are illustrated in Figure~\ref{fig_common_opt}, which presents a color map with $v$ and $c$ on the horizontal and vertical axes, respectively.
The diagram is divided into three regions according to the linear boundaries identified in the proposition.
The white region, defined by condition \eqref{no_K}, corresponds to values of $(v, c)$ for which $\underline{\pi}$ is optimal.
In contrast, the dark red region, determined by \eqref{full_common_K}, corresponds to the values under which $\overline{\pi}$ is optimal.
The remaining region---shaded in a gradient from yellow to red---indicates where partial disclosure is optimal.

When partial disclosure is optimal, the information structure is state-identifiable, though the informativeness of each agent's signal varies with parameters.
The color intensity in the figure reflects the value of $s(\beta, v, c)$ for a given value of $\beta$, which measures the fraction of the variance in $\theta$ eliminated by conditioning on $\sigma^\pi_i$.
Darker shades correspond to higher values.
The isoquants of this measure appear as rays from the origin, capturing the fact that $s(\beta, v, c)$ is determined by the ratio between $v$ and $c$ when $\beta$ is fixed.
As $c/v$ approaches $f(\beta)$ from $-1$, the informativeness of individual signals increases.

Notably, the optimal information structure is not noise-free in the partial disclosure region.
In this case, a direct calculation shows that the sign of $\Var[\sigma^\pi_i \mid \theta] = x - y^2$ is determined by $(c - v)(c - f(v))$, and it is positive whenever both \eqref{no_K} and \eqref{full_common_K} are violated.
This implies that in the common-value case, the designer’s optimal policy involves adding noise to each agent’s signal individually, in a way that these noise components cancel out in aggregate, thereby ensuring state-identifiability.


Several papers in the literature have studied information design in symmetric common-value LQG games that are similar to the model of this section.
\citet{angeletospavan2007} and \citet{uiyoshizawa2015} consider symmetric Gaussian information structures in which each player receives both conditionally i.i.d.\ private signal and public signal and characterize the optimal combination of them.
\citet{bergemannmorris2013} examine a homogeneous-product Cournot game and derive a symmetric Gaussian information structure that maximizes producer surplus.
However, whether these symmetric Gaussian structures remain optimal when the designer can choose among ``all'' information structures, including asymmetric or non-Gaussian information ones, has not been established.
Our analysis addresses this gap by showing that the symmetric Gaussian structure is robustly optimal: it cannot be outperformed, even when the designer is free to adopt more general disclosure policies.

\subsection{Private-value case}

Next, we consider the private-value case, i.e., $\rho_z < 1$, under which each agent's personal state contains some idiosyncratic component.
The next result obtains as an immediate consequence from Corollaries~\ref{cor_personal2}~and~\ref{cor_sym_bangbang}, and thus the proof is omitted.

\begin{proposition} \label{prop_private_K}
Suppose that $\rho_z < 1$.
\begin{enumerate}[\rm i).]
\item No disclosure $\underline{\pi}$ is optimal if and only if \eqref{no_K} holds.
\item Otherwise, any optimal information structure is state-identifiable and noise-free, while full disclosure is optimal if and only if $v \ge 0$ and $c = -\beta v$.
\end{enumerate}
Moreover, if $(v,c)$ does not lie on the boundary of the region defined by \eqref{no_K}---i.e., if either \begin{inparaenum}[\rm (a)]
\item $v > 0$ or 
\item $c \neq v$ and $c \neq -v/(n-1)$
\end{inparaenum}
---then the optimal information structure is unique.
\end{proposition}

The optimality of $\underline{\pi}$ is supported on the same region defined by \eqref{no_K}, as observed in the common-value case.
On the boundary of this region, there may exist a plethora of optimal information structures other than $\underline{\pi}$, but this is the only situation where such multiplicity can arise.
Outside of the region, a unique optimal information structure exists, and it is both state-identifiable and noise-free.
However, each individual agent typically remains only partially informed, as $\overline{\pi}$ is optimal only when $(v, c)$ lies on the line defined by $c = -\beta v$, with $v \ge 0$.
Put differently, $\overline{\pi}$ is generically suboptimal when agents' personal states contain some idiosyncratic components, regardless of how small these are.

Beyond these general properties, our duality methodology allows us to explicitly compute the variables $(x, y, \rho_x, \rho_y)$ that characterize the optimal information structure.
While the actual computation involves solving a system of polynomial equations and is therefore algebraically messy, we outline the essential steps below, which can be carried out systematically.

First, in the current setting, letting $\Lambda = \lambda I$, the key matrices for optimality are calculated as $A_\Lambda = \lambda(I - \beta G) - V$ and $B_\Lambda = \lambda I/2$.
Substituting these into the optimality condition \eqref{opt} in Theorem~\ref{thm_main}, we obtain
\begin{equation} \label{opt_private_K}
\sigma^\pi_i = \bar{\sigma}^\pi_i + \frac{\lambda}{2} \qty(\alpha^{\rm own} \cdot \qty(\theta_i - \bar{\theta}_i) + \alpha^{\rm opp} \cdot \textstyle \sum_{j \neq i} \qty(\theta_j - \bar{\theta}_j)),
\end{equation}
where $\alpha^{\rm own}$ and $\alpha^{\rm opp}$ denote the diagonal and off-diagonal entries of the inverse matrix $A_\Lambda^{-1}$, respectively.\footnote{Specifically, since $A_\Lambda = (\lambda - v) I - (\lambda \beta + c) K$, the eigenvalues of $A_\Lambda$ are given by $\alpha_1 = (1+\beta)\lambda - (v - c)$ (with multiplicity $n - 1$) and $\alpha_2 = (1 - (n - 1)\beta)\lambda - (v + (n - 1)c)$ (with multiplicity $1$). Since $A_\Lambda \succ 0$ holds at the optimum, both $\alpha_1$ and $\alpha_2$ are positive. It then follows that $A_\Lambda^{-1}$ has a constant diagonal entry $\alpha^{\rm own} = \frac{\alpha_1 + (n - 1)\alpha_2}{n \alpha_1 \alpha_2}$ and off-diagonal entry $\alpha^{\rm off} = \frac{\alpha_1 - \alpha_2}{n \alpha_1 \alpha_2}$.}

Next, by computing the covariance of both sides of equation~\eqref{opt_private_K} with $\sigma^\pi_i$, $\sigma^\pi_j$, $\theta_i$, and $\theta_j$ (for $i \neq j$), we obtain four equations that must be satisfied by the variables $(x, y, \rho_x, \rho_y)$ along with the dual variable $\lambda$.
Together with the obedience condition \eqref{feasible_K}, this yields a system of five equations in five unknowns, which can be solved in a systematic manner.

Crucially, \eqref{opt_private_K} shows that each agent’s equilibrium strategy is given as a linear function of $\bm{\theta}$.
As long as the optimal information structure differs from $\underline{\pi}$---i.e., $\lambda \neq 0$---the signal each agent receives typically mixes information about her own personal state and those of her opponents.
Nonetheless, the overall signal profile is constructed to eliminate fundamental uncertainty in the aggregate.
This method of implementing aggregative full disclosure is qualitatively different from the common-value case, where action recommendations are generated by adding noise terms to the common state, which offset one another in aggregate.

\begin{figure}[t]
\centering
\includegraphics[width=5in]{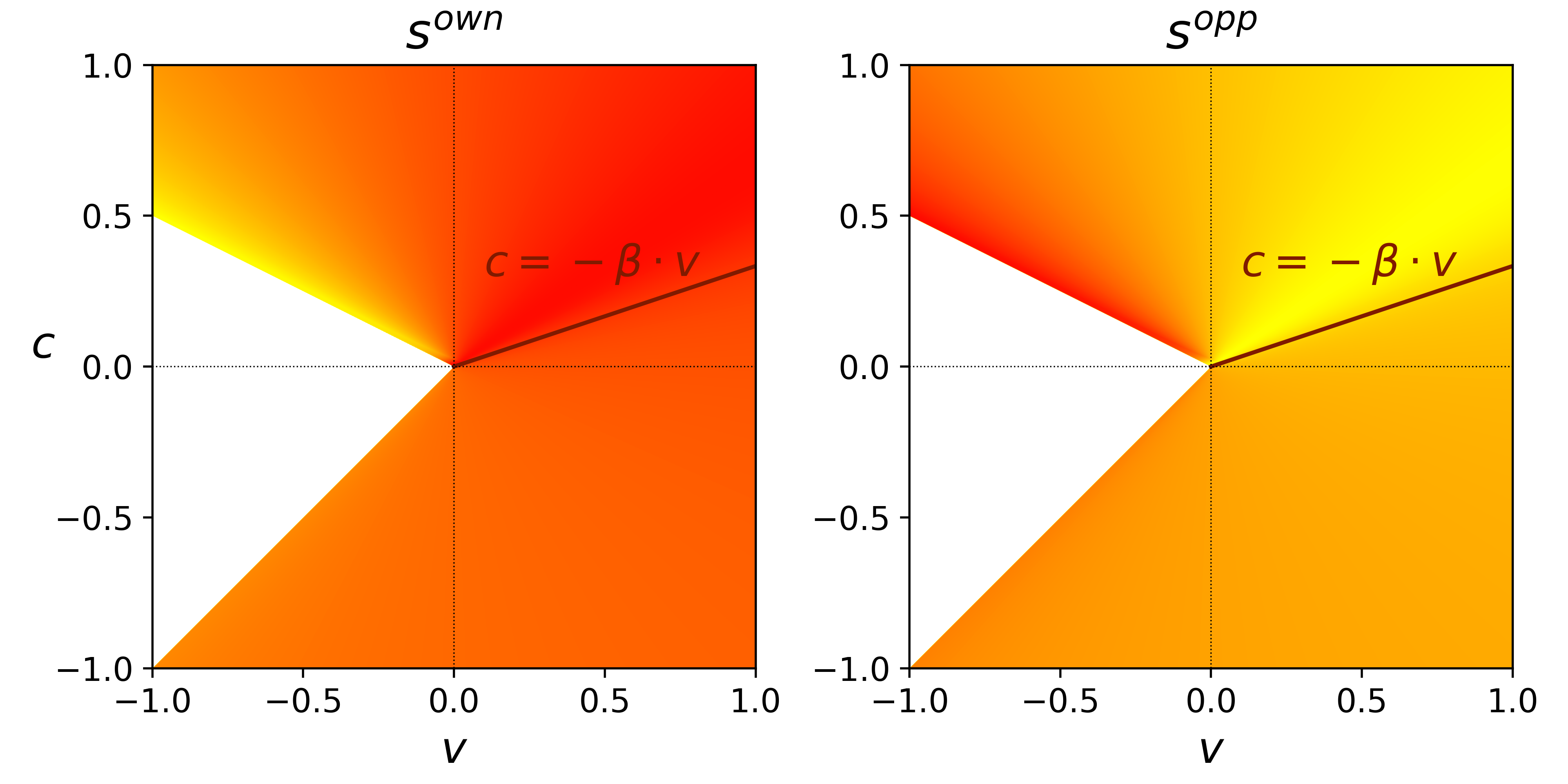} 	
\caption{In the private-value case ($\rho_z < 1$), optimal information structures are visualized as a function of $(v,c)$. The diagrams are obtained by setting $n=3$, $\beta=-1/3$, and $\rho = 0$.}
\label{fig_private_opt}
\end{figure}

In Figure~\ref{fig_private_opt}, we graphically describe optimal $\pi$ in the pure private-value case with $\rho=0$ by computing the reduction of agent's own state variance, as well as that of their opponent's state variance, defined as follows:
\begin{equation*}
s^{\rm own} \equiv 1-\Var\qty[\theta_i \mid \sigma^\pi_i]
\quad {\rm and} \quad
s^{\rm opp} \equiv 1-\Var\qty[\theta_j \mid \sigma^\pi_i],
\quad i \neq j.
\end{equation*}
The left and right panels respectively depict the values of $s^{\rm own}$ and $s^{\rm opp}$, where the color scale is interpreted in the same manner as in Figure~\ref{fig_common_opt}.
We note that, in general, $s^{\rm own}$ and $s^{\rm opp}$ do not attain the maximal value of $1$ (even when $\overline{\pi}$ is optimal) in the private-value case, because $\theta_i$ and $\theta_j$ are not perfectly correlated, and thus $\sigma^\pi_i$ cannot be perfectly correlated with both of them simultaneously.


\subsection*{Proof of Proposition~\ref{prop_common_K}}

\eqref{prop_common_K1} follows directly from Corollary~\ref{cor_sym_bangbang}.
As in the proof of Proposition~\ref{prop_welfare_regular}, we let $R = \1\1^\top/n$ instead of $R = I$ for the remainder of this proof, without altering the underlying game structure.

\begin{lemma} \label{lem_K_full1}
The scalar matrix $\lambda I$ is dual feasible if $\lambda \ge \frac{v-c}{1+\beta}$ and $\lambda > \frac{v+(n-1)c}{1-(n-1)\beta}$.
Conversely, it is not dual feasible if $\lambda < \frac{v-c}{1+\beta}$ or $\lambda < \frac{v+(n-1)c}{1-(n-1)\beta}$.
\end{lemma}

\begin{proof}
With $\Lambda = \lambda I$, we can compute $A_\Lambda$ as
\[
A_\Lambda = \lambda (I-\beta K) - V
= \qty(\lambda (1+\beta) -v+c) I - \qty(\lambda \beta + c) \1\1^\top.
\]
Since the eigenvalues of $\1\1^\top$ are $0$ and $n$, the matrix $A_\Lambda$ has two eigenvalues:
\[
\alpha_1 = \lambda (1+\beta) -v+c \quad {\rm and} \quad
\alpha_2 = (1-(n-1)\beta)\lambda - (v+(n-1)c)
\]
Thus, we have $A_\Lambda \not\succeq O$ if $\alpha_1 < 0$ or $\alpha_2 < 0$, i.e., if $\lambda < \frac{v-c}{1+\beta}$ or $\lambda < \frac{v+(n-1)c}{1-(n-1)\beta}$.
In this case, $\Lambda$ is not dual feasible.

Now suppose $\lambda \ge \frac{v-c}{1+\beta}$ and $\lambda > \frac{v+(n-1)c}{1-(n-1)\beta}$. Then $A_\Lambda \succeq O$, and Lemma~\ref{lem_block} implies that $\Lambda$ is dual feasible if $B_\Lambda = A_\Lambda A_\Lambda^+ B_\Lambda$.
Since $B_\Lambda = \lambda \1\1^\top/n$ is constant, it suffices to verify that $\1$ lies in the range of $A_\Lambda$.
The $i$-th coordinate of $A_\Lambda \1$ is $\alpha_2$ for all $i$, which is strictly positive by assumption. Thus, taking $\bm{x} = \1/\alpha_2$ yields $A_\Lambda \bm{x} = \1$.
\end{proof}

Let us prove \eqref{prop_common_K2}.
Under $\overline{\pi}$, we have $x > 0$ and $\rho_x = 1$, so the designer's value is given as $V \bullet X = n\qty(v + (n-1)c) x$.
On the other hand, by Lemma~\ref{lem_value}, the designer's optimal value must be represented as $A_\Lambda \bullet X$.
In particular, since we can take $\Lambda = \lambda I$ by Proposition~\ref{prop_sym}, $A_\Lambda \bullet X = n \qty(\lambda - v - (n-1)(\lambda \beta + c)) x$.
Equating this with $V \bullet X$ pins down the value of $\lambda$ as
\begin{equation*}
\overline{\lambda} = \frac{2\qty(v+(n-1)c)}{1-(n-1)\beta}.
\end{equation*}
Hence, $\overline{\pi}$ is optimal if and only if $\overline{\lambda} I$ is dual feasible.
Observe that
\begin{equation} \label{pf_K_full1}
\overline{\lambda} \ge \frac{v+(n-1)c}{1-(n-1)\beta} \quad \iff \quad c \ge -\frac{v}{n-1}.
\end{equation}
Also, by simple calculations, we can see that
\begin{equation} \label{pf_K_full2}
\overline{\lambda} \ge \frac{v-c}{1+\beta} \quad \iff \quad c \ge - \qty(\frac{1+(n+1)\beta}{2n-1+(n-1)\beta}) \cdot v = f(\beta) v.
\end{equation}
Lemma~\ref{lem_K_full1} implies that $\overline{\pi}$ is not optimal if either inequality fails.
It also implies that if both are satisfied and \eqref{pf_K_full1} holds strictly, then $\overline{\pi}$ is optimal.
It remains to consider the case where \eqref{pf_K_full1} holds with equality and \eqref{pf_K_full2} weakly.
In this case, we have $v + (n-1)c = 0$ and $v \le 0$, so \eqref{prop_common_K1} implies that $\underline{\pi}$ is optimal with value zero.
Since $v + (n-1)c = 0$ also implies that the value of $\overline{\pi}$ is zero, both $\underline{\pi}$ and $\overline{\pi}$ are optimal.

Lastly, to prove \eqref{prop_common_K3}, suppose that neither $\underline{\pi}$ nor $\overline{\pi}$ is optimal.
By \eqref{prop_common_K1} and \eqref{prop_common_K2}, this case occurs when $c < v$ and $c < f(\beta)v$.

\begin{lemma} \label{lem_K_full2}
If $c < v$ and $c < f(\beta)v$, then $c < - \beta v$.
\end{lemma}

\begin{proof}
Solving the equation $(1-\alpha) + \alpha f(\beta) = -\beta$ for $\alpha$, we obtain
\[
\alpha = \frac{(1+\beta)(2n-1+(n-1)\beta)}{2(n+(n-1)\beta)}.
\]
If $\alpha \in (0,1)$, then $-\beta$ is represented as a convex combination of $1$ and $f(\beta)$, and hence $c < v$ and $c < f(\beta)v$ together imply $c < - \beta v$.
Let us show that $\alpha \in (0,1)$.
By $\beta > -1$, we can easily see that $\alpha > 0$.
To verify $\alpha < 1$, observe that
\[
2(n+(n-1)\beta) - (1+\beta)(2n-1+(n-1)\beta) = 1- n\beta - (n-1)\beta^2.
\]
The minimum value of the right-hand side is $1+\frac{n^2}{4(n-1)}$, attained at $\beta = - \frac{n}{2(n-1)}$.
\end{proof}

Consider a candidate profile $(x,y,\rho_x,\rho_y)$ that generates the symmetric primal solution $(X,Y)$, where $\rho_y = 1$.
By Theorem~\ref{thm_main}, at the optimum, these parameters must satisfy the conditions $A_\Lambda X = B_\Lambda Y$ and $A_\Lambda Y = B_\Lambda Z$ for some dual feasible scalar matrix $\Lambda = \lambda I$, where we recall that $A_\Lambda = \lambda (I-\beta K)-V$ and $B_\Lambda = \lambda\1\1^\top/2n$.

Since $\rho_y = 1$ implies $Y = y\1\1^\top$, expanding the condition $A_\Lambda X = B_\Lambda Y$ yields
\begin{equation} \label{partial_common_K1}
(\lambda - v) x - (n-1)(\lambda \beta + c) \rho_x x = \frac{\lambda y}{2} = - (\lambda \beta + c) x + (\lambda - v) \rho_x x - (n-2)(\lambda \beta + c) \rho_x x.
\end{equation}
Also, substituting $Y = y\1\1^\top$ and $Z = \1\1^\top$ into the condition $A_\Lambda Y = B_\Lambda Z$ yields
\begin{equation} \label{partial_common_K2}
(\lambda - v) y - (n-1)(\lambda \beta + c) y = \frac{\lambda}{2}.
\end{equation}
Arranging \eqref{partial_common_K1}, we get
\[
\underbrace{\qty(\lambda (1+\beta) -v+c)}_{(*)}(1-\rho_x) x = 0.
\]
Since $\underline{\pi}$ is suboptimal, we must have $x > 0$; and since $\overline{\pi}$ is suboptimal, we must have $\rho_x < 1$.
Together, these imply that $(*)$ must be zero, which pins down the value of $\lambda$ as
\begin{equation*}
\hat{\lambda} = \frac{v-c}{1+\beta}.
\end{equation*}
Lemma~\ref{lem_K_full1} implies that $\hat{\lambda} I$ is dual feasible if $\frac{v-c}{1+\beta} > \frac{v+(n-1)c}{1-(n-1)\beta}$.
A direct calculation shows that this inequality is equivalent to $c < - \beta v$, which, by Lemma~\ref{lem_K_full2}, holds true for the parameter region under consideration.

Given $\hat{\lambda}$, the remaining variables $(x,y,\rho_x)$ can now be derived systematically:
First, substituting $\hat{\lambda}$ into \eqref{partial_common_K2} yields the value of $y$.
Then, substituting this value of $y$ into the obedience condition \eqref{feasible_K} and \eqref{partial_common_K1}, we can solve for $x$ and $\rho_x$.
After carrying out the algebra, the resulting values coincide with those stated in the proposition.
Finally, the variance reduction can be computed using the formula $\Var[\theta \mid \sigma^\pi_i] = 1-y^2/x$.
\hfill {\it Q.E.D.}

\renewcommand{\thesection}{D}
\renewcommand{\theequation}{D\arabic{equation}}
\renewcommand{\thelemma}{D\arabic{lemma}}
\renewcommand{\thecorollary}{D\arabic{corollary}}
\renewcommand{\theproposition}{D\arabic{proposition}}

\renewcommand{\theHlemma}{D\arabic{lemma}}
\renewcommand{\theHproposition}{D\arabic{proposition}}
\renewcommand{\theHequation}{D\arabic{equation}}

\setcounter{section}{3}
\setcounter{equation}{0}
\setcounter{lemma}{0}
\setcounter{proposition}{0}

\section{Public LQG information design}
\label{app_pub}

In this Appendix~\ref{app_pub}, we study information design in a general environment $(Q,R,V,W,Z)$ but by restricting attention to the class of public information structures.
Formally, an information structure $\pi$ is \emph{public} if $\eta \equiv \eta_1 = \cdots = \eta_n$ occurs with probability one, where $\eta$ is called a \emph{public signal}.
Let $\Pi^{\rm pb}$ denote the subset of $\Pi$, consisting of all public information structures.
Then, the optimization problem of our interest is given as follows:
\begin{equation*}
v^{\rm pb} \equiv \max_{\pi\in \Pi^{\rm pb}} \E [v(\bm{\sigma}^\pi, \bm{\theta})],
\end{equation*}
Henceforth, we assume that $Z \succ O$.

As in the case of private information structures, our first task is to derive the SDP formulation of public LQG information design.
A crucial aspect of public information is that agents do not face any strategic uncertainty about the opponents' actions, which simplify the characterization of a BNE and its induced moments.
Specifically, under a public information structure, the first-order condition \eqref{FOC} is
\[
\sum_{j=1}^n q_{ij} \sigma^\pi_j(\eta) = \sum_{k=1}^m r_{ik} \E \qty[\theta_k \mid \eta], \quad \forall i \in N,
\]
involving no conditional expectations in the left-hand side.
From this system of equations, the unique BNE is obtained as ${\bm \sigma}^\pi(\eta) = Q^{-1}R \E \qty[\bm{\theta} \mid \eta]$.

Calculating the equilibrium action-state covariance from this closed-form expression for BNE, we can express $\mathbf{M}^{\pi}$ as a function of the state variance reduction, $\Var \qty[E \qty[\bm{\theta} \mid \eta]]$, measuring the amount of the state variance that can be reduced by observing the public signal $\eta$.
This matrix necessitates the statistical requirement of $Z \succeq \Var \qty[E \qty[\bm{\theta} \mid \eta]] \succeq O$.
Conversely, in the next lemma, we show that any symmetric matrix that satisfies these matrix inequalities can be induced as an equilibrium action-state covariance matrix under some public Gaussian information structure.
Moreover, it shows that the state variance reduction serves as a sufficient statistic for determining the designer's expected value of each public information structure.

\begin{lemma} \label{lem_moment_pub}
For any $\pi \in \Pi^{\rm pb}$, let $S=\Var[\E\qty[\bm{\theta} \mid \eta]]$.
Then, it holds that $Z \succeq S \succeq O$, and $\mathbf{M}^\pi$ coincides with
\begin{equation*}
\mathbf{M}_S \equiv \mqty[Q^{-1}RS(Q^{-1}R)^\top & Q^{-1}RS \\ S(Q^{-1}R)^\top & Z] \in \calS_+^{n+m}.
\end{equation*}
Conversely, for any $S \in \calS^m$ such that $Z \succeq S \succeq O$, there exists a public Gaussian information structure $\pi \in \Pi^{\rm pb} \cap \Pi^{\rm g}$ such that $\mathbf{M}^{\pi} = \mathbf{M}_S$.
Moreover, the designer's expected objective is expressed as
\[
\E \qty[v(\bm{\sigma}^\pi, \bm{\theta})] = \overline{C} \bullet S,
\]
where $\overline{C} \in \calS^m$ is the matrix defined in Example~\ref{ex_full}.
\end{lemma}

As similarly to the private information case, these lemmas deliver pivotal implications for simplifying the public LQG information design.
First, Lemma~\ref{lem_moment_pub} implies that the set of inducible covariance matrices under all public information structures coincides with that under all public Gaussian ones, since 
\[
\qty{\mathbf{M}^\pi: \pi\in \Pi^{\rm pb}}
= \qty{\mathbf{M}_S: Z \succeq S \succeq O}
= \qty{\mathbf{M}^\pi: \pi \in \Pi^{\rm pb} \cap \Pi^{\rm g}}.
\]
In addition, since the designer's objective depends on the choice of public information structures only through the state variance reduction, it is without loss of optimality to focus on public Gaussian information structures among all public ones.

Second, the public LQG information design problem can be restated as the maximization of the linear objective function $\overline{C} \bullet S$ subject to the two-sided positive semidefiniteness constraints of the form $Z \succeq S \succeq I$.
Moreover, by rescaling, this SDP problem turns out to be equivalent to
\begin{equation}
\max_{S \in \calS^m} \quad \widehat{C} \bullet S \quad \text{subject to} \quad I \succeq S \succeq O, \quad \text{where} \quad \widehat{C} \equiv Z^{\frac{1}{2}}\overline{C}Z^{\frac{1}{2}}. \label{pb_sdp}
\end{equation}
To see the equivalence, let $\widehat{C}$ be given as above, and let $\widehat{S} = Z^{-1/2}SZ^{-1/2}$.
Then, it should be easily confirmed that $Z \succeq S \succeq O$ if and only if $I \succeq \widehat{S} \succeq O$, and that $\overline{C} \bullet S = \widehat{C} \bullet \widehat{S}$, from which we can justify \eqref{pb_sdp} as the equivalent restatement of the original problem.

From \eqref{pb_sdp}, we readily see that if $\widehat{C}$ is negative semidefinite, then $\widehat{C} \bullet S \le 0$ holds for all feasible $S$, thereby no disclosure being optimal.
Conversely, if $\widehat{C}$ is positive semidefinite, then the objective is non-decreasing in $S$ with respect to the matrix ordering $\succeq$, whence it is optimal to choose $S=I$, which corresponds to full disclosure.
In particular, when $m=1$, the matrix $\widehat{C}$ reduces to a scalar variable, and its sign is unambiguously determined.
As a result, either of the two extreme choices, $\underline{\pi}$ and $\overline{\pi}$, must be optimal in $\Pi^{\rm pb}$.

In general, when the state is multi-dimensional, $\widehat{C}$ may be neither positive nor negative semidefinite so that partial disclosure is optimal.
Below, following the methodology developed by \cite{tamura2018}, we characterize the optimal public information structure and quantify the informativeness of the optimal public signal based on the eigendecomposition of $\widehat{C}$.

Allowing for geometric multiplicity, let $\gamma_1,\ldots,\gamma_m$ denote the $m$ real eigenvalues of $\widehat{C}$ and by $\bm{u}_1,\ldots,\bm{u}_m$ the corresponding eigenvectors, which  constitute an orthonormal system in $\R^m$.
We align the eigenvalues in a non-increasing order, and let $k^*$ be the last index such that the corresponding eigenvalue exceeds $0$, i.e.,
\begin{equation*}
\gamma_1 \ge \gamma_2 \ge \cdots \ge \gamma_{k^*} > 0 \ge \gamma_{k^*+1} \ge \cdots \ge \gamma_m,
\end{equation*}
where the convention $k^* = 0$ is adopted when $\widehat{C}$ is negative semidefinite.
By the spectral theorem, $\widehat{C}$ can be factorized as follows:
\[
\widehat{C}  = U \Diag(\gamma_1,\ldots,\gamma_m) U^\top,
\]
where $U \equiv [\bm{u}_1 \cdots \bm{u}_m]$ is an $m \times m$ orthogonal matrix, which has each $\bm{u}_k$ as the $k$-th column vector.
In addition, let $U_{k^*} \equiv [\bm{u}_1 \cdots \bm{u}_{k^*}] \in \calM^{m,k^*}$.

\begin{proposition}\label{prop_pub}
Let $\pi$ be a public information structure, which sends a multi-dimensional public signal $\bm{\eta}=U_{k^*}Z^{-1/2}\bm{\theta}$.
Then $\pi$ is optimal in $\Pi^{\rm pb}$, and the associated state variance reduction is given by
\begin{equation} \label{opt_pub_var}
\Var \qty[\E \qty[\bm{\theta} \mid \bm{\eta}]] = Z^{\frac{1}{2}} U_{k^*} U_{k^*}^\top Z^{\frac{1}{2}}.
\end{equation}
Moreover, $v^{\rm pb}$ is equal to the sum of strictly positive eigenvalues of $\widehat{C}$, i.e., $\gamma_1+\cdots+\gamma_{k^*}$.
Finally, $k^*$ satisfies $p(\mathbf{V}) - n \le k^* \le p(\mathbf{V})$, where $p(\mathbf{V})$ denotes the number of strictly positive eigenvalues of $\mathbf{V}$.
\end{proposition}

Proposition~\ref{prop_pub} characterizes the optimal form of the public signal via the spectrum of the matrix $\widehat{C}$ that is defined through primitives.
For example, if $\widehat{C}$ is positive definite---i.e., all its eigenvalues are positive---then full disclosure is optimal, as \eqref{opt_pub_var} reduces to $Z$ when $k^* = m$.
In contrast, if $\widehat{C}$ is negative semidefinite---i.e., all its eigenvalues are nonpositive---then no disclosure is optimal.
In general, $\widehat{C}$ is neither positive nor negative semidefinite.
Yet, the optimal public signal is identified as a multi-dimensional statistic, with its dimension equal to the number of positive eigenvalues of $\widehat{C}$.
Specifically, there is a monotonic relationship between the number of positive eigenvalues of $\widehat{C}$ and the informativeness of the optimal public signal, as the associated state variance reduction \eqref{opt_pub_var} increases in $k^*$ with respect to the matrix ordering.
In other words, as $\widehat{C}$ approaches a positive definite matrix, the optimal signal discloses more information about $\bm{\theta}$.

Implementing an optimal public information structure requires the designer to have detailed knowledge of the underlying game, since the key matrix $\widehat{C}$ depends on the entire payoff structure.
In practice, however, $Q$ and $R$ may be private information to the agents, as they are components of their payoff functions, and the designer may be agnostic about these parameters.
Still, in Proposition~\ref{prop_pub}, we derive relationships between the number of positive eigenvalues of $\widehat{C}$ and those of $\mathbf{V}$.
This enables the designer to at least partially identify the optimal degree of public information disclosure based solely on her objective function.

\citet{tamura2018} studies an information design problem with a single agent, whose best response is given as the best estimate about the state $\E[\bm{\theta} \mid \bm{\eta}]$ upon receiving signal $\bm{\eta}$.
By assuming a quadratic objective, he then shows that the designer's objective maximization reduces to the following SDP problem:
\[
\max_{S \in \calS^m} \quad W\bullet S \quad \text{subject to} \quad \Var \qty[\bm{\theta}] \succeq S \succeq O,
\]
where $W \in \calS^m$ is a constant matrix, and obtains a closed form solution to it.
He also demonstrates that this SDP formulation is useful in studying an LQG network game with $n$ agents and obtains an optimal public information structure in a special case of an LQG network game when $m=n$. 
In this regard, Proposition \ref{prop_pub} generalizes his result by allowing for an arbitrary $m$, while our proof relies heavily on his solution method.

To conclude, we address a natural question of what conditions are needed for the optimal public information structure to be ``globally'' optimal among all information structures, including non-public ones.

\begin{corollary} \label{cor_pb_sub}
The optimal public information structure in Proposition \ref{prop_pub} attains $v^{\rm p}$ if and only if there exists $\Lambda \in \calS^n_{\rm diag}$ with $A_\Lambda \succeq O$ such that
\begin{gather}
\qty(A_\Lambda Q^{-1}R - B_\Lambda) Z^\frac{1}{2} \bm{u}_k = \mathbf{0}, \quad \forall k=1,\ldots,k^*; \quad {\rm and} \label{opt_pb1} \\
B_\Lambda Z^\frac{1}{2} \bm{u}_k = \mathbf{0}, \quad \forall k=k^*+1,\ldots,m. \label{opt_pb2}
\end{gather}
\end{corollary}

This corollary stems from the observation that the optimal public information structure in Proposition~\ref{prop_pub} is noise-free by construction, allowing us to apply Corollary~\ref{cor_nfsi2}.
Incidentally, the optimality condition involves both matrix equations in Corollary \ref{cor_no_full}, while each equation only needs to be \emph{partially} satisfied on the linear subspace spanned by $\{Z^{1/2}\bm{u}_1,\ldots,Z^{1/2}\bm{u}_{k^*}\}$ and $\{Z^{1/2}\bm{u}_{k^*+1},\ldots,Z^{1/2}\bm{u}_{m}\}$, respectively.
Nevertheless, since \eqref{opt_pb1} and \eqref{opt_pb2} collectively impose $nm$ number of linear equations on $n$ Lagrangian multipliers, the optimal public information structure is not globally optimal for generic payoff structures.
This indicates that the designer typically benefits from manipulating each agent's higher-order belief about others' private information, aligning with the findings of \cite{arielibabichenko2019} and \cite{hoshino2022}.


\subsection*{Proof of Lemma \ref{lem_moment_pub}}

As discussed in the main text, a unique BNE is characterized as $\bm{\sigma}^\pi = Q^{-1}R\E \qty[\bm{\theta} \mid \eta]$ for any public information structure $\pi \in \Pi^{\rm pb}$, from which $\mathbf{M}^\pi = \mathbf{M}_{\Var \qty[\E \qty[\bm{\theta} \mid \eta]]}$ is confirmed.
Moreover, $Z \succeq \Var \qty[\E \qty[\bm{\theta} \mid \eta]] \succeq O$ holds by the law of total variance.

Conversely, take any $S \in \calS^m$ with $Z \succeq S \succeq O$.
Let $\pi \in \Pi^{\rm pb} \cap \Pi^{\rm g}$ be a public Gaussian information structure, sending an $m$-dimensional public signal $\bm{\eta}$, such that $\bm{\eta}$ and $\bm{\theta}$ are jointly normally distributed with $\Var \qty[\bm{\eta}] = \Cov\qty[\bm{\eta}, \bm{\theta}] = S$.
Such a random vector exists because
\[
\Var \mqty[S & S \\ S & Z] \succeq O,
\]
which is confirmed by $Z - SS^+S = Z-S \succeq O$.
Moreover, by Lemma~\ref{lem_cond_normal}, the unique BNE under $\pi$ is computed as
\[
\bm{\sigma}^\pi \qty(\bm{\eta}) = Q^{-1}R \E \qty[\bm{\theta} \mid \bm{\eta}]
= Q^{-1}R SS^+ \bm{\eta} - \underbrace{Q^{-1}R \qty(\bar{\bm{\theta}} + SS^+ \E \qty[\bm{\eta}])}_{\rm constant}.
\]
Hence, it follows that
\begin{gather*}
\Var \qty[\bm{\sigma}^\pi(\bm{\eta})] = Q^{-1}RSS^+ \Var \qty[\bm{\eta}] S^+S\qty(Q^{-1}R)^\top = Q^{-1}RS\qty(Q^{-1}R)^\top, \\
\Cov \qty[\bm{\sigma}^\pi(\bm{\eta}), \bm{\theta}] = Q^{-1}RSS^+ \Cov\qty[\bm{\eta}, \bm{\theta}] = Q^{-1}RSS^+.
\end{gather*}
Thus we have $\mathbf{M}^\pi = \mathbf{M}_S$.

Now, substituting $\mathbf{M}^\pi = \mathbf{M}_S$ into the expected objective function in Lemma~\ref{lem_obj}, we have
\begin{align*}
\E \qty[v(\bm{\sigma}^\pi, \bm{\theta})] &= \mathbf{V} \bullet \mathbf{M}_S \\
&= V \bullet Q^{-1}RS(Q^{-1}R)^\top + \frac{1}{2} W \bullet Q^{-1}RS + \frac{1}{2} W^\top \bullet S(Q^{-1}R)^\top \\
&= \tr \qty(VQ^{-1}RS(Q^{-1}R)^\top) + \frac{1}{2} \tr \qty(W^\top Q^{-1}RS) + \frac{1}{2} \tr \qty(WS(Q^{-1}R)^\top) \\
&= \tr \qty( \qty((Q^{-1}R)^\top VQ^{-1}R + \frac{W^\top Q^{-1}R + (Q^{-1}R)^\top W}{2}) S) \\
&= \overline{C} \bullet S,
\end{align*}
which concludes the proof of Lemma~\ref{lem_moment_pub}.
\hfill {\it Q.E.D.}


\subsection*{Proof of Proposition \ref{prop_pub}}

Consider the eigendecomposition $\widehat{C} = UDU^\top$, where $D = \Diag(\gamma_1,\ldots,\gamma_m)$.
We first show that the sum of positive eigenvalues of $\widehat{C}$ serves as an upper bound on the value of any public information strcutrue.

\begin{lemma} \label{lem_pub_bound}
For any matrix $S \in \calS^m$ with $I \succeq S \succeq O$, all diagonal entries of $U^\top SU$ are confined in $[0,1]$.
Consequently,
\[
\widehat{C} \bullet S \le \sum_{k=1}^{k^*} \gamma_k.
\]
\end{lemma}

\begin{proof}
Denote by $s_1,\ldots,s_m$ the diagonal entries of $U^\top SU$.
Since $S \succeq O$, we have $U^\top SU \succeq O$, thereby all of $s_1,\ldots,s_m$ being nonnegative.
Moreover, since $I \succeq S$, and since $U$ is orthogonal, $I - U^\top SU = U^\top (I-S) U \succeq O$.
This implies that all diagonal entries of $I-U^\top SU$ are nonnegative, and thus all $s_1,\ldots,s_m$ are weakly less than $1$.
Hence, it follows that
\[
\widehat{C} \bullet S
= \tr \qty(UDU^\top S)
= \tr \qty(DU^\top SU)
= D \bullet U^\top SU
= \sum_{k=1}^m \gamma_k s_k
\le \sum_{k=1}^{k^*} \gamma_k,
\]
where the inequality holds since $\gamma_k > 0$ if and only if $k \le k^*$ and $s_k \in [0,1]$ for all $k$.
\end{proof}

We show that the designer can achieve the upper bound derived in Lemma~\ref{lem_pub_bound} by sending the public signal $\bm{\eta} = U_{k^*}^\top Z^{-1/2} \bm{\theta}$.
Notice that $\Var \qty[\bm{\eta}] = U_{k^*}^\top U_{k^*}$ and $\Cov \qty[\bm{\eta},\bm{\theta}] = U_{k^*}^\top Z^{1/2}$ hold.
Then, by Lemma~\ref{lem_cond_normal}, we have
\[
\Var \qty[\E \qty[\bm{\theta} \mid \bm{\eta}]]
= Z^{\frac{1}{2}}U_{k^*} \qty(U_{k^*}^\top U_{k^*})^{-1} U_{k^*}^\top Z^{\frac{1}{2}}
= Z^{\frac{1}{2}}U_{k^*}U_{k^*}^\top Z^{\frac{1}{2}},
\]
where the second equality holds since the column vectors of $U_{k^*}$ constitutes an orthonormal system, thereby $U_{k^*}^\top U_{k^*}$ being an identify matrix.
Hence, by Lemma \ref{lem_moment_pub}, the designer's expected value can be calculated as follows:
\begin{align*}
\overline{C} \bullet \Var\qty[\E\qty[\bm{\theta} \mid \bm{\eta}]]
&= \tr \bigl(\overline{C} Z^{\frac{1}{2}}U_{k^*}U_{k^*}^\top Z^{\frac{1}{2}} \bigr)
= \tr \bigl(Z^{\frac{1}{2}} \overline{C} Z^{\frac{1}{2}}U_{k^*} U_{k^*}^\top \bigr) \\
= \tr \bigl(\widehat{C}U_{k^*}U_{k^*}^\top\bigr)
&= \tr \bigl(UDU^\top U_{k^*}U_{k^*}^\top \bigr)
= \tr \bigl(DU^\top U_{k^*}U_{k^*}^\top U \bigr)
= \sum_{k=1}^{k^*} \gamma_k.
\end{align*}
Note that the last equality holds because $U^\top U_{k^*}=[\delta_{ij}]_{m\times k^*}$, where $\delta_{ij}$ is the Kronecker delta.

Let us show that $p(\mathbf{V})-n \le p(\widehat{C}) \le p(\mathbf{V})$.
Since $\widehat{C}$ and $\overline{C}$ are congruent, it follows that $p(\widehat{C}) = p(\overline{C})$ by Sylvester's law of inertia.
Moreover, we can express $\overline{C}$ as follows:
\[
\overline{C} = \mathbf{Q}^\top \mathbf{V} \mathbf{Q}, \quad {\rm where} \quad \mathbf{Q} \equiv \mqty[Q^{-1}R \\ I] \in \calM^{n+m,m}.
\]
Note that $\rank\qty(\mathbf{Q}) = m$ so that $\ker\qty(\mathbf{Q}) = \{\mathbf{0}\}$.

Now, suppose that the $k$-th largest eigenvalue of $\overline{C}$ is strictly positive.
By the Courant--Fischer theorem, there exist $\epsilon > 0$ and a $k$-dimensional linear subspace $\calU \subseteq \R^m$ such that
\[
\min_{\bm{u} \in \calU \setminus \{\mathbf{0}\}} \frac{\bm{u}^\top \overline{C}\bm{u}}{\bm{u}^\top \bm{u}} \ge \epsilon.
\]
Since $\ker\qty(\mathbf{Q}) = \{\mathbf{0}\}$, it follows that
\[
\frac{\qty(\mathbf{Q}\bm{u})^\top \mathbf{V} \qty(\mathbf{Q}\bm{u})}{\qty(\mathbf{Q}\bm{u})^\top \qty(\mathbf{Q}\bm{u})}
\cdot \frac{\qty(\mathbf{Q}\bm{u})^\top \qty(\mathbf{Q}\bm{u})}{\bm{u}^\top \bm{u}} \ge \epsilon, \quad \forall \bm{u} \in \calU \setminus \{\mathbf{0}\},
\]
from which we have
\[
\min_{\bm{u} \in \calU\setminus \{\mathbf{0}\}}\frac{\qty(\mathbf{Q}\bm{u})^\top \mathbf{V} \qty(\mathbf{Q}\bm{u})}{\qty(\mathbf{Q}\bm{u})^\top \qty(\mathbf{Q}\bm{u})}
\ge \frac{\epsilon}{\delta}, \quad {\rm where} \quad \delta = \max\qty{\frac{\qty(\mathbf{Q}\bm{u})^\top \mathbf{Q}\bm{u}}{\bm{u}^\top \bm{u}}: \bm{u} \in \R^m \setminus \{\mathbf{0}\}}.
\]
Note that $\delta > 0$ since $\mathbf{Q}^\top \mathbf{Q}$ is positive definite.
Then, since $\{\mathbf{Q}\bm{u} \in \R^{n+m}: \bm{u} \in \calU\}$ is a $k$-dimensional linear subspace of $\R^{n+m}$, again by the Courant--Fischer theorem, it follows that the $k$-th largest eigenvalue of $\mathbf{V}$ is strictly positive.
Thus $p(\overline{C}) \le p(\mathbf{V})$.

Let $p'(\cdot)$ denote the number of non-negative eigenvalues of a given square matrix.
Then, we can show that $p'(\overline{C}) \le p'(\mathbf{V})$ by slightly changing the preceding proof.
Moreover, applying the same result to $-\overline{C}$ and $-\mathbf{V}$, it follows that $p'(-\overline{C}) \le p'(-\mathbf{V})$.
Then, noticing that $p'(-\overline{C}) = m-p(\overline{C})$ and $p'(-\mathbf{V}) = (n+m)-p(\mathbf{V})$, we obtain
\[
m-p(\overline{C}) \le (n+m)-p(\mathbf{V}),
\]
from which $p(\mathbf{V})-n \le p(\overline{C})$.
\hfill {\it Q.E.D.}


\subsection*{Proof of Corollary \ref{cor_pb_sub}}

Let $\pi$ be the optimal public information structure as in Proposition~\ref{prop_pub}.
By construction, $\pi$ is noise-free, and
\[
\Cov\qty[\bm{\sigma}^\pi, \bm{\theta}] = Q^{-1}RZ^{\frac{1}{2}} U_{k^*} U_{k^*}^\top Z^{\frac{1}{2}}.
\]
Hence, by Corollary~\ref{cor_nfsi2} and Lemma~\ref{lem_Zfull}, $\pi$ is optimal in $\Pi$ if and only if there exists $\Lambda \in \calS^n_{\rm diag}$ with $A_\Lambda \succeq O$ such that
\[
A_\Lambda Q^{-1}RZ^{\frac{1}{2}} U_{k^*} U_{k^*}^\top Z^{\frac{1}{2}} = B_\Lambda Z.
\]
By construction, $U_{k^*}U_{k^*}^\top = \sum_{k=1}^{k^*} \bm{u}_k \bm{u}_k^\top$ and $I=UU^\top = \sum_{k=1}^{m} \bm{u}_k \bm{u}_k^\top$.
Hence, the above condition can be equivalently rewritten as follows:
\begin{align}
&A_\Lambda Q^{-1}RZ^{\frac{1}{2}} U_{k^*} U_{k^*}^\top = B_\Lambda Z^{\frac{1}{2}} \notag \\
&\quad \iff \quad A_\Lambda Q^{-1}RZ^{\frac{1}{2}} \qty(\sum_{k=1}^{k^*} \bm{u}_k \bm{u}_k^\top) = B_\Lambda Z^{\frac{1}{2}} \qty(\sum_{k=1}^{m} \bm{u}_k \bm{u}_k^\top) \notag \\
&\quad \iff \quad \qty(A_\Lambda Q^{-1}RZ^{\frac{1}{2}} - B_\Lambda) Z^{\frac{1}{2}} \sum_{k=1}^{k^*} \bm{u}_k \bm{u}_k^\top - B_\Lambda Z^{\frac{1}{2}} \sum_{k=k^*+1}^{m} \bm{u}_k \bm{u}_k^\top = O. \label{opt_pb0}
\end{align}
It is evident that \eqref{opt_pb0} holds if both \eqref{opt_pb1} and \eqref{opt_pb2} are satisfied.

Conversely, suppose that \eqref{opt_pb0} is satisfied.
Then, since $\{\bm{u}_1,\ldots,\bm{u}_m\}$ constitutes an orthonormal basis in $\R^m$, multiplying both sides of \eqref{opt_pb0} by $\bm{u}_k$ from right, we obtain \eqref{opt_pb1} when $k \le k^*$, and \eqref{opt_pb2} when $k > k^*$.
\hfill {\it Q.E.D.}

\end{document}